\documentclass[11pt,letterpaper]{article}

\usepackage{mathtools}
\usepackage{sty}
\usepackage{color}
\usepackage{hyperref}

\newcommand{\R}{\mathbb{R}}

\newcommand{\pseudo}{\widetilde{\mathbb{E}}}

\newcommand{\inj}{\hookrightarrow}
\newcommand{\nb}[2]{A^{(#1)}_{#2}}
\newcommand{\sa}[2]{A^{\langle #1 \rangle}_{#2}}
\newcommand{\lam}{2\sqrt{d-1}}

\newcommand{\mkm}{\mu_{\textsc{km}}}
\newcommand{\km}{\textsc{km}}
\newcommand{\cc}{\textsf{cc}}
\newcommand{\tr}{\mathrm{Tr}}
\newcommand{\eps}{\varepsilon}
\newcommand{\cov}{\mathrm{Cov}}
\newcommand{\F}{\mathrm{F}}

\newcommand{\MC}{\mathrm{MC}}

\newcommand{\id}{I}
\newcommand{\onesmat}{J}
\newcommand{\onesvec}{\One}
\newcommand{\indicator}[1]{\bm{1}\left[#1\right]}

\title{Spectral Planting and the Hardness of Refuting Cuts, Colorability, and Communities in Random Graphs}

\usepackage{authblk}
\author[1]{Afonso S.\ Bandeira\thanks{Email: \textit{bandeira@math.ethz.ch}. Some of this work was done while with the Department of Mathematics at the Courant Institute of Mathematical Sciences, and the Center for Data Science, at New York University; and partially supported by NSF grants DMS-1712730 and DMS-1719545, and by a grant from the Sloan Foundation.}}
\author[2]{Jess Banks\thanks{Email: \textit{jess.m.banks@berkeley.edu}. Supported by the NSF Graduate Research Fellowship Program under Grant DGE-1752814.}}
\author[3]{Dmitriy Kunisky\thanks{Email: \textit{kunisky@cims.nyu.edu}. Partially supported by NSF grants DMS-1712730 and DMS-1719545.}}
\author[4]{Cristopher Moore\thanks{Email: \textit{moore@santafe.edu}. Partially supported by NSF grant IIS-1838251.}}
\author[3]{Alexander S.\ Wein\thanks{Email: \textit{awein@cims.nyu.edu}. Partially supported by NSF grant DMS-1712730 and by the Simons Collaboration on Algorithms and Geometry.}}

\affil[1]{Department of Mathematics, ETH Zurich}
\affil[2]{Department of Mathematics, UC Berkeley}
\affil[3]{Department of Mathematics, Courant Institute of Mathematical Sciences, NYU}
\affil[4]{Santa Fe Institute}

\begin{document}

\maketitle

\thispagestyle{empty}

\begin{abstract}
We study the problem of efficiently \emph{refuting} the $k$-colorability of a graph, or equivalently \emph{certifying} a lower bound on its chromatic number. We give formal evidence of average-case computational hardness for this problem in sparse random regular graphs, showing optimality of a simple spectral certificate. This evidence takes the form of a \emph{computationally-quiet planting}: we construct a distribution of $d$-regular graphs that has significantly smaller chromatic number than a typical regular graph drawn uniformly at random, while providing evidence that these two distributions are indistinguishable by a large class of algorithms. We generalize our results to the more general problem of certifying an upper bound on the maximum $k$-cut. 

This quiet planting is achieved by minimizing the effect of the planted structure (e.g.\ colorings or cuts) on the graph spectrum. Specifically, the planted structure corresponds exactly to eigenvectors of the adjacency matrix. This avoids the pushout effect of random matrix theory, and delays the point at which the planting becomes visible in the spectrum or local statistics. To illustrate this further, we give similar results for a Gaussian analogue of this problem: a quiet version of the spiked model, where we plant an eigenspace rather than adding a generic low-rank perturbation.

Our evidence for computational hardness of distinguishing two distributions is based on three different heuristics: stability of belief propagation, the local statistics hierarchy, and the low-degree likelihood ratio. Of independent interest, our results include general-purpose bounds on the low-degree likelihood ratio for multi-spiked matrix models, and an improved low-degree analysis of the stochastic block model. 
\end{abstract}

\newpage

\section{Introduction}

Assuming the widely believed $\texttt{P}\neq \texttt{NP}$ hypothesis, many combinatorial problems in graphs are known to be computationally hard. Prominent examples from graph theory and network science include finding large cliques or independent sets, clustering or maximizing cuts, and finding vertex colorings or computing chromatic numbers. Fortunately, the worst-case computational difficulty of many of these problems appears to not be predictive of their feasibility in typical graphs, motivating the study of forms of \emph{average-case complexity} for many of these problems. Many remarkable examples exist dating back to at least the work of Karp and others in the mid 70s~\cite{Karp76},\footnote{See also Karp's lecture on the occasion of his Turing Award~\cite{KarpTuring}.} they include the problem of vertex colorings~\cite{Grimmett_McDiarmid-75} in random graphs, the related problem of finding the largest independent set\footnote{Recall that the set of nodes of the same color in a vertex coloring is an independent set.} (or, equivalently, clique) and many others~\cite{Karp76}. In both of these problems, for certain natural distributions of random graphs, a multiplicative gap of 2 is identified between the typical optimal solution and the solution found by the best known polynomial-time algorithm~\cite{KarpTuring}, and improving over this has since been a standing open problem. Motivated by this question, Kucera~\cite{Kucera_95} and Alon et al.~\cite{Alon_HiddenClique_98} studied random graph models with planted structures (either a large independent set or clique, or a coloring with an unusually small number of colors) and investigate when such structures are easy to detect. Foreshadowing to what follows, we point out that the existence of a planted structure that cannot be detected efficiently implies that that it is impossible to efficiently refute the existence of such a structure in the underlying unplanted model. In this paper we will focus on the problem of computing the chromatic number, and the related problem of understanding the size of the largest $k$-cut. The random graph model throughout is the uniform distribution over $d$-regular graphs on $n$ nodes.

\paragraph{Refuting colorability.}
For an integer $k \ge 1$, a graph $G = (V,E)$ is \emph{$k$-colorable} if there exists an assignment $\sigma \colon V \to [k]$ of ``colors'' to the vertices such that $\sigma(i) \ne \sigma(j)$ for every edge $(i,j) \in E$. The \emph{chromatic number} $\chi(G)$ of $G$ is defined as the minimum value of $k$ for which $G$ is $k$-colorable.

A random $d$-regular graph $\bG$ on $n$ vertices (we will write random variables in bold-face font throughout the paper) is known to have a typical chromatic number $\chi(\bG) \sim \frac{1}{2}\frac{d}{\log d}$ (see, e.g., \cite{COEH-2016-ChromaticNumberRandomRegular} and references therein) in the double limit $n \to \infty$ followed by $d \to \infty$, where $f \sim g$ denotes $f/g \to 1$.
The problem we will study is that of algorithmically \emph{refuting} the $k$-colorability of a graph, which we define in Section~\ref{sec:prelim}.
Informally speaking, an algorithm that refutes $k$-colorability provides an efficiently-verifiable \emph{proof} that $G$ is not $k$-colorable.
As a simple example, an algorithm exhibiting a $(k + 1)$-clique refutes $k$-colorability.
More generally, one may encode a $k$-coloring as a collection of boolean variables satisfying certain logical relations depending on $G$, and refute coloring by deriving a contradiction from those axioms.

We will provide evidence that the refutation problem is computationally hard when $\bG$ is a uniformly random $d$-regular graph. The proof strategy is to construct a different distribution over $d$-regular graphs whose typical chromatic number is $\chi(\bG)\sim \frac{1}{2}\sqrt{d}$, and to argue that this distribution is computationally hard to distinguish from a uniformly random $d$-regular graph (whose chromatic number is instead $\chi(\bG) \sim \frac12\frac{d}{\log d}$)---we think of this new distribution as having a \emph{computationally-quiet planting}\footnote{This terminology is inspired by the notion of \emph{quiet planting} from prior work~\cite{quiet-1,quiet-2}, although our notion is somewhat different.} of a coloring with few colors. We will see below that the value $\frac{1}{2}\sqrt{d}$ coincides with a simple spectral bound on $\chi(\bG)$ for $\bG$ uniformly random, and so our result is a tight computational lower bound on refutation algorithms.
The formal evidence for computational hardness used in this paper is threefold: we provide consistent pieces of evidence based on (i) the Kesten--Stigum threshold for the belief propagation message-passing algorithm, (ii) the Local Statistics semidefinite programming hierarchy, and (iii) the low-degree likelihood ratio for an analogous Gaussian model.

\paragraph{The spectral refutation.}
Instead of just refuting $k$-colorability, we will consider the more general task of refuting the existence of large \emph{$k$-cuts} in a graph.
We define the fractional size of the largest such cut as
\begin{equation}
    \MC_k(G) \colonequals \max_{\sigma : V \to [k]} \frac{|\{(u, v) \in E: \sigma(u) \neq \sigma(v)\}|}{|E|} \in (0, 1].
    \label{eq:MC-def}
\end{equation}
Intuitively, $\MC_k(G)$ describes how close $G$ is to being $k$-colorable, as the cut counts the fraction of polychromatic edges under the coloring $\sigma$.
$\MC_k(G)$ is non-decreasing in $k$, and for any $k$, $G$ is $k$-colorable if and only if $\MC_k(G) = 1$; therefore, the chromatic number is given by
\begin{equation}
    \label{eq:chi-MC-relations}
    \chi(G) = \min\{k: \MC_k(G) = 1\} = 1 + \max\{k: \MC_k(G) < 1\}.
\end{equation}

\noindent Accordingly, upper bounds on $\MC_k$ away from the maximum value of 1 yield lower bounds on the chromatic number.
This task is often called \emph{certifying} a bound on the optimization problem $\MC_k$.
The relations \eqref{eq:chi-MC-relations} show how certifying such a bound in turn refutes colorability. In fact, Hoffman's early work \cite{Hoffman-1970-Eigenvalues} proposed a technique for any $d$-regular graph $G$ that gives, when rephrased in our notation, a refutation of colorability by bounding $\MC_k$ via the minimum eigenvalue $\lambda_{\min}(A_{G})$ of the adjacency matrix $A_{G}$ of $G$.
Namely, Hoffman showed that for any $d$-regular graph $G$,
\begin{equation}
    \MC_k(G) \leq \frac{k - 1}{k}\left(1 + \frac{-\lambda_{\min}(A_{G})}{d}\right).
    \label{eq:MC-spectral-bound}
\end{equation}
(Note that $\lambda_{\min}(A_{G}) \le 0$ since $\Tr(A_{G}) = 0$.) A short proof of~\eqref{eq:MC-spectral-bound} will be given in Section~\ref{sec:prelim}.
We note that $\frac{k - 1}{k}$ is the expected value of the objective of \eqref{eq:MC-def} when $\bm \sigma \colon V \to [k]$ is chosen uniformly at random, so expressions like the right-hand side of \eqref{eq:MC-spectral-bound} should be viewed as expressing a factor of ``gain'' over this value. For $\bG$ a uniformly \emph{random} $d$-regular graph on $n$ vertices, whose law we denote $\mathcal{G}_{n, d}$, a theorem due to Friedman
\cite{Friedman-2003-SecondEigenvalue} states that $\lambda_{\min}(\bG) = -2\sqrt{d-1} + o(1)$ with high probability\footnote{Here and throughout, $o(1)$ pertains to the limit $n \to \infty$ with $d$ fixed. We say that an event occurs \emph{with high probability (w.h.p.)}\ if it has probability $1-o(1)$.} for any fixed $d \ge 3$.
This implies that when $\bG \sim \sG_{n, d}$, Hoffman's spectral approach with high probability certifies the upper bound
\begin{equation}\label{eq:MC-spectral-bound-value}
    \MC_k(\bG) \leq \frac{k - 1}{k}\left(1 + \frac{2\sqrt{d - 1}}{d} + o(1)\right).
\end{equation}
This in turn translates to a lower bound on the chromatic number of $\chi(\bG) \ge (1 - o_d(1)) \frac{1}{2}\sqrt{d}$, where $o_d$ pertains to the double limit $n \to \infty$ followed by $d \to \infty$ (see Section~\ref{sec:notation}).

Equipped with this direct analysis of a simple technique, the natural question arises: can any polynomial-time algorithm produce a bound that, like Hoffman's, is valid for any graph $G$, but is typically \emph{tighter} for $\bG \sim \sG_{n, d}$? 
Most prior work on this question has studied bounds provided by a semidefinite program computing the \emph{Lov\'{a}sz $\vartheta$ function}, which, as shown by~\cite{BKM-2017-LovaszThetaRandomGraphs}, is equivalent to the degree-2 sum-of-squares relaxation of $\MC_k(G)$.
The work \cite{CO-2003-LovaszThetaRandomGraphs} and a later, more precise analysis by \cite{BKM-2017-LovaszThetaRandomGraphs} showed that, with high probability when $\bG \sim \mathcal{G}_{n, d}$, this relaxation certifies a bound still no better than
\begin{equation}
    \MC_k(\bG) \leq \frac{k - 1}{k}\left(1 + \frac{2\sqrt{d - 1}}{d + 2\sqrt{d - 1}} - o(1)\right) = \frac{k - 1}{k}\left(1 + \frac{2\sqrt{d - 1}}{d} - O_d\left(\frac{1}{d}\right)\right).
    \label{eq:MC-bound-sos2}
\end{equation}
Note that~\eqref{eq:MC-bound-sos2} matches the spectral bound~\eqref{eq:MC-spectral-bound-value} in the double limit $n \to \infty$ followed by $d \to \infty$. No polynomial-time certifier is known to asymptotically improve upon this bound. It appears plausible, then, that the spectral bound \eqref{eq:MC-spectral-bound-value} is an optimal efficiently-computable certificate on $\MC_k(\bG)$ for $\bG \sim \sG_{n, d}$.
In this work, we will argue that this is indeed the case.

\paragraph{Detecting a planted $k$-cut.}

The following line of reasoning will be central to this work: to prove computational hardness of a refutation problem, it is sufficient to construct a \emph{computationally-quiet planted distribution} (e.g.\ \cite{BKW-2019-ConstrainedPCA}). To illustrate the meaning of this, suppose our goal is to show computational hardness of refuting (w.h.p.)\ $k$-colorability of a graph drawn from $\sG_{n,d}$. Suppose we are able to construct a \emph{planted} distribution $\mathcal{P}$ over $d$-regular graphs such that (i) a typical graph drawn from $\mathcal{P}$ is $k$-colorable, and (ii) $\mathcal{P}$ is \emph{computationally quiet} in the sense that no polynomial-time algorithm can distinguish (w.h.p.)\ between a sample from $\mathcal{P}$ and a sample from $\sG_{n,d}$. It then follows that no polynomial-time algorithm can refute $k$-colorability in $\sG_{n,d}$, because if such a refutation algorithm were to exist, it must succeed w.h.p.\ on $\sG_{n,d}$ and must fail w.h.p.\ on $\mathcal{P}$, thus providing a solution to the distinguishing problem. More generally, to show hardness of certifying an upper bound on the maximum $k$-cut, we need a planted distribution for which there exists a large $k$-cut.

As discussed in \cite{BKM-2017-LovaszThetaRandomGraphs}, a natural planted distribution with a large $k$-cut is a $d$-regular variant of the popular \textit{stochastic block model (SBM)}. To sample a graph from this $d$-regular distribution, which we will denote $\mathcal{G}^{\mathrm{sbm}}_{n, d, k, \eta}$, first sample a balanced labelling $\bsigma : [n] \to [k]$ uniformly at random (\emph{balanced} means that $|\bsigma^{-1}(i)| = n/k$ for every $i \in [k]$), and then choose $\bG$ uniformly among $d$-regular graphs conditional on the event\footnote{For now, we will not consider the integrality conditions on the parameters $(n, d, k, \eta)$ that this implies, although these types of considerations will be important later.}
\begin{equation}
    \frac{|\{(u,v) \in E: \bsigma(u) \neq \bsigma(v)\}|}{|E|} = \frac{k-1}{k}(1 - \eta).
\end{equation}
Note that this ensures $\MC_k(\bG) \geq \tfrac{k-1}{k}(1 - \eta)$, since the planted partition $\bsigma$ witnesses a $k$-cut with that fraction of bichromatic edges. We will mostly be concerned with the \textit{disassortative} regime of this model where $\eta \in [-\tfrac{1}{k-1},0]$, and the planted $k$-cut $\bsigma$ is larger than a typical one; for instance, when $\eta = -\tfrac{1}{k-1}$ and the planted cut includes every edge, we have the well-studied planted coloring model.

The relevance of this distribution to certifying bounds on $\MC_k$ for $\bG \sim \sG_{n, d}$ is as follows: if an algorithm can with high probability over $\sG_{n,d}$ certify that $\MC_k(\bG) < \tfrac{k-1}{k}(1 - \eta)$, then (as discussed above) it is simple to build another \emph{testing} (or \emph{detection}) algorithm that distinguishes between $\bG \sim \sG_{n, d}$ and $\bG \sim \sG^{\mathrm{sbm}}_{n, d, k, \eta}$ with high probability. The advantage of taking this point of view is that there is a rich literature, originating in heuristic methods from statistical physics, that has provided a great deal of evidence that polynomial-time testing between $\sG_{n, d}$ and $\sG^{\mathrm{sbm}}_{n, d, k, \eta}$ is impossible below the \emph{Kesten--Stigum threshold}, i.e., when
\begin{equation}
    d < d_{\mathrm{KS}}^{\,\mathrm{sbm}} = d_{\mathrm{KS}}^{\,\mathrm{sbm}}(\eta) \colonequals \frac{1}{\eta^2} + 1.
    \label{eq:KS-sbm}
\end{equation}
We will refer to this claim as the \emph{SBM conjecture}. Polynomial-time algorithms are known to succeed when $d > d_{\mathrm{KS}}^{\,\mathrm{sbm}}$ \cite{massoulie,mns,AS-acyclic,banks2019local}. While proving such a conjecture seems to be beyond the reach of current techniques (even under an assumption such as $\texttt{P}\neq \texttt{NP}$), various forms of concrete evidence have been given (either for the $d$-regular SBM and other variants). These include results on stability of belief propagation~\cite{sbm-ks-1,sbm-ks-2}, the local statistics hierarchy~\cite{banks2019local}, and the low-degree likelihood ratio~\cite{HS-bayesian,sam-thesis}. We will discuss all of these methods further in Section~\ref{sec:heuristics}.

Rearranging \eqref{eq:KS-sbm} we find that, conditional on the SBM conjecture, the above argument implies that when $\bG \sim \sG_{n, d}$, no polynomial-time algorithm can certify with high probability a bound stronger than
\begin{equation}
    \MC_k(G) \leq \frac{k - 1}{k}\left(1 + \frac{1}{\sqrt{d - 1}} - o(1)\right).
    \label{eq:MC-bound-KS}
\end{equation}

\noindent Comparing this to \eqref{eq:MC-spectral-bound-value} and \eqref{eq:MC-bound-sos2}, we see a discrepancy between the best known certification algorithms and the above hardness result. For large $d$, this discrepancy amounts to a factor of $2$ in the ``gain'' term. This begs the question of whether better certification algorithms exist, or whether the hardness result can be improved. We will see below that it is the latter.

\paragraph{A quieter planting.}

Our main contribution is to show an improved hardness result by using a ``better'' planted distribution. The superior planted distribution is the following more rigid version of the SBM.

\begin{definition}[Equitable stochastic block model]
\label{def:eSBM}
    The \emph{equitable stochastic block model (eSBM)}, denoted $\sG_{n, d, k, \eta}^{\,\mathrm{eq}}$, is the probability distribution over $d$-regular graphs on $n$ vertices sampled as follows: first, choose a uniformly random balanced partition $\bsigma: V \to [k]$. Then, letting
    $$
        M = \eta \id_k + \tfrac{1 - \eta}{k}\onesmat_k
    $$
    (where $\id_k$  and $\onesmat_k$ are the $k\times k$ identity and all-ones matrices, respectively), for each $i \in [k]$ place a random $dM_{i,i}$-regular graph on the color class $\bsigma^{-1}(i)$, and for each $i<j \in [k]$ place a random bipartite $dM_{i,j}$-regular graph between $\bsigma^{-1}(i)$ and $\bsigma^{-1}(j)$. This model is only defined when $k|n$, $dM$ is a nonnegative integer matrix, and $dM_{i,i}n/k$ is even\footnote{This last condition ensures that it is possible to place a $dM_{i,i}$-regular graph on $n/k$ vertices.} for all $i$.
\end{definition}

\noindent As in $\sG^{\text{sbm}}_{n,d,k,\eta}$, the planted cut $\bsigma$ has fractional size $\tfrac{k-1}{k}(1 - \eta)$, and we will again restrict to the disassortative case $\eta \in [-\frac{1}{k-1},0]$. This model is discussed in~\cite{BDGHT-2016-RecoveryRigidityRegularSBM}, where it is alternatively called the ``regular block model.'' We instead follow~\cite{newman-martin,barucca} in using the term ``equitable,'' so as to differentiate it from the standard $d$-regular variant of the SBM discussed in the previous section.

Using some of the same methods that provide concrete evidence for the SBM conjecture---namely stability of belief propagation and the local statistics hierarchy---we will show that the equitable SBM appears to exhibit a different computational threshold from the standard SBM. Specifically, we conjecture that no polynomial-time algorithm can distinguish (w.h.p.)\ between $\sG_{n,d}$ and $\sG_{n, d, k, \eta}^{\,\mathrm{eq}}$ when
\begin{equation}\label{eq:KS-eq}
    d < d_{\mathrm{KS}}^{\,\mathrm{eq}} = d_{\mathrm{KS}}^{\,\mathrm{eq}}(\eta) \colonequals \frac{2}{\eta^2}\left(1 + \sqrt{1 - \eta^2}\right).
\end{equation}

\noindent We will state a formal version of this ``eSBM conjecture'' later (Conjecture~\ref{conj:esbm}), which actually pertains to a slightly ``noisy'' version of the equitable block model. We remark that when $k$ is large, $\eta$ must be close to zero (in the disassortative case), and so we have approximately $d^{\,\mathrm{eq}}_{\mathrm{KS}} \approx 4d^{\,\mathrm{sbm}}_{\mathrm{KS}}$.

Repeating our earlier argument for hardness of certification, with the eSBM in place of the SBM, yields the following corresponding result: conditional on the eSBM conjecture, no polynomial-time algorithm can certify a better bound than
\begin{equation} \label{eq:spectral-bound-aftereSBM}
    \MC_k(\bG) \leq \frac{k - 1}{k}\left(1 + \frac{2\sqrt{d - 1}}{d} - o(1)\right)
\end{equation}
when $\bG \sim \sG_{n,d}$, which matches the spectral bound~\eqref{eq:MC-spectral-bound-value}. However, we have ignored an important caveat here: the equitable block model only exists when the parameters $(n,d,k,\eta)$ satisfy certain integrality conditions. For this reason, our actual lower bound (Theorem~\ref{thm:graph-cert-hardness}) is sometimes weaker than~\eqref{eq:spectral-bound-aftereSBM} would suggest; see Section~\ref{sec:integrality} for discussion. When $d \gg k^2$, the integrality conditions are negligible and we obtain a tight lower bound, essentially matching~\eqref{eq:spectral-bound-aftereSBM}. Another setting where we obtain tight results is for the problem of refuting colorability (or more accurately, near-colorability) in the double limit $n \to \infty$ followed by $d \to \infty$, which is discussed in Remarks~\ref{rem:near-coloring} and~\ref{rem:exact-coloring}. This corresponds to the choice $\eta = -\frac{1}{k-1}$. Here, when $\bG \sim \sG_{n,d}$, the following results hold asymptotically: the true value of $\chi(\bG)$ is $\frac{1}{2} \frac{d}{\log d}$, the spectral approach certifies a lower bound of $\frac{1}{2} \sqrt{d}$ on $\chi(\bG)$, the basic SBM planting implies hardness of certifying a lower bound better than $\sqrt{d}$, and the improved eSBM planting implies hardness of certifying a lower bound better than $\frac{1}{2} \sqrt{d}$ (which is tight, matching the spectral bound).

\paragraph{Why is the equitable SBM quieter?}

Here, we give some intuition for why the equitable model is a good quiet planting. For the sake of illustration, it helps to consider the simple rank-1 Wigner spiked matrix model: $\bY = \eta \bv\bv^\top + \bW$ where $\eta > 0$ (the signal-to-noise ratio), $\|\bv\| = 1$ (the planted ``signal,'' drawn from some prior), and $\bW$ (the ``noise'') is a GOE matrix, i.e., a symmetric matrix with $\mathcal{N}(0,1/n)$ entries (see Definition~\ref{def:goe}). For large $n$, the eigenvalues of $\bW$ follow the semicircle law and are contained in the interval $[-2,2]$. For $1 < \eta < 2$, a surprising ``pushout'' effect occurs: although the planted signal $\bv$ has quadratic form $\bv^\top \bY \bv \approx \eta < 2$, its presence causes there to exist some other unit vector $\bu$ achieving $\bu^\top \bY \bu \approx \eta + 1/\eta > 2$ \cite{bbp,fp,CDF-wigner}; as a result, simply checking the largest eigenvalue allows one to distinguish $\bY$ from $\bW$. Even though the signal $\bv$ is ``small'', the vector $\bu$ (which is the leading eigenvector of $\bY$) is able to achieve a ``large'' quadratic form by correlating nontrivially with both the signal $\bv$ and the noise $\bW$. The main result of~\cite{BKW-2019-ConstrainedPCA} can be interpreted as giving a ``quieter'' way to plant the signal ``orthogonal to the noise'' with no pushout effect, i.e., $\bv^\top \bY \bv \approx 2$ and the maximum eigenvalue of $\bY$ is $\approx 2$.

It turns out that a similar effect is at play in the SBM. Here the minimum eigenvalue of the adjacency matrix $\bA$ of a random $d$-regular graph converges to $\lambda_{\min} \colonequals - 2\sqrt{d-1}$. As will be made clear in Section~\ref{sec:prelim}, the SBM is in some sense analogous to the spiked Wigner model with multiple planted vectors, namely $k$ ``coloring vectors'' $\bv_i$ for $i \in [k]$, that encode the planted labelling $\bsigma$ as follows: $(\bv_i)_u = c(k \cdot \indicator{\bsigma(u) = i} - 1)$, where $c$ is chosen so that each $\|\bv_i\| = 1$. Planting a $k$-cut of value $\tfrac{k-1}{k}(1+|\eta|)$ via either the SBM or eSBM has the effect that all planted coloring vectors achieve a small quadratic form: $\bv_i^\top \bA \bv_i \approx -d |\eta|$. In the SBM there is a pushout effect similar to the spiked Wigner model, whereby an eigenvalue less than $\lambda_{\min}$ can be created even when $d|\eta| < |\lambda_{\min}|$ (see~\cite{Nadakuditi_12} for results when the degree grows slowly with $n$). The eSBM, however, has the property that each coloring vector $\bv_i$ is an eigenvector. In particular, the subspace spanned by $\{\bv_i\}_{i\in[k]}$ is orthogonal to the other eigenvectors, which are thus unaffected by the planted structure. As a result, there is no pushout effect in the eSBM, allowing for a larger $k$-cut to be planted without disrupting the minimum eigenvalue.

An alternative viewpoint is that the standard plantings (the spiked Wigner model or SBM) pick a random solution (e.g., a cut) and then condition on that \textit{particular} solution having the desired value. In contrast, the quieter plantings are perhaps more similar to conditioning on the event ``there exists a solution having the desired value.''

\begin{remark}
    \label{rem:conj-esbm-noise}
    The fact that the coloring vectors $\bv_i$ are eigenvectors of the eSBM can actually be exploited to give a polynomial-time algorithm for distinguishing $\sG_{n,d}$ from $\sG_{n, d, k, \eta}^{\,\mathrm{eq}}$ for \emph{any} settings of the parameters. See, for example, \cite{barucca} for some discussion of such algorithms. For this reason, it is crucial that our eSBM conjecture (Conjecture~\ref{conj:esbm}) adds a small amount of noise to the graph in order to ``defeat'' these types of algorithms. We discuss this issue further in Section~\ref{sec:xor}.
\end{remark}

\paragraph{The case $k = 2$: large cuts in $\sG_{n, d}$.}

Here, we briefly discuss the specific case $k = 2$, which is better understood in the existing literature.
In this case, $\MC_2(G)$ is merely the fraction of edges crossing the largest cut of $G$, and thus up to this scaling is the solution to the well-known \emph{max-cut} problem. Equivalently, letting $A_G$ be the adjacency matrix of $G$, we have $\MC_2(G) = \frac{1}{2}(1 + \frac{1}{2 |E|} \Gamma_2(A_G))$ where
\[ 
    \Gamma_2(A) \colonequals -\min_{x \in \{\pm 1\}^n} x^\top A x \ge 0. 
\]
When $\bG \sim \sG_{n,d}$, the behavior of $\Gamma_2(A_{\bG})$ turns out to be deeply connected to its Gaussian analogue $\Gamma_2(\bW)$ where $\bW$ is a GOE matrix. The quantity $\Gamma_2(\bW)$ has been studied in statistical physics, being the ground state energy of the \emph{Sherrington-Kirkpatrick model} of spin glasses \cite{SK-1975-SolvableModel}.
The deep but non-rigorous analysis of Parisi \cite{Parisi-SK} proposed an asymptotic value $\frac{1}{n}\EE\,\Gamma_2(\bW) \to 2P_{*} \approx 1.526$, which was later proven rigorously in a sequence of mathematical works \cite{Guerra-2003-BrokenRSB,Talagrand-Parisi,Panchenko-ultra,Panchenko-SK}. By relating the graph model to the Gaussian model, \cite{DMS-2017-CutsSparseRandomGraphs} gave an asymptotic formula for the size of the largest cut in a random regular graph with large degree:
\begin{equation}
    \label{eq:max-cut-val}
    \lim_{n \to \infty}\Ex_{\bG \sim \sG_{n, d}}\MC_2(\bG) = \frac{1}{2}\left(1 + \frac{2P_{*}}{\sqrt{d}} + o_d\left(\frac{1}{\sqrt{d}}\right)\right).
\end{equation}
For the Gaussian setting, it was shown in~\cite{BKW-2019-ConstrainedPCA} using a quiet planting approach that (conditional on a certain complexity assumption based on the low-degree likelihood ratio) the best possible upper bound on $\frac{1}{n} \Gamma_2(\bW)$ that can be certified in polynomial time is $2$, in contrast with the true value $2P_*$; the optimal bound is given by a simple spectral certificate involving the maximum eigenvalue of $\bW$. Given this, we might expect that the best efficiently-certifiable bound on $\MC_2(\bG)$ for $\bG \sim \sG_{n, d}$ is given by replacing $2P_*$ in~\eqref{eq:max-cut-val} with 2.
Indeed, \cite{MS-2016-SDPSparseRandomGraphs} and \cite{MRX-2019-SOS4} showed respectively that the degree-2 and degree-4 sum-of-squares relaxations can certify a bound no better than
\begin{equation}
    \label{eq:k2-best-known-bound}
    \MC_2(\bG) \leq \frac{1}{2}\left(1 + \frac{2}{\sqrt{d}} + o_d\left(\frac{1}{\sqrt{d}}\right)\right).
\end{equation}

\noindent Our results extend the picture emerging from this literature in two important ways.
We show (see Theorem~\ref{thm:graph-cert-hardness} and discussion in Section~\ref{sec:integrality}) that conditional on the eSBM conjecture, \eqref{eq:k2-best-known-bound} is in fact the optimal bound on $\MC_2(\bG)$ certifiable in polynomial time.
Thus, if the eSBM conjecture holds then no constant-degree sum-of-squares relaxation can improve upon on the known results for degree-2 and degree-4. Furthermore, we justify this bound with an explicit quiet planting of a large cut in a random regular graph.

\paragraph{The Gaussian $k$-cut model.}

Above, we have seen an intimate connection between $\Gamma_2(A_{\bG})$ where $A_{\bG}$ is the adjacency matrix of $\bG \sim \sG_{n, d}$, and its Gaussian counterpart $\Gamma_2(\bW)$ where $\bW \sim \mathrm{GOE}(n)$. More generally, $\MC_k(\bG)$ can be written in terms of a certain quantity $\Gamma_k(A_{\bG})$ defined in Section~\ref{sec:prelim}, which has a natural Gaussian counterpart $\Gamma_k(\bW)$.

In fact, explicit formulas similar to \eqref{eq:max-cut-val} are also known that relate the asymptotic values (in the double limit $n \to \infty$ followed by $d \to \infty$) of $\Gamma_k(A_{\bG})$ and $\Gamma_k(\bW)$ even for $k > 2$ \cite{Sen-2018-SparseRandomHypergraphs,JKS-2018-kCutPotts}.
The broadly applicable techniques used in these results---the Lindeberg exchange method and related probabilistic interpolation arguments---suggest that there is a general and fundamental relationship between the graph model in the large-degree limit and the Gaussian model. Yet, it is not clear whether this implies any relation between the respective thresholds for efficient certification.

To clarify this matter, we also give results showing that the problem of certifying upper bounds on $\Gamma_k(\bW)$ under the Gaussian model exhibits similar behavior to the graph model: no polynomial-time certifier can improve over the basic spectral bound. The proof is again based on quiet planting, and can be seen as an extension of the results of~\cite{BKW-2019-ConstrainedPCA} which handled the $k=2$ case. Our results rely on a complexity assumption concerning the low-degree likelihood ratio, which we discuss further in Section~\ref{sec:low-deg-intro}. As a by-product, we develop a general framework for bounding the low-degree likelihood ratio of certain Gaussian models, which may be of independent interest. Specifically, we conduct the low-degree analysis of a broad class of multi-spiked matrix models (both Wigner and Wishart), and give an improved low-degree analysis of the stochastic block model which suggests that fully-exponential time is needed below the Kesten--Stigum threshold.

\subsection{Heuristics for Average-Case Computational Hardness}
\label{sec:heuristics}

Our results on hardness of certification rely on unproven conjectures about average-case hardness, such as the eSBM conjecture. Proving these types of conjectures seems to be beyond the reach of current techniques (even when assuming standard complexity conjectures such as $\texttt{P} \ne \texttt{NP}$), as illustrated by the fact that no such proof is known for the famous \emph{planted clique} problem. However, a myriad of heuristic techniques have emerged for predicting hardness of average-case problems by proving lower bounds against certain classes of algorithms. Taken together, these methods create a fairly coherent theory of computational complexity for a large class of high-dimensional Bayesian inference problems. In this section we describe the three methods that will be used in this work: belief propagation, the local statistics hierarchy, and the low-degree likelihood ratio. We remark that these are not the only such methods, some others being average-case reductions~\cite{BR-reduction,BBH}, sum-of-squares lower bounds~\cite{BHKKMP-2019-PlantedClique}, statistical query lower bounds~\cite{sq-clique}, the overlap gap property~\cite{gz-ogp}, and analysis of non-convex loss landscapes~\cite{ABC}. Finally, in Section~\ref{sec:xor} we discuss the XOR-SAT problem---an important counterexample for many of the above heuristics---and the related issues of ``robustness'' that arise in our eSBM conjecture.

\subsubsection{Belief Propagation and the Kesten--Stigum Threshold}
\label{sec:bp-intro}

The sharp computational phase transition known as the \emph{Kesten--Stigum (KS) threshold} in the stochastic block model, was first predicted by~\cite{sbm-ks-1,sbm-ks-2} using non-rigorous ideas inspired by statistical physics. The idea is to consider the \emph{belief propagation (BP)} algorithm, an iterative method that attempts to recover the planted community structure by keeping track of ``beliefs'' about each node's community label and updating these in a locally-Bayesian-optimal way. BP has an ``uninformative'' fixed point, which is a natural starting point for the algorithm where the beliefs reflect no knowledge of the communities. It was shown in~\cite{sbm-ks-1,sbm-ks-2} that if the signal-to-noise ratio (SNR) lies above the KS threshold then the uninformative fixed point is unstable, suggesting that BP should leave it and find a community assignment that correlates with the truth; and if the SNR lies below the KS threshold then the uninformative fixed point is stable, meaning that BP will remain there and fail to find a nontrivial solution. It was later proven that indeed it is possible to nontrivially recover the communities in polynomial time when above the KS threshold~\cite{massoulie,mns,AS-acyclic}, whereas no such algorithm is known below the KS threshold. More generally, similar computational thresholds have been predicted in various models (e.g.,~\cite{LKZ-sparse,LKZ-mmse}) by examining the stability of BP or its simplified variant, \emph{approximate message passing (AMP)}~\cite{amp}. For many high-dimensional inference problems, it is known that BP and AMP achieve optimal information-theoretic performance (e.g.~\cite{DAM}); and when they don't, it is often conjectured that they achieve the best possible performance among efficient algorithms. Thus, stability of BP provides concrete evidence for computational hardness.

\subsubsection{The Local Statistics Hierarchy}

Introduced by one of the authors, Mohanty, and Raghavendra in \cite{banks2019local} and building off of the work of Hopkins and Steurer \cite{HS-bayesian}, the Local Statistics hierarchy is a family of increasingly powerful semidefinite programming algorithms for solving Bayesian hypothesis testing problems. By analogy with the Sum-of-Squares algorithm, we take this hierarchy as a proxy for hardness: the higher we need to go to perform the hypothesis test, the harder it is.

Consider a generic inference scheme where we are to distinguish between a null model $\QQ$ which outputs unstructured data $\bG \in \RR^m$, and a planted model $\PP$ generating structured data $\bG \in \RR^m$ according to some random and hidden signal $\bx \in \RR^n$. Letting $x = (x_1,...,x_n)$ be a set of variables, we may regard the conditional expectation $\EE_{\bx \sim \PP} [p(\bx) | \bG]$ as a random linear functional from $\RR[x] \mapsto \RR$ that is positive in a certain sense: $\EE_{\bx \sim \PP} [p(\bx)^2 | \bG] \ge 0$ for every polynomial $p$. 

The Local Statistics hierarchy is parameterized by two integers $(D_x,D_G)$. Given as input some $G_0 \in \RR^m$, it attempts to find a linear functional that approximates this conditional expectation $\Ex [p(x) | G]$ in the planted model. In particular, borrowing terminology from Sum-of-Squares programming, we search for a ``pseudoexpectation'' functional $\pseudo$ that assigns a real number to every polynomial of degree at most $D_x$ in $\RR[x]$, with the constraints that (i) $\pseudo p(x)^2 \ge 0$, and (ii) $\pseudo p(x) \approx \Ex_{(\bx,\bG) \sim \PP} p(\bx)$ for every polynomial $p(x) \in \RR[x]$ whose coefficients are of degree at most $D_G$ in the input $G_0$. It is well-known that this may be written as a SDP on matrices of size $O(n^{D_x})$ with $O(m^{D_G})$ affine constraints.

In many cases, this paper included, this problem is soluble with high probability when the input $G_0$ is sampled from the planted model $\PP$. For instance, the evaluation map $p(x) \mapsto p(\bx)$ is a feasible solution provided that the $\PP$ is sufficiently concentrated and these polynomials do not fluctuate too much about their expectations. On the other hand, by taking $D_x$ and $D_G$ sufficiently large it becomes infeasible when the input is drawn from a different distribution. The $D_x$ and $D_G$ necessary measure the hardness of the hypothesis testing problem.

\subsubsection{The Low-Degree Likelihood Ratio}
\label{sec:low-deg-intro}

As was first discovered in a series of works in the sum-of-squares literature \cite{BHKKMP-2019-PlantedClique,HS-bayesian,HKPRSS-2017-SOSSpectral,sam-thesis}, analyzing the \emph{low-degree likelihood ratio} gives predictions of computational hardness that match widely-believed conjectures for many hypothesis testing problems (as corroborated by the various other heuristics mentioned above). 
In essence, this method takes low-degree polynomials as a proxy for all polynomial-time algorithms, and analyzes whether any low-degree polynomial can distinguish two distributions $\QQ$ and $\PP$ (with the same interpretations as above) as $n \to \infty$.

The key to making this analysis tractable is to choose the correct soft notion of ``successfully distinguishing.''
This is done by considering the maximization
\begin{equation}\label{eq:ldlr-var}
    \begin{array}{ll}
    \text{maximize} & \frac{\Ex_{\bx \sim \PP} p(\bx)}{(\Ex_{\bx \sim \QQ}p(\bx)^2)^{1/2}} \\
    \text{such that} & p \neq 0 \text{ is a polynomial of degree} \leq D.
    \end{array}
\end{equation}
In words, we seek to maximize $p(\bx)$ under $\PP$ in expectation, while keeping its typical magnitude under $\QQ$ modest.

We note that, if we did not restrict $p$ to be a low-degree polynomial, then the optimal $p$ would equal the classical \emph{likelihood ratio}, $L \colonequals \frac{d\PP}{d\QQ}$.
In absence of computational constraints, thresholding $L$ gives an optimal test between $\PP$ and $\QQ$ in the sense of minimizing error probabilities, as shown in the classical Neyman-Pearson lemma \cite{NP-1933-MostEfficientTests}.
Moreover, the value of the above problem would be the norm $\|L\|$ of $L$ in $L^2(\QQ)$.
If that norm is bounded as $n \to \infty$, then $\PP$ and $\QQ$ cannot be distinguished w.h.p.\ by \emph{any} test, by an application of Le Cam's second moment method for contiguity (see \cite{LCY-2012-AsymptoticsStatistics,lowdeg-notes} for further exposition). 

With the further constraint to low-degree polynomials, the result is similar: the optimal $p$ is the aforementioned low-degree likelihood ratio, the orthogonal projection of $L$ to the subspace of degree-$D$ polynomials in $L^2(\QQ)$, which we denote $L^{\leq D}$.
The value of the problem is its norm, $\|L^{\leq D}\|$.
We again consider whether, as $n \to \infty$, for $D = D(n)$ slowly growing with $n$, this norm diverges or remains bounded.
If it diverges, we expect that low-degree polynomials can distinguish $\PP$ from $\QQ$ w.h.p.\ (at an intuitive level, the algorithm we have in mind is thresholding $L^{\leq D}$); if it remains bounded, we conclude that low-degree polynomials cannot distinguish $\PP$ from $\QQ$ w.h.p.\ in this particular sense, and therefore expect that no polynomial-time algorithm can do so either.

The precise scaling of $D(n)$ relates to the efficiency of algorithms that the heuristic pertains to---higher degree polynomials describe more time-consuming computations.
However, while constant-degree polynomials may be evaluated in polynomial time, some other polynomial-time computations require slightly higher degree polynomials to express.
A crucial example is approximating the spectral norm of a matrix with dimensions polynomial in $n$, which requires a polynomial of degree $\Theta(\log n)$ to approximate accurately.
Taking this into account, a low-degree lower bound with $D(n) \gg \log n$ is taken as evidence that no polynomial-time test exists.
Similarly, we view a low-degree lower bound with $D(n) \gg n^{\delta}$ for some $\delta \in (0, 1)$ as suggesting a lower bound against tests with runtime $O(\exp(n^{\delta}))$.
See \cite{lowdeg-notes} for further discussion.

\subsubsection{A Caveat: XOR-SAT}
\label{sec:xor}

In the planted 3-XOR-SAT problem, we are given $m$ clauses of the form $x_i x_j x_k = b$ for some choice of $i,j,k \in [n]$ and $b \in \{\pm 1\}$. The goal is to distinguish the case where the clauses are completely random from the case where there is a planted assignment of $\{\pm 1\}$ values to the variables $x_1,\ldots,x_n$ such that all clauses are satisfied. This is a notable counterexample for many heuristics for average-case hardness, including sum-of-squares and all three methods mentioned above (see e.g., Lecture~3.2 of \cite{sos-notes} or Chapter 18 of \cite{MM-2009-InformationPhysicsComputation}). These heuristics predict that the distinguishing task is possible in polynomial time only when $m \gtrsim n^{3/2}$, whereas in reality the problem is much easier: Gaussian elimination can be used to decide with certainty whether or not there is a satisfying assignment. However, if we change the planted distribution so that only a $1-\varepsilon$ fraction of the clauses are satisfied, Gaussian elimination breaks down and the best known algorithms indeed require $m \gtrsim n^{3/2}$. In this sense, the above heuristics seem to predict the threshold for ``robust'' algorithms. We expect that a similar phenomenon is at play in the eSBM model: while there exist algorithms to solve the problem by exploiting brittle algebraic structure in the eigenvectors (see Remark~\ref{rem:conj-esbm-noise}), we conjecture (Conjecture~\ref{conj:esbm}) that the threshold predicted by the heuristics is the correct computational threshold for a \emph{noisy} version of the problem.

\subsection{Notation}\label{sec:notation}

We use standard asymptotic notation such as $o(\cdot)$ and $O(\cdot)$; unless stated otherwise, this always pertains to the limit $n \to \infty$ with other parameters (such as $d, k, \eta$) held fixed. We use e.g., $o_d(\cdot)$ or $O_d(\cdot)$ when considering the double limit $n \to \infty$ followed by $d \to \infty$; for example, $f(n,d) = o_d(g(n,d))$ means that for any $\eps > 0$ there exists $d_0 > 0$ such that for all $d \ge d_0$ there exists $n_0 > 0$ such that for all $n \ge n_0$ we have $|f(n,d)/g(n,d)| \le \eps$. An event occurs \emph{with high probability (w.h.p.)}\ if it has probability $1-o(1)$. We write $f \sim g$ to mean $f/g \to 1$.

Throughout, all graphs are assumed to be \emph{simple}---i.e., without self-loops and multiple edges---unless stated otherwise. A \emph{$d$-regular} graph has degree $d$ at every vertex. $\sG_{n, d}$ denotes the uniform distribution over $d$-regular $n$-vertex graphs.

We will use $\id_k$ denote the $k \times k$ identity matrix and $\onesmat_k$ for the $k \times k$ all-ones matrix; $e_1,e_2,...$ are the standard unit basis vectors, and $\onesvec$ will denote the all-ones vector. For matrices, $\|\cdot\|$ denotes the operator (spectral) norm and $\|\cdot\|_\F$ denotes the Frobenius norm. We use $\indicator{A}$ for the ($\{0,1\}$-valued) indicator of an event $A$. We write $[k] = \{1,2,\ldots,k\}$ and $\NN = \{0,1,2,\ldots\}$.

\section{Main Results}

\subsection{Setup and Definitions}
\label{sec:prelim}

We consider a general framework that captures the problem of certifying bounds on max-$k$-cut in a random graph, as well as a Gaussian variant of this problem.

\begin{definition}\label{def:labeling_partition}
For a labeling $\sigma: [n] \to [k]$, the associated \emph{partition matrix} $P = P^{(\sigma)} \in \RR^{n \times n}$ is given by
\[ 
    P_{i,j} = \begin{cases} 1 & \sigma(i) = \sigma(j), \\ -1/(k-1) & \sigma(i) \neq \sigma(j). \end{cases}
\]
\end{definition}

\noindent A basic fact is that $P \succeq 0$ and $\Rank(P) = k-1$. This can be seen by realizing $P$ as the Gram matrix of a certain collection of vectors: assign each $u \in [n]$ one of the $k$ unit vectors in $\RR^{k-1}$ pointing to the corners of a simplex, according to its label $\sigma(u) \in [k]$.

Let $\Pi$ be the set of all partition matrices. For a given matrix $A \in \RR^{n \times n}$ and a given $k$, we will be interested the problem of algorithmically certifying an upper bound on the value
\begin{equation}
    \label{eq:Gamma-min}
    \Gamma_k(A) \colonequals \max_{P \in \Pi} \langle P,-A \rangle.
\end{equation}
If $A_G \in \{0,1\}^{n \times n}$ is the adjacency matrix of a graph $G = (V,E)$ (defined with $(A_G)_{i,i} = 0$), then $\langle P^{(\sigma)}, -A_G \rangle = 2(k \cdot m_\sigma(G) - |E|)/(k-1)$, where $m_\sigma(G)$ is the number of monochromatic edges of $G$ under the labeling $\sigma$. As a result,
\begin{equation}
    \MC_k(G) = \frac{k-1}{k} \left(1 +  \frac{\Gamma_k(A_G)}{2|E|}\right).
\end{equation}

\noindent Thus, an upper bound on $\Gamma_k(A_G)$ translates to an upper bound on $\MC_k(G)$, which in turn can be used to refute $k$-colorability. We now formally define the certification task.

\begin{definition}
    \label{def:cert}
    Let $\QQ = \QQ_n$ be a sequence of distributions $\RR^{n \times n}$ and let $\mathcal{A} = \mathcal{A}_n : \RR^{n\times n} \to \RR$ be a sequence of algorithms. We say that $\mathcal{A}$ \emph{certifies the upper bound} $B$ on $\Gamma_k$ over $\bA \sim \QQ$ if both of the following hold:
    \begin{enumerate}
        \item[(i)] for \emph{every} $A \in \RR^{n \times n}$, $\mathcal{A}(A) \ge \Gamma_k(A)$, and
        \item[(ii)] if $\bA \sim \QQ_n$ then $\mathcal{A}_n(\bA) \le B$ with probability $1-o(1)$.
    \end{enumerate}
\end{definition}

\noindent Crucially, $\mathcal{A}(A)$ must \emph{always} be a valid upper bound, even if $A$ is atypical under $\QQ$. In exponential time it is possible to compute $\Gamma_k(A)$ exactly and thus achieve perfect certification; we are interested instead in polynomial-time certification procedures.

There is a simple spectral approach~\cite{Hoffman-1970-Eigenvalues} to certifying bounds on $\Gamma_k(A)$. Let $\lambda_{\min} = \lambda_{\min}(A)$ denote the minimum eigenvalue of $A$. For any $P \in \Pi$ we have $P \succeq 0$ and $A - \lambda_{\min} I \succeq 0$, so
\[ 
    0 \le \langle P, A - \lambda_{\min} I \rangle = \langle P,A \rangle - \lambda_{\min} n 
\]
and $-\lambda_{\min} n$ is an efficiently-computable upper bound on $\Gamma_k(A)$. Our main results give evidence that it is computationally hard to improve upon this spectral bound when $\bA$ is drawn from certain distributions $\QQ$: (i) random $d$-regular graphs, and (ii) Gaussian matrices.

\subsection{Random Regular Graphs}

We will be concerned with the following task of distinguishing two distributions, also called \emph{hypothesis testing} or \emph{detection}.

\begin{definition}\label{def:detection}
    Let $\PP_n$ and $\QQ_n$ be probability measures on the same space $\Omega_n$. We say that an algorithm $t_n: \Omega_n \to \{\textsc{p},\textsc{q}\}$ \emph{distinguishes $\PP_n$ and $\QQ_n$ with high probability}, or (equivalently) \emph{achieves strong detection} (between $\PP_n$ and $\QQ_n$) if
    \[ 
        \PP_n[t_n(\bx) = \textsc{q}] + \QQ_n[t_n(\bx) = \textsc{p}] = o(1)
    \]
    as $n \to \infty$.
\end{definition}

\noindent Our results for random regular graphs are conditional on a conjecture regarding computational hardness of detection in a noisy variant of the equitable stochastic block model (eSBM). The extra noise is crucial, as discussed in Remark~\ref{rem:conj-esbm-noise} and Section~\ref{sec:xor}. Specifically, the noise takes the form of ``rewiring'' a small constant fraction of the edges as follows.

\begin{definition}[Noise Operator]
    If $G = (V,E)$ is a $d$-regular $n$-vertex graph and $\delta > 0$, let $\bT_\delta(G)$ denote the random $d$-regular graph obtained from $G$ by making $\lfloor \delta n \rfloor$ `swaps.' That is, repeatedly choose a pair of distinct edges $(i,j), (k,\ell)$ uniformly at random\footnote{Here, the ordering of the tuples matters, so e.g., $(i,j)$ and $(j,i)$ should be chosen with equal probability.} conditioned on the following events: $i \ne k$, $j \ne \ell$, $(i,k) \notin E$, and $(j,\ell) \notin E$. Remove edges $(i,j)$ and $(k,\ell)$, and add edges $(i,k)$ and $(j,\ell)$.
\end{definition}

\noindent Recall the definition of the eSBM (Definition~\ref{def:eSBM}) and the associated threshold $d_{\mathrm{KS}}^{\,\mathrm{eq}}(\eta)$ defined in~\eqref{eq:KS-eq}.

\begin{conjecture}[eSBM Conjecture]
    \label{conj:esbm}
    Let $\widetilde{\sG}^{\,\mathrm{eq}}_{n,d,k,\eta, \delta}$ denote the distribution over $d$-regular $n$-vertex graphs given by $\bT_\delta(\bG)$ where $\bG \sim \sG_{n, d, k, \eta}^{\,\mathrm{eq}}$. Suppose $\delta > 0$, $\eta \in [-\frac{1}{k-1},1]$, $k \geq 2$, and $d < d_{\mathrm{KS}}^{\,\mathrm{eq}}(\eta)$ are all fixed.
    Also suppose $k\,|\,(1-\eta)d$ so that the equitable model is defined for an infinite sequence of values of $n$.
    Then, there exists no polynomial-time algorithm that with high probability distinguishes $\widetilde{\sG}^{\,\mathrm{eq}}_{n,d,k,\eta,\delta}$ from $\sG_{n, d}$ as $n \to \infty$ (in the sense of Definition~\ref{def:detection}).
\end{conjecture}

\noindent The above conjecture implies the following result on hardness of certifying bounds on max-$k$-cut.

\begin{theorem}
    \label{thm:graph-cert-hardness}
    Assume the eSBM conjecture (Conjecture~\ref{conj:esbm}) holds. Fix $d > 0$ and $k \ge 2$ and let $\eta \in [-\frac{1}{k-1},0]$ be such that $|\eta| < \frac{2\sqrt{d-1}}{d}$ and $k \,|\, (1-\eta)d$. Then for any $\epsilon > 0$, no polynomial-time algorithm can certify, in the sense of Definition~\ref{def:cert}, the upper bound $nd(|\eta|-\epsilon)$
    on $\Gamma_k$ over $\sG_{n, d}$. Equivalently, no polynomial-time algorithm can certify for $\bG \sim \sG_{n, d}$ the bound
    \begin{equation} 
        \MC_k(\bG) \leq \frac{k - 1}{k}\left(1 + |\eta| - \epsilon\right). \end{equation}
\end{theorem}

\noindent The proof idea is described in the Introduction; the full details are given in Section~\ref{sec:reduction-graph}. The effects of the integrality condition $k \,|\, (1-\eta)d$ are discussed in Section~\ref{sec:integrality} below; this condition can be ignored when $d \gg k^2$.

\begin{remark}[Near-Coloring]
\label{rem:near-coloring}
An important special case where Theorem~\ref{thm:graph-cert-hardness} is tight is $\eta = -\frac{1}{k-1}$, corresponding to refutation of near colorability. Here the integrality condition $k \,|\, (1-\eta)d$ reduces to $(k-1) \,|\, d$. Thus for any $d,k$ satisfying $(k-1) \,|\, d$ and $k > 1+\frac{d}{2\sqrt{d-1}}$ we have that for any $\epsilon > 0$, no polynomial-time algorithm can certify $\MC_k(\bG) \le 1 - \epsilon$. There is an infinite sequence of $(d,k)$ values with $k \sim \frac{1}{2} \sqrt{d}$ for which the conditions on $d,k$ are satisfied, namely $d = 4(k-1)(k-2)$ for all $k > 2$. Thus, our result is asymptotically tight, matching the spectral algorithm in the double limit $n \to \infty$ followed by $d \to \infty$.
\end{remark}

\begin{remark}[Exact Coloring]
\label{rem:exact-coloring}
As stated, Theorem~\ref{thm:graph-cert-hardness} only shows hardness of refuting a near-coloring, as opposed to an exact coloring. While we expect a similar hardness result to hold for exact coloring, this does not follow from the eSBM conjecture because the noise operator $\bT_\delta$ prevents us from planting an exact coloring. Hardness of refuting exact coloring would follow from a variant of the eSBM conjecture where the noise operator only makes swaps that do not affect the value of the planted cut.
\end{remark}

Our next two results give concrete evidence for the eSBM conjecture using two different heuristics (discussed in Section~\ref{sec:heuristics}): stability of belief propagation, and the local statistics hierarchy. Both of these methods predict $d_{\mathrm{KS}}^{\,\mathrm{eq}}(\eta)$ as the computational threshold for the eSBM.

\begin{theorem}[Informal; Kesten--Stigum threshold of eSBM]
\label{thm:ks-main}
    For $k \ge 4$, the Kesten--Stigum threshold of the eSBM with parameters $(n, d, k, \eta)$, defined as the smallest number $d_{\mathrm{KS}}^{\,\mathrm{eq}}$ so that, for all $d > d_{\mathrm{KS}}^{\,\mathrm{eq}}$, the ``uninformative'' fixed point of the belief propagation iteration is unstable, is given by
    \[
        d_{\mathrm{KS}}^{\,\mathrm{eq}} = d_{\mathrm{KS}}^{\,\mathrm{eq}}(\eta) \colonequals \frac{2}{\eta^2}\left(1 + \sqrt{1 - \eta^2}\right).
    \]
\end{theorem}

\noindent Further details, as well as the full analysis leading to the above result, can be found in Section~\ref{sec:bp}.

\begin{theorem}[Local Statistics analysis of eSBM]
    \label{thm:local-stats-main}
    If $d > d^{\,\mathrm{eq}}_{\mathrm{KS}}(\eta)$, then there exist $D$ sufficiently large and $\delta > 0$ so that the degree-$(2,D)$ Local Statistics algorithm with error tolerance $\delta$ can distinguish $\sG_{n, d}$ and $\sG^{\,\mathrm{eq}}_{n,d,k,\eta}$ with high probability. If $d \leq d^{\,\mathrm{eq}}_{\mathrm{KS}}(\eta)$, no such $D$ and $\delta$ exist.
\end{theorem}

\noindent Further details on the Local Statistics hierarchy, and the proof of Theorem~\ref{thm:local-stats-main}, can be found in Section~\ref{sec:local-stats}.

\subsubsection{Discussion of Integrality Condition}
\label{sec:integrality}

The integrality constraints in the eSBM (Definition~\ref{def:eSBM}) and in the above results may seem somewhat mystifying. In this section we will try to shed some light on the limits of our method for quiet planting, and how they compare to the power of spectral refutation. The integrality condition $k\,|\,(1-\eta)d$ from Theorem~\ref{thm:graph-cert-hardness} is required so that the equitable block model exists for an infinite sequence of values of $n$. This condition can be written as $\eta d \equiv d \pmod{k}$, which means $\eta$ is constrained to lie in a certain grid of spacing $\frac{k}{d}$ within the interval $[-\frac{1}{k-1},0]$. In order to extract the strongest possible hardness result from Theorem~\ref{thm:graph-cert-hardness}, the goal is to pick $\eta$ as to maximize $|\eta|$ subject to the integrality condition and $|\eta| < \frac{2\sqrt{d-1}}d$. If $d \gg k^2$ then we have $\frac{k}{d} \ll \frac{2\sqrt{d-1}}{d}$ and so the grid allows $\eta$ to be chosen extremely close to $\frac{2\sqrt{d-1}}{d}$; this means the effects of the integrality condition are negligible and we essentially obtain the ideal hardness result~\eqref{eq:spectral-bound-aftereSBM} which matches the spectral algorithm. In particular, our results are tight in the regime $d \to \infty$ with $k$ fixed.

As discussed in Remark~\ref{rem:near-coloring}, we also get tight results in the particular case $\eta =  -\frac{1}{k-1}$ corresponding to near-coloring, because the integrality condition simplifies in a helpful way.

On the other hand, for small degree the integrality conditions can impose significant limitations. A particularly severe example is that of $k=2$ and $d=3$, which corresponds to Max-Cut in a random $3$-regular graph. The spectral bound corresponds to
\begin{equation}\label{eq:spectralbound_d3k2}
\mathrm{MC}_2(\bG) \leq \frac12\left(1 + \frac{2\sqrt{3-1}}{3} \right) \approx 0.97.
\end{equation}
On the other hand, the integrality condition from Theorem~\ref{thm:graph-cert-hardness} requires that $2|3(1-\eta)$, in addition to the condition $|\eta| <  \frac{2\sqrt{3-1}}{3}$. Since $3\left(1+\frac{2\sqrt{3-1}}{3}\right) \approx 5.83$ the largest $|\eta|$ admissible is such that that $3\left(1-\eta\right) = 4$, corresponding to $\eta = \frac{1}{3}$. This means that Theorem~\ref{thm:graph-cert-hardness} only addresses certification below
\[
\mathrm{MC}_2(\bG) \leq \frac12\left(1 + \frac{1}{3} \right) = \frac23 \approx 0.67.
\]
This is unfortunate since random 3-regular graphs have $\mathrm{MC}_2(\bG) > 0.88$ with high probability~\cite{Diaz_03}, and calculations from statistical physics~\cite{Zdeborova_10} suggest the even larger value $\mathrm{MC}_2(\bG) \approx 0.92$. That is, using the equitable block model to plant a 2-cut in a 3-regular graph, by giving each vertex one neighbor in its own group and two neighbors in the other group, plants a cut that is smaller than naturally arising cuts in random 3-regular graphs. 

In general, aside from the special case of near-coloring (Remark~\ref{rem:near-coloring}), our results are most compelling when $k$ is small and $d$ is large, so that the integrality condition in Theorem~\ref{thm:graph-cert-hardness} does not create a significant gap from the limiting threshold $|\eta| = \frac{2\sqrt{d-1}}{d}$. We believe the investigation of quiet plantings in small-degree graphs to be an interesting direction of future research.

\subsection{The Gaussian $k$-Cut Model}

\subsubsection{Hardness of Certification}

In this section we discuss a Gaussian analogue of the coloring problem, namely the problem of certifying upper bounds on $\Gamma_k(\bW)$ where $\bW$ is a GOE matrix defined as follows.

\begin{definition}
    \label{def:goe}
    The \emph{Gaussian orthogonal ensemble} is the following distribution $\GOE(n)$ over random matrices: $\bW \sim \GOE(n)$ is symmetric ($\bW_{u,v} = \bW_{v,u}$) with diagonal entries $\bW_{u,u} \sim \mathcal{N}(0,2/n)$ and off-diagonal entries $\bW_{u,v} \sim \mathcal{N}(0,1/n)$, where the values $\{\bW_{u,v} \,:\, u \le v\}$ are independent.
\end{definition}

\noindent It is well known that (as $n \to \infty$) the eigenvalues of $\bW \sim \GOE(n)$ follow the Wigner semicircle law supported on $[-2,2]$, and in particular, $\lambda_{\min}(\bW) \to -2$ almost surely. Thus, the spectral approach (see Section~\ref{sec:prelim}) certifies the upper bound $\Gamma_k(\bW) \le (2 + o(1)) n$. Our main results give rigorous evidence in support of the following conjecture, which states that improving upon this spectral bound requires fully exponential time.

\begin{conjecture}
    \label{conj:gauss-cert}
    For any constants $k \ge 2$, $\epsilon > 0$ and $\delta > 0$, there is no algorithm of runtime $\exp(O(n^{1-\delta}))$ that certifies the upper bound $(2-\epsilon)n$ on $\Gamma_k(\bW)$ over $\bW \sim \GOE(n)$.
\end{conjecture}

\noindent Our evidence for this conjecture can be seen as a generalization of~\cite{BKW-2019-ConstrainedPCA}, which handles the $k=2$ case. In Section~\ref{sec:reduction-gaussian} we give a reduction from a certain hypothesis testing problem to the certification problem in question. In analogy to our results on coloring, this reduction can be seen as constructing a planted distribution that directly plants an eigenspace with no pushout effect; this planting is computationally quiet, conditional on hardness of the testing problem. The testing problem is a particular instance of the spiked Wishart model with a rank-$(k-1)$ negative spike. In Section~\ref{sec:low-deg-spike} below, we state results analyzing the low-degree likelihood ratio (see Section~\ref{sec:low-deg-intro}) for this model; these results suggest that fully exponential time is required (in the appropriate parameter regime).

As a by-product of our analysis, we give new bounds on the low-degree likelihood ratio for a wide class of multi-spiked matrix models (both Wigner and Wishart), which may be of independent interest. These results also extend to certain binary-values analogues of these problems, including (a variant of) the stochastic block model. These results are also discussed in Section~\ref{sec:low-deg-spike} below.

\subsubsection{Low-Degree Hardness for General Spiked Models}
\label{sec:low-deg-spike}

We consider general variants of the spiked Wigner and Wishart models, defined as follows.

\begin{definition}
    Let $\sX = (\sX_n)$ be a probability measure over $\RR^{n \times n}_{\sym}$, the symmetric $n \times n$ matrices. The \emph{general spiked Wigner} model with spike prior $\sX$ and signal-to-noise ratio $\lambda \in \RR$ is specified by the following null and planted distributions over $\RR^{n \times n}_{\sym}$.
    \begin{itemize}
        \item Under $\QQ$, draw $\bY \sim \GOE(n)$.
        \item Under $\PP$, let $\bY = \lambda \bX + \bW$ where $\bX \sim \sX$ and $\bW \sim \GOE(n)$, independently.
    \end{itemize}
\end{definition}

\begin{definition}
    \label{def:general-wishart-model}
    Let $\sX = (\sX_n)$ be a probability measure over $\RR^{n \times n}_{\sym}$. The \emph{general spiked Wishart} model with spike prior $\sX$, signal-to-noise ratio $\beta > -1$, and number of samples $N \in \NN$, is specified by the following null and planted distributions over $(\by_1, \dots, \by_N) \in (\RR^n)^N$.
    \begin{itemize}
        \item Under $\QQ$, draw $\by_u \sim \sN(0, \id_n)$ independently for $u \in [N]$.
        \item Under $\PP$, first draw $\widetilde \bX \sim \sX$ and define
        \begin{equation}
            \label{eq:wish-prior}
            \bX = \begin{cases}
                \widetilde \bX & \text{if }\beta \widetilde \bX \succ -\id_n, \\
                0 & \text{else}.
            \end{cases}
        \end{equation}
        Then draw $\by_u \sim \sN( 0, \id_n + \beta \bX)$ independently for $u \in [N]$.
    \end{itemize}
\end{definition}

\noindent The purpose of~\eqref{eq:wish-prior} is to ensure that $\id_n + \beta X$ is a valid covariance matrix. We will consider priors for which the first case of~\eqref{eq:wish-prior} occurs with high probability. Specifically, we focus on the following class of priors that are PSD with constant rank.

\begin{definition}
    \label{def:pi}
    Fix an integer $k \ge 1$. Let $\pi$ be a probability measure supported on a bounded subset of $\RR^k$, satisfying $\EE[\pi] = 0$ and $\|\cov(\pi)\| = 1$. Let $\sX(\pi)$ denote the spike prior which outputs $\bX = \frac{1}{n}\bU\bU^\top$, where $\bU$ is $n \times k$ with each row distributed according to $\pi$. (We do not allow $\pi$ to depend on $n$.)
\end{definition}

\noindent It is well known in random matrix theory that polynomial-time detection in the Wigner model is possible when $|\lambda| > 1$, by thresholding the maximum (or minimum if $\lambda < 0$) eigenvalue of $\bY$~\cite{fp,CDF-wigner}. Similarly, poly-time detection in the Wishart model is possible when $\beta^2 > n/N$, by thresholding the maximum or minimum eigenvalue of the sample covariance matrix $\bY = \frac{1}{N} \sum_i \by_i \by_i^\top$~\cite{bbp,BS-spiked}. In the general setting above, we have the following bounds on the norm of the low-degree likelihood ratio $\|L^{\le D}\|$ (see Section~\ref{sec:low-deg-intro}), which suggest that fully exponential time is required to solve the detection problem below this spectral threshold. Our results are consistent with the computational thresholds predicted by~\cite{LKZ-mmse}, where it was shown that the approximate message passing algorithm fails below the spectral threshold. The proofs are deferred to Section~\ref{sec:low-deg}.

\begin{theorem}
    \label{thm:general-wigner}
    Fix constants $k \ge 1$ and $\lambda \in \RR$. Fix $\pi$ satisfying the requirements in Definition~\ref{def:pi}. Consider the general spiked Wigner model with spike prior $\sX(\pi)$. If $|\lambda| < 1$ then $\|L^{\le D}\| = O(1)$ for any $D = o(n/\log n)$.
\end{theorem}

\begin{theorem}
    \label{thm:general-wishart}
    Fix constants $k \ge 1$, $\beta > -1$, and $\gamma > 0$. Fix $\pi$ satisfying the requirements in Definition~\ref{def:pi}. Consider the general spiked Wishart model with any $N = N_n$ satisfying $n/N \to \gamma$ as $n \to \infty$, and with spike prior $\sX(\pi)$. If $\beta^2 < \gamma$ then $\|L^{\le D}\| = O(1)$ for any $D = o(n/\log n)$.
\end{theorem}

\begin{remark}
    In the setting of Theorem~\ref{thm:general-wishart}, the first case of~\eqref{eq:wish-prior} holds with high probability because $\beta > -1$ and $\|\widetilde \bX\| \to \|\cov(\pi)\| = 1$ (in probability). To see this, write $\|\widetilde \bX\| = \|\frac{1}{n} \bU\bU^\top\| = \|\frac{1}{n} \bU^\top \bU\|$; being an average of $n$ i.i.d.\ $k \times k$ matrices, $\frac{1}{n} \bU^\top \bU$ converges in probability to its expectation, which is $\cov(\pi)$.
\end{remark}

We also extend our framework to binary-valued problems. In Proposition~\ref{prop:binary-compare} we give a general result analyzing the low degree likelihood ratio in binary-valued problems via a comparison to the analogous Gaussian-valued problem. As an application, we study the following variant of the stochastic block model (SBM).

\begin{definition}
    \label{def:sbm}
    The \emph{stochastic block model} with parameters $k \ge 2$, $d > 0$, $\eta \in [-1/(k-1),1]$ (constants not depending on $n$) is specified by the following null and planted distributions over $n$-vertex graphs.
    \begin{itemize}
        \item Under $\QQ$, for every $u<v$ the edge $(u,v)$ occurs independently with probability $d/n$.
        \item Under $\PP$, each vertex is independently assigned a community label drawn uniformly from $[k]$. Conditioned on these labels, edges occur independently. If $u,v$ belong to the same community then edge $(u,v)$ occurs with probability $(1+(k-1)\eta)d/n$; otherwise $(u,v)$ occurs with probability $(1-\eta)d/n$.
    \end{itemize}
\end{definition}

\noindent Here $k$ is the number of communities, $d$ is the average degree, and $\eta$ is a signal-to-noise ratio: the planted $k$-cut cuts a fraction $\frac{k-1}{k}(1-\eta)$ of the edges on average. Known polynomial-time algorithms only succeed at distinguishing $\PP$ from $\QQ$ above the so-called Kesten--Stigum (KS) threshold, i.e., when $d \eta^2 > 1$ \cite{massoulie,mns,AS-acyclic}. We prove the following in Appendix~\ref{app:low-deg-sbm}.

\begin{theorem}
    \label{thm:sbm}
    Consider the stochastic block model as in Definition~\ref{def:sbm} with parameters $k, d, \eta$ fixed. If $d\eta^2 < 1$ then $\|L^{\le D}\| = O(1)$ for any $D = o(n/\log n)$.
\end{theorem}

\noindent Prior work \cite{HS-bayesian,sam-thesis} has already given a low-degree analysis of this variant of the SBM, showing that the problem is low-degree-hard below the KS bound. Our result offers two advantages: (i) the proof is streamlined, following easily from our general-purpose machinery, without the need for direct combinatorial calculations, and (ii) we bound $\|L^{\le D}\|$ for $D$ all the way up to $o(n/\log n)$ instead of only $n^{0.01}$. As discussed in Section~\ref{sec:low-deg-intro}, item (ii) constitutes evidence that distinguishing $\PP$ from $\QQ$ in the SBM requires fully exponential time $\exp(n^{1-o(1)})$ below the KS bound.

\section{Reduction from Detection to Certification}

In this section we give formal proofs, for both the graph and Gaussian models, that hardness of a particular detection problem implies hardness of certification.

\subsection{The Graph Model}
\label{sec:reduction-graph}

We now give the proof of Theorem~\ref{thm:graph-cert-hardness}, which shows that hardness of detection in the noisy eSBM model implies hardness of certifying bounds on max-$k$-cut.

\begin{proof}[Proof of Theorem~\ref{thm:graph-cert-hardness}]
Assume for the sake of contradiction that some algorithm $\mathcal{A}$ certifies the upper bound $\MC_k(\bG) \le \frac{k-1}{k}(1+|\eta|-\eps)$ when $\bG \sim \sG_{n,d}$. We will use this to distinguish between $\widetilde{\sG}^{\,\mathrm{eq}}_{n,d,k,\eta,\delta}$ and $\sG_{n, d}$ for $\delta = \frac{d \eps (k-1)}{5k}$; this contradicts Conjecture~\ref{conj:esbm} because the assumption $|\eta| < \frac{2\sqrt{d-1}}{d}$ implies $d < d_{\mathrm{KS}}^{\,\mathrm{eq}}(\eta)$. Our detection algorithm takes as input a graph $G$ and outputs $\textsc{q}$ (``null'') if $\mathcal{A}(G) \le \frac{k-1}{k}(1+|\eta|-\eps)$ and $\textsc{p}$ (``planted'') otherwise. If $\bG \sim \sG_{n,d}$ then $\mathcal{A}(\bG) \le \frac{k-1}{k}(1+|\eta|-\eps)$ with high probability by assumption, and so the distinguisher outputs $\textsc{q}$. Now consider the case $\bG \sim \widetilde{\sG}^{\,\mathrm{eq}}_{n,d,k,\eta,\delta}$. A graph drawn from $\sG_{n, d, k, \eta}^{\,\mathrm{eq}}$ has a planted $k$-cut of fractional size $\frac{k-1}{k}(1+|\eta|)$, and the noise operator $\bT_\delta$ can remove at most $2\delta n$ edges from this cut. Thus,
\[ \mathcal{A}(\bG) \ge \MC_k(\bG) \ge \frac{k-1}{k}(1+\eta) - \frac{2\delta n}{|E|} \]
where
\[ \frac{2\delta n}{|E|} = \frac{4\delta}{d} = \frac{4}{d} \cdot \frac{d\eps(k-1)}{5k} < \eps \cdot \frac{k-1}{k}. \]
This implies $\mathcal{A}(\bG) > \frac{k-1}{k}(1+|\eta|-\eps)$ and so the distinguisher outputs $\textsc{p}$.
\end{proof}

\subsection{The Gaussian Model}
\label{sec:reduction-gaussian}

\begin{definition}
    \label{def:pi-k}
    Let $\pi_k$ be the distribution over $\RR^k$ given by $\sqrt{k} e_{\bi} - \onesvec/\sqrt{k}$ where $\bi \sim [k]$ uniformly at random. Let $\sX_k = \sX(\pi_k)$ be the associated spike prior, as defined in Definition~\ref{def:pi}.
\end{definition}

\noindent Here, $e_1,e_2,...$ denote the standard unit basis vectors and $\onesvec$ denotes the all-ones vector.

\begin{theorem}
    \label{thm:gauss-reduction}
    Suppose there exist constants $k \ge 2$ and $\epsilon > 0$ such that there is a time-$t(n)$ algorithm to certify the upper bound $(2-\epsilon)n$ on $\Gamma_k(\bW)$ over $\bW \sim \GOE(n)$. Then for some constants $\beta \in (-1,0)$ and $\gamma > 1$, there is a time-$(t(n)+\mathrm{poly}(n))$ algorithm achieving strong detection in the general spiked Wishart model with spike prior $\sX_k$.
\end{theorem}

\noindent Note that the above parameters satisfy $\beta^2 < \gamma$, which is in the ``hard'' regime of the Wishart model. Thus, Theorem~\ref{thm:gauss-reduction} together with the low-degree-hardness of the Wishart model in that regime (Theorem~\ref{thm:general-wishart}) constitute rigorous evidence for Conjecture~\ref{conj:gauss-cert}.

\begin{proof}
    Let $\mathcal{A}$ be the purported certification algorithm. We will use this to solve detection in the Wishart model, given Wishart samples $\by_1,\ldots,\by_N$. Sample $\widetilde \bW \sim \GOE(n)$ and let $\blambda_1 \le \cdots \le \blambda_n$ denote its (random) eigenvalues. Sample a uniformly random orthonormal basis $\bv_{n-N+1},\ldots,\bv_n$ for $\bV \colonequals \mathrm{span}\{\by_1,\ldots,\by_N\}$ and a uniformly random orthonormal basis $\bv_1,\ldots,\bv_{n-N}$ for the orthogonal complement $\bV^\perp$. Let $\bW = \sum_{i=1}^n \blambda_i \bv_i \bv_i^\top$. Our Wishart detection algorithm is as follows: if $\mathcal{A}(\bW) \le (2-\epsilon)n$ then output $\textsc{q}$; otherwise, output $\textsc{p}$.
    
    We now prove that this achieves strong detection. If the Wishart samples were drawn from $\QQ$ then $\bV$ is a uniformly random $N$-dimensional subspace and so $\bW \sim \GOE(n)$. This means $\mathcal{A}(\bW) \le (2-\epsilon)n$ with high probability by assumption, and so our algorithm correctly outputs $\textsc{q}$. It remains to show that if the Wishart samples were drawn from $\PP$, then $\Gamma_k(\bW) > (2-\epsilon)n$ with high probability, and so the algorithm is forced to output $\textsc{p}$.
    
    Suppose the Wishart samples were drawn from $\PP$ with planted matrix $\bX$. With high probability we are in the first case of~\eqref{eq:wish-prior}, i.e., $\bX = \frac{1}{n} \bU\bU^\top$ where each row of $\bU$ is drawn independently from $\pi_k$. Note that $\frac{1}{k-1} \bU\bU^\top$ is a partition matrix, and so
    \[ 
        \Gamma_k(\bW) \ge -\frac{1}{k-1} \langle \bU\bU^\top, \bW \rangle. 
    \]
    We can bound
    \begin{align*}
        \langle \bU\bU^\top, \bW \rangle &= \left\langle \bU\bU^\top, \sum_{i=1}^n \blambda \bv_i \bv_i^\top \right\rangle \\
        &\le \left\langle \bU\bU^\top, \blambda_{n-N} \sum_{i=1}^{n-N} \bv_i \bv_i^\top + \blambda_n \sum_{i=n-N+1}^{n} \bv_i \bv_i^\top \right\rangle \\
        &= \left\langle \bU\bU^\top, \blambda_{n-N} \left(\id_n - \sum_{i=n-N+1}^{n} \bv_i \bv_i^\top\right) + \blambda_n \sum_{i=n-N+1}^{n} \bv_i \bv_i^\top \right\rangle \\
        &= \left\langle \bU\bU^\top, \blambda_{n-N} \id_n + (\blambda_n - \blambda_{n-N}) \sum_{i=n-N+1}^{n} \bv_i \bv_i^\top \right\rangle
    \end{align*}
    where we have used the fact that $\sum_{i=1}^n \bv_i \bv_i^\top = \id_n$ since $\{\bv_i\}_{i \in [n]}$ is an orthonormal basis. We will bound the pieces of this expression separately.
    
    First we bound $\langle \bU\bU^\top,\id_n \rangle$. The nonzero eigenvalues of $\bU\bU^\top$ are the same as the nonzero eigenvalues of $\bU^\top \bU$. Since $\bU^\top \bU$ is the sum of $n$ i.i.d.\ $k \times k$ matrices, $\frac{1}{n} \bU^\top \bU$ converges in probability to its expectation, which is $\cov(\pi) = \id_k - \onesmat_k/k$ where $\onesmat_k$ is the $k \times k$ all-ones matrix. This means $\langle \bU\bU^\top, \id_n \rangle = \Tr(\bU\bU^\top) = n \Tr(\frac{1}{n} \bU^\top \bU) = (1+o(1))n(k-1)$ with high probability.
    
    Next we bound $\langle \bU\bU^\top,\sum_{i=n-N+1}^n \bv_i \bv_i^\top \rangle$. Recall that $\{\bv_i\}_{i=n-N+1}^{n}$ is an orthonormal basis for $\mathrm{span}\{\by_1,\ldots,\by_N\}$, and so $\sum_{i=n-N+1}^n \bv_i \bv_i^\top \preceq \frac{1}{\bmmu} \bY$ where $\bY = \frac{1}{N} \sum_{i=1}^N \by_i \by_i^\top$ and $\bmmu$ is the smallest nonzero eigenvalue of $\bY$. Since $\bY$ is a spiked covariance matrix, Theorem~1.2 of \cite{BS-spiked} gives $\bmmu \to (\sqrt{\gamma}-1)^2 > 0$ in probability. Therefore,
    \[ 
        \left\langle \bU\bU^\top,\sum_{i=n-N+1}^n \bv_i \bv_i^\top \right\rangle \le \left\langle \bU\bU^\top,\frac{1}{\bmmu} \bY \right\rangle = \frac{1}{\bmmu N} \sum_{i=1}^N \|\bU^\top \by_i\|^2. 
    \]
    For fixed $U$ (and therefore fixed $X$), note that $U^\top \by_i$ follows a multivariate Gaussian distribution with mean zero and covariance
    \[ 
        \EE[U^\top \by_i \by_i^\top U] = U^\top \EE[\by_i \by_i^\top] U = U^\top (I + \beta X) U = U^\top U + \frac{\beta}{n} U^\top UU^\top U. 
    \]
    Recalling that $\frac{1}{n}\bU^\top \bU \to \id_k - \onesmat_k/k$, we have $\frac{1}{n} (\bU^\top \bU + \frac{\beta}{n} \bU^\top \bU\bU^\top \bU) \to (1+\beta)(\id_k - \onesmat_k/k)$ in probability. Thus,
    $\frac{1}{nN} \sum_{i=1}^N \|\bU^\top \by_i\|^2$ converges in probability to $(1+\beta)\Tr(\id_k-\onesmat_k/k) = (1+\beta)(k-1)$, and we can conclude that
    \[ 
        \langle \bU\bU^\top,\sum_{i=n-N+1}^n \bv_i \bv_i^\top \rangle \le (1+o(1))\frac{n}{\bmmu}(1+\beta)(k-1) = (1+o(1))n(\sqrt{\gamma}-1)^2(1+\beta)(k-1)
    \]
    with high probability.
    
    The eigenvalues of $\widetilde \bW$ converge to the Wigner semicircle law on $[-2,2]$, and so we have $\blambda_n - \blambda_{n-N} \le 4+o(1)$ with high probability. Also, by taking $\gamma > 1$ close enough to $1$, we can ensure $\lambda_{n-N} \le -2 + \epsilon/2$ with high probability.
    
    Putting it all together, we now have
    \begin{align*}
        \Gamma_k(\bW) &\ge -\frac{1}{k-1}\langle \bU\bU^\top, \bW \rangle \\
        &\ge -\frac{1}{k-1}\left\langle \bU\bU^\top, \lambda_{n-N} \id_n + (\blambda_n - \blambda_{n-N}) \sum_{i=n-N+1}^{n} \bv_i \bv_i^\top \right\rangle \\
        &\ge -\frac{1}{k-1}\left[(-2+\epsilon/2)(1+o(1))n(k-1) +  (4+o(1))n(\sqrt{\gamma}-1)^{-2}(1+\beta)(k-1)\right] \\
        &> (2-\epsilon)n
    \end{align*}
    for sufficiently large $n$, provided we choose $\beta > -1$ sufficiently close to $-1$. This completes the proof.
\end{proof}

\section{Belief Propagation and the Kesten--Stigum Transition}
\label{sec:bp}

In this section, we carry out a stability analysis of belief propagation (as discussed in Section~\ref{sec:bp-intro}) and derive the result presented in Theorem~\ref{thm:ks-main}. The analysis resembles that of the original work~\cite{sbm-ks-1,sbm-ks-2} that predicted the Kesten--Stigum threshold in the ordinary stochastic block model, but the equitability constraints create additional technical complexity in our setting. We start with the special case of the equitable coloring model in Section~\ref{sec:bp-col}, and generalize to the equitable block model in Section~\ref{sec:bp-cut}. Throughout this section, it turns out to be convenient to parametrize the equitable model in a different way (defined below) than used in the Introduction.

\subsection{The Equitable Coloring Model}
\label{sec:bp-col}

The equitable coloring model is obtained by setting $\eta = \tfrac{1}{k-1}$ in Definition \ref{def:eSBM}; for brevity let's define
\begin{equation}
    \label{eq:d}
    c \colonequals \frac{1 + \eta}{k},
\end{equation}
so that in the planted coloring, each vertex has exactly $c$ neighbors of every other color. From the point of view of the vertices, this is a complicated constraint affecting a star of $d+1$ vertices. As a result, a factor graph with a variable node for each vertex, and a constraint node corresponding to each vertex and its neighbors, is not even locally treelike. Instead, we define a variable for each edge, giving a pair of colors. The constraint then demands that the $d$ edges incident to each vertex agree on its color, and that the colors of their other endpoints are equitable.

This lets us define a message-passing algorithm. Regarding each edge $(u,v)$ of the graph $\bG \sim \sG_{n,d,k,(k-1)^{-1}}^{\text{eq}}$ as a pair of directed edges $u \to v, v\to u$, each directed edge $u \to v$ sends a message $\mu^{u \to v}$ to vertex $v$ consisting of the estimated probabilities $\mu^{u \to v}_{r,s}$ that $u$ and $v$ are color $r$ and $s$ respectively, for each $r, s \in [k]$ with $r \ne s$. Vertex $v$ then sends out messages $\mu^{v \to w}$ to the directed edges $(v,w)$ which are computed as follows:
\begin{enumerate}
\item For each of $v$'s neighbors $u$ other than $w$, choose a pair of colors $(r_u,s_u)$ independently from the distribution $\mu^{u \to v} = (\mu^{u \to v}_{rs})$.
\item Condition on the event that the $s_u$ are identical for all $u$. Call this color $s$.
\item Condition on the event that all but one of the colors other than $s$ appear $c$ times in the list $(r_u)$, and that one color appears $c-1$ times. Call this color $t$.
\item Then $\mu^{v \to w} = (\mu^{v \to w}_{st})$ is the resulting conditional distribution of the pair $(s,t)$, i.e., the probability that $v$ and $w$ are color $s$ and $t$ respectively.
\end{enumerate}
Formally we can write
\begin{equation}
\label{eq:update-mu}
\mu^{v \to w}_{st} = \frac{\psi^{v \to w}_{st}}{z^{v \to w}}
\end{equation}
where
\begin{align}
\psi^{v \to w}_{st} &= \sum_{(r_u : u \in \partial v \setminus k) \in [k]^{d-1}} 
\left( \prod_u \mu^{u \to v}_{r_u,s} \right)
\left( \prod_{u \in [k], u \ne s} \indicator{ |\{ u: r_u=u \}| = \begin{cases} c & u \ne s, t \\ c-1 & u=t \end{cases}} \right)
\label{eq:update-w} \\
z^{v \to w} &= \sum_{r,s \in [k]: r \ne s} \psi^{v \to w}_{rs} \, . 
\label{eq:update-z}
\end{align}

Clearly the uniform messages $\mu^{u \to v}_{rs} = 1/(k(k-1))$ are a fixed point of this algorithm. We want to study its stability to small perturbations, and in particular the matrix of partial derivatives
\begin{equation}
\label{eq:M1}
    Y_{rs, s't} = \frac{\partial \mu^{v \to w}_{s't}}{\partial \mu^{u \to v}_{rs}} \, . 
\end{equation}
To compute this matrix, suppose that we perturb the incoming message $\mu^{u \to v}$ for a pair $(r,s)$ with $r \ne s$,
\[
\mu^{u \to v}_{r's'} = \frac{1}{k(k-1)} (1+\eps \delta_{rr'} \delta_{ss'})
\]
where $\delta$ is the Kronecker delta, $\delta_{rr'} = 1$ if $r=r'$ and $0$ otherwise. This perturbation does not respect the normalization $\sum_{rs} \mu^{u \to v}_{rs} = 1$, but this will simply show up as $Y$ having zero row and column sums since normalization projects perturbations to the subspace perpendicular to the uniform vector.

We then have several cases. If $s' \ne s$, for all $t \ne s'$ then $\psi^{v \to w}_{s't}$ is unchanged from its value at the uniform fixed point, namely
\begin{equation}
\psi^{v \to w}_{s't} 
= \left( \frac{1}{k(k-1)} \right)^{d-1} {d-1 \choose c-1, c, \ldots, c} 
= \left( \frac{1}{k(k-1)} \right)^{d-1} \frac{c (d-1)!}{c!^{k-1}} 
\colonequals \psi \, . 
\label{eq:case1}
\end{equation}
For $s'=s$ and $t=r$, we have
\begin{align}
\psi^{v \to w}_{st} 
&= \left( \frac{1}{k(k-1)} \right)^{d-1} \left( 
(1+\eps) {d-2 \choose c-2, c, \ldots, c} + (k-2) {d-2 \choose c-1, c-1, c, \ldots c} 
\right) \nonumber \\
&= \psi \left( (1+\eps) \frac{c-1}{d-1} + (k-2) \frac{c}{d-1} \right) \nonumber \\
&= \psi \left( 1 + \eps \,\frac{c-1}{d-1} \right) \, , 
\qquad (t=r)
\label{eq:case2}
\end{align}
where the two terms in the first line come from $r_u=r$ and $r_u \ne r,s$ respectively. 
Finally, for $s'=s$ and $t \ne r$, we have
\begin{align}
\psi^{v \to w}_{st} 
&= \left( \frac{1}{k(k-1)} \right)^{d-1} \left( 
(1+\eps) {d-2 \choose c-1, c-1, c, \ldots, c} 
+ {d-2 \choose c-2, c, \ldots c} 
+ (k-3) {d-2 \choose c-1, c-1, c, \ldots c} 
\right) \nonumber \\
&= \psi \left( (k-2 + \eps) \frac{c}{d-1} + \frac{c-1}{d-1} \right) \nonumber \\
&= \psi \left( 1 + \eps \,\frac{c}{d-1} \right) \, , 
\qquad (t \ne r)
\label{eq:case3}
\end{align}
where the three terms in the first line come from $r_u=r$, $r_u=t$, and $r_u \ne r,s,t$ respectively. 

While this level of bookkeeping is comforting, both~\eqref{eq:case2} and~\eqref{eq:case3} are simply $\psi (1+\eps P)$ where $P$ is the fraction of $(d-1)$-tuples that contribute to~\eqref{eq:update-w} such that $r_u=r$. We can write~\eqref{eq:case1}, \eqref{eq:case2}, and~\eqref{eq:case3} as
\begin{equation}
\label{eq:w}
\psi^{v \to w}_{s't} = \psi \left( 1 + \eps \delta_{ss'} \frac{c - \delta_{rt}}{d-1} \right).
\end{equation}
Summing over all distinct $s', t$ gives the normalization factor 
\begin{align}
z^{v \to w} 
&= k(k-1) \psi + \eps \psi \left( \frac{c-1}{d-1} + (k-2) \frac{c}{d-1} \right) \nonumber \\
&= k(k-1) \psi + \eps \psi \nonumber \\
&= k(k-1) \psi \left( 1 + \frac{\eps}{k(k-1)} \right) \, ,
\label{eq:z}
\end{align}
with the multiplicative factor $1+\eps P$ where $P=1/(k(k-1))$ is now the probability that a random edge $u \to v$ has colors $r$ and $s$ on its endpoints.

Combining~\eqref{eq:z} with~\eqref{eq:w}, and~\eqref{eq:update-mu} gives
\begin{equation}
\label{eq:update-eps}
\mu^{v \to w}_{s't} 
= \frac{1}{k(k-1)} 
\left( 1 + \eps \left( - \frac{\delta_{ss'} \delta_{rt}}{d-1} + \frac{c \delta_{ss'}}{d-1} - \frac{1}{k(k-1)}  \right) + O(\eps^2) \right) 
\, .
\end{equation}
Canceling the factor $1/(k(k-1))$ gives the matrix of partial derivatives~\eqref{eq:M1},
\begin{equation}
\label{eq:M}
Y_{rs, s't} 
= - \frac{\delta_{ss'} \delta_{rt}}{d-1} + \frac{c \delta_{ss'}}{d-1} - \frac{1}{k(k-1)} \, . 
\end{equation}
Using~\eqref{eq:M} and~\eqref{eq:d} the reader can check that the rows and columns of $Y$ sum to zero, as alluded to above:
\begin{equation}
\label{eq:zero-sum}
\forall r,s: \sum_{\substack{s',t: \\ s' \ne t}} Y_{rs,s't} 
= 0 \, , 
\quad 
\forall s',t: \sum_{\substack{r,s: \\ r \ne s}} Y_{rs,s't} 
= 0 \, . 
\end{equation}

To diagonalize $Y$, it is useful to treat the $k(k-1)$-dimensional space $\mathcal{U}$ spanned by ordered pairs $(r,s)$ with $r \ne s$ as the space of $k \times k$ matrices $U=(U_{rs})$ with zeroes on the diagonal. Then we can interpret the three terms in~\eqref{eq:M} as follows:
\begin{itemize}
\item The term $\delta_{ss'} \delta_{rt}$ is the transpose operator, sending $U$ to $U^\top$. 
\item The term $\delta_{ss'}$ sends $U$ to $(U J)^\top$ where $J$ is the all-1s matrix.
\item The term $-1/k(k-1)$ subtracts the mean entry of $U$ from each entry of $Y(U)$.
\item Finally, we set all the diagonal elements of $Y(U)$ to zero.
\end{itemize}
Thus we can rewrite~\eqref{eq:M} as
\begin{equation}
\label{eq:Mu}
Y(U) = \Pi\left[ - \frac{1}{d-1} U^\top + \frac{c}{d-1} \,(U J)^\top - \frac{1}{k(k-1)} \,J U J \right]
\end{equation}
where $\Pi$ is the projection operator that sets the diagonal entries of a matrix to zero.

To diagonalize $Y$, recall that if two linear operators commute, they share the same eigenvectors. Clearly $Y$ commutes with relabelings of the colors, i.e., with the $S_k$-action that conjugates $U$ with a permutation matrix. This action preserves the following subspaces of matrices:
\begin{itemize}
\item The symmetric matrices with zero diagonal
\item The antisymmetric matrices
\item The matrices whose row (resp.\ column) sums are zero
\item The matrices whose rows (resp.\ columns) are uniform, other than being zero on the diagonal 
\end{itemize}
\ldots and their intersections. More abstractly, $\mathcal{U}$ is the $k(k-1)$-dimensional combinatorial representation where $S_k$ acts on distinct ordered pairs $(r,s)$ by sending $(r,s)$ to $(\pi(r), \pi(s))$. We can find the eigenvectors and eigenvalues of $Y$ by decomposing $\mathcal{U}$ into a direct sum of irreducible representations of $S_k$. This decomposition includes one copy of the trivial representation $\rho_{(k)} = \id$, and one copy each of $\rho_{(k-2,1,1)}$ and $\rho_{(k-2,2)}$. (To avoid some case-checking we assume that $k \ge 4$. In particular, if $k = 3$ then $\rho_{(k-2,2)}$ disappears.) By Schur's Lemma, when restricted to each of these irreducible subspaces $Y$ is a scalar matrix with a single eigenvalue. These are as follows:
\begin{itemize}
\item The trivial representation is spanned by the matrix with $1$s everywhere off the diagonal. By~\eqref{eq:zero-sum} this has eigenvalue zero. 
\item The copy of $\rho_{(k-2,1,1)}$ consists of antisymmetric matrices with zero row and column sums. These are annihilated by $v$ and are eigenvectors of the transpose with eigenvalue $-1$. Thus they are eigenvectors of $Y$ with eigenvalue $+1/(d-1)$. This eigenspace has dimension $(k-1)(k-2)/2$. 
\item The copy of $\rho_{(k-2,2)}$ consists of symmetric matrices with zero row and column sums and zeroes on the diagonal. These are annihilated by $v$ and are eigenvectors of the transpose with eigenvalue $+1$. Thus they are eigenvectors of $Y$ with eigenvalue $-1/(d-1)$. This eigenspace has dimension $(k-1)(k-2)/2 - 1 = k(k-3)/2$. 
\end{itemize}

However, we are not done. In addition to these multiplicity-free irreducible representations, $\mathcal{U}$ includes two copies of the ``standard'' representation $\rho_{(k-1,1)}$, one each in the symmetric and antisymmetric subspace. Each one has dimension $k-1$, and together they span an isotypic subspace of dimension $2(k-1)$. By Schur's lemma, when restricted to this subspace, $Y$ is the tensor product of the identity with a $2 \times 2$ matrix, giving it two additional eigenvalues. One of these will turn out to be the dominant one and will control where the Kesten--Stigum transition occurs.

This isotypic subspace is spanned by matrices of the form
\begin{equation}
\label{eq:standard-isotypic}
U_{ij} = \begin{cases} 
0 & i = j \\
\alpha & i=1, j \ne 1 \\
\beta & j=1, i \ne 1 \\
\gamma & i \ne 1, j \ne 1, i \ne j
\end{cases} 
\quad \text{where} \quad \gamma = - \frac{\alpha+\beta}{k-2} \, , 
\end{equation}
and their images under conjugation by permutation matrices, i.e., where the ``special'' row and column ranges from $1$ to $k$.
That is, 
\[
U = \begin{pmatrix}
0 & \alpha & \alpha & \cdots & \alpha \\
\beta & 0 & \gamma & \cdots & \gamma \\
\beta & \gamma & 0 & & \vdots \\
\vdots & \vdots & & \ddots & \gamma \\
\beta & \gamma & \cdots & \gamma & 0
\end{pmatrix}
\]
where $\gamma$ is set so that $U$'s entries sum to zero. The reader can check that these matrices are orthogonal to the other irreducible subspaces with respect to the trace inner product $\langle U, U' \rangle = \tr\ U^\top U'$.

Using~\eqref{eq:Mu}, we find that $Y(U)$ is also of this form but with entries $\alpha'$ and $\beta'$, where
\begin{equation}
\label{eq:alphabeta}
\begin{pmatrix} \alpha' \\ \beta' \end{pmatrix} 
= m \cdot 
\begin{pmatrix} \alpha \\ \beta \end{pmatrix} 
\quad \text{where} \quad 
m = \begin{pmatrix}
0 & 1 \\
\frac{-1}{d-1} & \frac{-c}{d-1} 
\end{pmatrix} \, . 
\end{equation}
Thus $Y$ on this isotypic subspace is $m \otimes \id$ where $\id$ is the $(k-1)$-dimensional identity. The corresponding eigenvalues of $Y$ are those of $m$, namely the roots $\kappa$ of 
\[
(d - 1) \kappa^2 + c \kappa + 1 = 0 \, ,
\]
which are
\[
\kappa_\pm = \frac{-c \pm \sqrt{c^2 - 4(d-1)}}{2 (d-1)} \, . 
\]
When $c^2 < 4(d-1)$ the discriminant is negative, so these eigenvalues are complex and lie on the unit circle $|\kappa_\pm| = 1/\sqrt{d-1}$. But when $c^2 > 4(d-1)$ they are real, and $\kappa_+ > 1/\sqrt{d-1}$. 

The full Jacobian of belief propagation is the tensor product of this local matrix $Y$ with the non-backtracking matrix $B$. It is a consequence of Theorem 2 in \cite{bordenave2019eigenvalues} that, with high probability over $\bG$ sampled from the equitable SBM with $d < d_{\text{KS}}^{\text{eq}}$, the spectrum of the non-backtracking matrix consists of a ``trivial'' eigenvalue $d$ whose left and right eigenvectors is the all-ones vector, and remaining eigenvalues with modulus at most $\sqrt{d-1} + o_n(1)$ in the complex plane. Since perturbations along the uniform eigenvector of $B$ would violate the balance of colors, we are left with these remaining eigenavlues. Multiplying $\kappa_+$ by $\sqrt{d-1} + o_n(1)$ tells us that the Kesten--Stigum transition, where the largest eigenvalue of the Jacobian exceeds $1$ in absolute value, occurs when $c^2\approx4(d-1)$, with $\approx$ hiding $o_n(1)$ terms. Given~\eqref{eq:d} this is 
\[
    c \approx 2 \left( (k-1)+\sqrt{k(k-2)} \right)
\]
or 
\[
    d 
    \approx 2 (k-1)^2 \left( 1 + \sqrt{ \frac{k(k-2)}{(k-1)^2} } \right)
    \approx \frac{2}{\lambda^2} \left( 1 + \sqrt{ 1-\lambda^2 } \right)
\]
where $\lambda = -1/(k-1)$. Solving this condition for $\lambda$ in terms of $d$ gives
\begin{equation}
\label{eq:ks-thresh}
| \lambda | = \frac{2\sqrt{d-1}}{d} + o_n(1) \, . 
\end{equation}

\subsection{Generalizing to the Equitable Block Model}
\label{sec:bp-cut}

Let 
\begin{equation}
\label{eq:d-sbm}
d = a + b(k-1) \, . 
\end{equation}
Recall that a legal labeling of the \emph{equitable block model} on a $d$-regular graph is a $k$-coloring where each vertex has exactly $a$ neighbors of its own color, and exactly $b$ neighbors of each of the $k-1$ other colors. The case of equitable $k$-colorings corresponds to $a=0$ and $b=c$.

We can generalize the message-passing algorithm of the previous section as follows. 
Each directed edge $(u,v)$ again sends a message $\mu^{u \to v}$ to vertex $v$ consisting of the estimated probability $\mu^{u \to v}_{rs}$ that $u$ and $v$ are color $r$ and $s$ respectively, for each $r, s \in [k]$ with $r \ne s$. Vertex $v$ then sends out messages $\mu^{v \to w}$ to the directed edges $(v,w)$ which are computed as follows:
\begin{enumerate}
\item Give pairs $(s,t)$ the prior distribution
\begin{equation}
\label{eq:sbm-prior}
P(s,t) = \frac{1}{kd} \begin{cases} 
a & (s=t) \\
b & (s \ne t) \, . 
\end{cases}
\end{equation}
\item Multiplicatively reweight this distribution by the probability that, if for each of $v$'s neighbors $u$ other than $w$ we choose a pair of colors $(r_u,s_u)$ independently from the distribution $\mu^{u \to v} = (\mu^{u \to v}_{rs})$, then:
\begin{itemize}
\item $s_u=s$ for all $u$, and
\item (if $s=t$) $s$ appears $a-1$ times in the list $(r_u)$, and every color other than $s$ appears $b$ times, 
\item (if $s \ne t$) $s$ appears $a$ times in the list $(r_u)$, $t$ appears $b-1$ times, and all other colors appear $b$ times each.
\end{itemize}
\item Then $\mu^{v \to w} = (\mu^{v \to w}_{st})$ is the resulting posterior distribution of $(s,t)$.
\end{enumerate}

Generalizing~\eqref{eq:update-w} and~\eqref{eq:update-z}, we can write
\begin{equation}
\label{eq:update-mu-sbm}
\mu^{v \to w}_{st} = \frac{\psi^{v \to w}_{st}}{z^{v \to w}}
\end{equation}
where
\begin{align}
\psi^{v \to w}_{st} &= \sum_{(r_u : u \in \partial v \setminus k) \in [k]^{d-1}} 
\left( \prod_u \mu^{u \to v}_{r_u,s} \right)
\begin{cases}
\displaystyle{ 
\,a \prod_{u \in [k]} \indicator{ |\{ u: r_u=u \}| 
= \begin{cases} 
a-1 & u = s \\ b & u \ne s  
\end{cases} } } & (s=t)
\\
\displaystyle{
\,b \prod_{u \in [k]} \indicator{ |\{ u: r_u=u \}| = \begin{cases} 
a & u = s \\ 
b-1 & u = t \\ 
b & u \ne s, t 
\end{cases} 
} } & (s \ne t) 
\end{cases}
\label{eq:update-w-sbm} \\
z^{v \to w} &= \sum_{r,s \in [k]: r \ne s} \psi^{v \to w}_{rs} \, . 
\label{eq:update-z-sbm}
\end{align}
Note that the factors of $a$ and $b$ in~\eqref{eq:update-w-sbm}, which did not appear in the coloring case, come from the prior distribution~\eqref{eq:sbm-prior}. The reader can check that the prior distribution on $(s,t)$ is now a fixed point of this algorithm, 
\[
\mu^{u \to v}_{st} = \frac{1}{kd} \begin{cases} 
a & (s=t) \\
b & (s \ne t) \, .
\end{cases}
\]
At this fixed point, \eqref{eq:update-w-sbm} gives
\begin{equation}
\label{eq:w-fixed}
\psi^{v \to w}_{st} = \begin{cases} 
a \psi & (s=t) \\
b \psi & (s \ne t) 
\end{cases}
\end{equation}
where
\[
\psi 
\colonequals \frac{a^{a-1} b^{b(k-1)}}{(kd)^{d-1}} {d-1 \choose a-1, b, \ldots, b}
= \frac{a^a b^{b(k-1)-1}}{(kd)^{d-1}} {d-1 \choose a, b-1, b, \ldots, b}
= \frac{a^a b^{b(k-1)}}{(kd)^{d-1}} \frac{(d-1)!}{a! b!^{k-1}} \, ,
\]
and where as always $0^0=0!=1$. 

We again consider perturbing the incoming message $\mu^{u \to v}$ for a pair $(r,s)$, 
\begin{equation}
\label{eq:perturb-sbm}
\mu^{u \to v}_{r's'} = (1+\eps \delta_{rr'} \delta_{ss'}) \frac{1}{kd} \begin{cases} 
a & (r=s) \\
b & (r \ne s) \, .
\end{cases}
\end{equation}
As before, if $s' \ne s$ then $\psi^{v \to w}_{s't}$ is unchanged for all $t$, and is still given by~\eqref{eq:w-fixed}.  For the other cases where $s'=s$, let us first assume that $r \ne s$. There are now three cases: $t=s$ (now that some neighbors have the same color), $t=r$, and $t \notin \{r, s\}$. We have
\begin{align}
\psi^{v \to w}_{st} 
&= 
\begin{cases} 
a \psi \left( 1 + \eps \,\frac{b}{d-1} \right) & t=s \\
b \psi \left( 1 + \eps \,\frac{b-1}{d-1} \right) & t=r \\ 
b \psi \left( 1 + \eps \,\frac{b}{d-1} \right) & t \notin \{r,s\} 
\end{cases} \qquad (r \ne s) 
\label{eq:w-r-neq-s-sbm}
\end{align}
where in each case we multiply by $1+\eps P$ where $P$ is the fraction of $(d-1)$-tuples that contribute to~\eqref{eq:update-w-sbm} where $r_u=r$. Summing over all $s', t$ gives
\begin{align}
z^{v \to w} 
&= kd \psi + \eps b \psi \left( \frac{b-1}{d-1} + (k-2) \frac{b}{d-1} + \frac{a}{d-1} \right) \nonumber \\
&= kd \psi + \eps b \psi \nonumber \\
&= kd \psi \left( 1 + \eps \frac{b}{kd} \right) 
\qquad (r \ne s) \, . 
\label{eq:z-r-neq-s-sbm}
\end{align}
This multiplicative factor is $1+\eps P$ where $P=b/(kd)$ is the prior probability that a random edge $u \to v$ has colors $r, s$ on its endpoints where $r \ne s$.

Analogously, if $r=s=s'$ we have 
\begin{align}
\psi^{v \to w}_{st} 
&= 
\begin{cases} 
a \psi \left( 1 + \eps \,\frac{a-1}{d-1} \right) & t=s \\
b \psi \left( 1 + \eps \,\frac{a}{d-1} \right) & t \ne s 
\end{cases} \qquad (r = s) \, ,
\label{eq:w-r-eq-s-sbm}
\end{align}
and
\begin{align}
z^{v \to w} 
&= kd \psi + \eps a \psi \left( \frac{a-1}{d-1} + (k-1) \frac{b}{d-1} \right) \nonumber \\
&= kd \psi + \eps a \psi \nonumber \\
&= kd \psi \left( 1 + \eps \frac{a}{kd} \right) 
\qquad (r = s) \, ,
\label{eq:z-r-eq-s-sbm}
\end{align}
where the multiplicative factors are again $1+\eps P$ where $P$ is the fraction of $(d-1)$-tuples contributing to~\eqref{eq:update-w-sbm}, or in~\eqref{eq:z-r-eq-s-sbm} the prior probability $P=a/(kd)$ of an edge having colors $r=s$ on its endpoints.

Putting all this together generalizes~\eqref{eq:update-eps} to
\begin{equation}
\mu^{v \to w}_{s't} 
= \frac{1}{kd} 
\begin{cases}
a \left( 1 + \eps \left( \delta_{ss'} \frac{a-1}{d-1} - \frac{a}{kd} \right) + O(\eps^2) \right) 
& [s'=t, r=s] \\
a \left( 1 + \eps \left( \delta_{ss'} \frac{b}{d-1} - \frac{b}{kd} \right) + O(\eps^2) \right) 
& [s'=t, r \ne s] \\ 
b \left( 1 + \eps \left( \delta_{ss'} \frac{a}{d-1} - \frac{a}{kd} \right) + O(\eps^2) \right) 
& [s' \ne t, r=s] \\
b \left( 1 + \eps \left( \delta_{ss'} \frac{b - \delta_{rt}}{d-1} - \frac{b}{kd} \right) + O(\eps^2) \right) 
& [s' \ne t, r \ne s] \, , 
\end{cases}
\label{eq:update-eps-sbm}
\end{equation}
the fourth case of which coincides with~\eqref{eq:update-eps} when $a=0$ and $b=c$. Comparing with~\eqref{eq:perturb-sbm} and accounting for factors of $a$ and $b$ gives the matrix of partial derivatives, 
\begin{equation}
\label{eq:M-sbm}
Y_{rs, s't} 
= -\frac{\delta_{ss'} \delta_{rt}}{d-1} 
+ \left( \frac{\delta_{ss'}}{d-1} - \frac{1}{kd} \right) \begin{cases} 
a & (s'=t) \\
b & (s' \ne t) 
\end{cases}
\, . 
\end{equation}
The reader can check the normalization conditions: the rows of $Y$ sum to zero, so that the uniform vector is a right eigenvector of eigenvalue zero, but the columns are orthogonal to the prior distribution~\eqref{eq:sbm-prior} so that it is a left eigenvector of eigenvalue zero. Thus~\eqref{eq:zero-sum} becomes
\begin{equation}
\label{eq:zero-sum-sbm}
\forall r,s: \sum_{\substack{s',t}} Y_{rs,s't} 
= 0 \, , \quad
\forall s', t: \sum_{\substack{r,s}} Y_{rs,s't} 
\begin{cases} a & (r=s) \\ b & (r \ne s) \end{cases}
= 0 \, .
\end{equation}

At the risk of multiplying entities without necessity, we can also write $Y$ in the style of~\eqref{eq:Mu}. If we think of $Y$'s action by right multiplication on $k^2$-dimensional vectors as a linear operator on $k$-dimensional matrices $U=(U_{s't})$, then 
\begin{equation}
\label{eq:Mu-sbm}
Y(U) = -\frac{1}{d-1} U^\top + \frac{1}{d-1} (\Upsilon(U) J)^\top - \frac{1}{kd} J \Upsilon(U) J \, ,
\end{equation}
where $J$ is again the all-$1$s matrix and $\Upsilon$ is a linear operator on matrices that reweights diagonal and off-diagonal elements by $a$ and $b$ respectively, 
\begin{equation}
\label{eq:ups}
\Upsilon(U)_{s't} = U_{s't} \begin{cases} 
a & (s'=t) \\
b & (s' \ne t) 
\end{cases} \, . 
\end{equation}

As in the coloring case, we use representation theory to diagonalize $Y$. The $k$-dimensional matrices form a $k^2$-dimensional representation of $S_k$ where permutation matrices act by conjugation. Since this representation sends pairs of colors $(r,s)$ to $(\pi(r), \pi(s))$, this is the tensor product of the natural permutation representation with itself. The permutation representation is a direct sum of the trivial representation (spanned by the uniform vector) with the $(k-1)$-dimensional standard representation (spanned by vectors that sum to zero). Taking its tensor square and decomposing gives the representations described above, as well as the subspace spanned by diagonal matrices, giving one additional copy each of the trivial representation and the standard representation. 

Thus in total we have a two-dimensional trivial subspace, one copy each of $\rho_{(k-2,1,1)}$ and $\rho_{(k-2,2)}$ with dimension $(k-1)(k-2)/2$ and $k(k-3)/2$ respectively, and a $3(k-1)$-dimensional subspace consisting of three copies of the standard representation $\rho_{(k-1,1)}$. We go through each of these subspaces, focusing on $Y$'s right eigenvectors.

First, the trivial subspace is spanned by the identity matrix $\id=(\delta_{s't})$ and the all-$1$s matrix $J$. As in~\eqref{eq:zero-sum-sbm} $J$ is a right eigenvector with eigenvalue zero, so $Y(J)=0$. Observing~\eqref{eq:Mu-sbm}, we have $\id^\top = \id$, $\Upsilon(\id) = a \id$, $\id \onesmat=\onesmat$, and $\onesmat \id \onesmat =k \onesmat$. This gives
\[
Y(\id) 
= -\frac{1}{d-1} \id + \left( \frac{a}{d-1} - \frac{a}{d} \right) \onesmat
= -\frac{1}{d-1} \id + \frac{a}{d(d-1)} \onesmat.
\]
Thus in this two-dimensional subspace $Y$ acts as the matrix 
\[
\frac{1}{d-1}
\begin{pmatrix}
-1 & 0 \\
a/d & 0 
\end{pmatrix}
\]
giving the eigenvalues $-1/(d-1)$ and $0$.

Next, as before the copy of $\rho_{(k-2,1,1)}$ consists of antisymmetric matrices $U$ with zero row and column sums. For these matrices we have $\Upsilon(U) = bU$ and $UJ=0$, while $U^\top = -U$. Thus they are again eigenvectors of $Y$ with eigenvalue $+1/(d-1)$.

The copy of $\rho_{(k-2,2)}$ consists of symmetric matrices $U$ with zero row and column sums and zeroes on the diagonal. Now we have $\Upsilon(U) = bU$, $UJ=0$, and $U^\top = U$, and~\eqref{eq:Mu-sbm} again makes them eigenvectors with eigenvalue $-1/(d-1)$. 

This leaves the three copies of the standard representation. This isotypic subspace is spanned by matrices like those in~\eqref{eq:standard-isotypic} but with nonzero diagonal entries, namely
\begin{equation}
U_{ij} = \begin{cases} 
\delta & i=j=1 \\
\zeta & i=j \ne 1 \\
\alpha & i=1, j \ne 1 \\
\beta & j=1, i \ne 1 \\
\gamma & i \ne 1, j \ne 1, i \ne j
\end{cases} 
\quad \text{where} 
\quad 
\gamma = - \frac{\alpha+\beta}{k-2} 
\quad \text{and} \quad
\zeta = -\frac{\delta}{k-1} \, ,
\label{eq:standard-isotypic-sbm}
\end{equation}
and their images under conjugation by permutation matrices, i.e., where the ``special'' row and column ranges from $1$ to $k$.
That is, 
\[
U = \begin{pmatrix}
\delta & \alpha & \alpha & \cdots & \alpha \\
\beta & \zeta & \gamma & \cdots & \gamma \\
\beta & \gamma & \zeta & & \vdots \\
\vdots & \vdots & & \ddots & \gamma \\
\beta & \gamma & \cdots & \gamma & \zeta
\end{pmatrix}
\]
where $\gamma$ and $\zeta$ are set so that $U$'s entries sum to zero and $U$ has zero trace. In particular, $U$ is orthogonal to both $\onesmat$ and $\id$, and hence to the trivial subspace. The reader can confirm that It is orthogonal to $\rho_{(k-2,1,1)}$ and $\rho_{(k-2,2)}$ as well.

Using~\eqref{eq:Mu-sbm} and a little work, we find that $Y(U)$ is also of this form but with entries $\alpha', \beta', \delta'$, where
\begin{equation}
\label{eq:alphabetadelta}
\begin{pmatrix} \alpha' \\ \beta' \\ \delta' \end{pmatrix} 
= m \cdot 
\begin{pmatrix} \alpha \\ \beta \\ \delta \end{pmatrix} 
\quad \text{where} \quad 
m = \frac{1}{d-1} \begin{pmatrix}
-b & -1 & \frac{-a}{k-1} \\
b(k-1)-1 & 0 & a \\
b(k-1) & 0 & a-1
\end{pmatrix} \, . 
\end{equation}
Thus $Y$ on this isotypic subspace is $m \otimes \id$ where $\id$ is the $(k-1)$-dimensional identity. The corresponding eigenvalues of $Y$ are those of $m$, which are namely the roots $\kappa$ of 
\[
-\frac{1}{d-1} 
\quad \text{and} \quad
\kappa_\pm = \frac{a - b \pm \sqrt{ (a - b)^2 - 4 (d - 1) }}{2 (d-1)} \, . 
\]
Analogous with the coloring case, if $(a-b)^2 < 4(d-1)$ these eigenvalues are complex and lie on the unit circle $|\kappa_\pm| = 1/\sqrt{d-1}$. But when $(a-b)^2 > 4(d-1)$, they are real, and $\kappa_+ > 1/\sqrt{d-1}$.  

We again multiply $\kappa_+$ by the modulus of the largest non-trivial eigenvalue of the non-backtracking matrix, $\sqrt{d-1} + o_n(1)$, to obtain the dominant eigenvalue of the Jacobian of belief propagation. The Kesten--Stigum transition occurs when this eigenvalue exceeds the unit circle, or when 
\[
    (a-b)^2 = 4(d-1) + o_n(1)\, . 
\]
Since in the equitable block model we have
\[
    \lambda = \frac{a-b}{d} \, ,
\]
this again occurs at 
\[
    | \lambda | = \frac{2 \sqrt{d-1}}{d} + o_n(1)\, . 
\]

\section{Local Statistics}
\label{sec:local-stats}

Throughout this section, we will for the sake of brevity write $\QQ = (\QQ_n)$ for the uniform distribution $\sG_{n,d}$ on $d$-regular graphs, and $\PP = (\PP_n)$ for the equitable stochastic block model $\sG_{n,k,d,\eta}^{\text{eq}}$ from Definition \ref{def:eSBM}. 

As in the preceding text, we are most concerned with the behavior of the null and planted models when the number of vertices is very large, and we will write \textit{with high probability (w.h.p.)}\ to describe a sequence of events that hold with probability $1 - o_n(1)$ in $\PP_n$ or $\QQ_n$ as $n\to \infty$, with other parameters ($d,k,\eta$) held fixed. The constant in the $o_n(1)$ may depend on these other parameters, and we will not make any attempt to quantify its rate, leaving us free to take union bounds over constantly many events.

In this section we study a family of semidefinite programming algorithms for the $\PP$ vs.\ $\QQ$ distinguishing problem. Like Sum of Squares, the \emph{Local Statistics algorithm} is phrased in the language of polynomials. Let us define a set of variables $x = \{x_{u,i}\}$ indexed by vertices $u \in [n]$ and group labels $i \in [k]$, and $G = \{G_{u,v}\}$ indexed by pairs of distinct vertices. We think of the planted model as outputting a random evaluation of these variables, namely a pair $(\bx,\bG)$, where $\bx \in \R^{n\times k}$ encodes the hidden community structure---with $\bx_{u,i} = 1$ if $\bsigma(u) = i$ and zero otherwise---and $\bG \in \R^{[n]\choose 2}$ is the Boolean vector indicating which edges are present in the graph. This allows us to regard polynomials $p \in \RR[x,G]$ as statistics of the planted distribution $\PP$, and we will in particular focus on the quantities $\EE_{(\bx,\bG)\sim \PP} \, p(\bx,\bG)$. 

The planted model outputs random variables $\bx$ and $\bG$ with a particular combinatorial structure: each variable is $\{0,1\}$-valued, and each vertex has exactly one label. This can be encoded in a set of polynomial constraints:
\begin{align*}
    G^2_{u,v} - G_{u,v} &= 0 & & \forall (u,v) \in {[n] \choose 2} \\
    x^2_{u,i} - x_{u,i} &= 0 & & \forall u \in [n], i \in [k] \\
    \sum_{i\in [k]}x_{u,i} - 1 &= 0 & & \forall u \in [n].
\end{align*}
Calling $\mathcal{I}_k$ the ideal of $\RR[x,G]$ generated by the polynomials on the left hand side of the equations above, then for any $p \in \mathcal{I}_k$, $p(\bx,\bG) = 0$. Moreover, the planted distribution has a pleasant symmetry property: the symmetric group $S_n$ acts naturally and simultaneously on the variables $x$ and $G$, with a permutation $\xi$ acting as $x_{u,i} \mapsto x_{\xi(u),i}$ and $G_{u,v} \mapsto G_{\xi(u),\xi(v)}$, and for any polynomial $p \in \RR[x,G]$, the expectation $\EE_{(\bx,\bG)\sim\PP} p(\bx,\bG)$ is constant on the orbits of this action. The Local Statistics algorithm, given as input a graph $G_0$, endeavors to find a ``pseudoexpectation'' that mimics the conditional expectation $\EE[ \cdot | G_0]$ on polynomials $p(x,G_0)$ of sufficiently low degree. 

\begin{definition}[Local Statistics Algorithm with Informal Moment Constraints]
    \label{def:local-stats-informal}
    The \emph{degree-$(D_x,D_G)$ Local Statistics algorithm} is the following SDP: given an input graph $G_0$, find $\pseudo : \RR[x]_{\le D_x} \to \R$ s.t.
    \begin{enumerate}
        \item (Positivity) $\pseudo p(x)^2 \ge 0$ whenever $\deg p^2 \le D_x$ 
        \item (Hard Constraints) $\pseudo p(x,G_0) = 0$ for every $p \in \mathcal{I}_k$
        \item (Moment Constraints) $\pseudo p(x,G_0) \approx \Ex_{(\bx,\bG)\sim\PP} p(\bx,\bG) $ whenever $\deg_G p(x,G) \le D_G$, $\deg_x p(x,G) \le D_x$, and $p$ is fixed under the $S_n$ action.
    \end{enumerate}
\end{definition}
\noindent
See, e.g., the survey \cite{Laurent-2009-SOS} for detailed discussion of how optimization problems of this form can be solved as SDPs.

We use the Local Statistics SDP for the $\PP$ vs.\ $\QQ$ hypothesis testing problem as follows: given $\bG$ sampled from one of these two distributions, we run Local Statistics, outputting $\textsc{p}$ if the SDP is feasible, and $\textsc{q}$ otherwise. The symbol $\approx$ in the moment constraints indicates that we will permit some additive error; this is necessary so that, when $\bG \sim \PP$, the SDP is with high probability satisfiable, by setting $\pseudo p(x,G_0) = \EE p(\bx,\bG)$. In fact, we will instantiate these moment constraints only on the elements of a certain combinatorially meaningful basis, and we will allow different additive error for different basis elements. In so doing we automatically satisfy positivity and the hard constraints, and the additive slack allows for for fluctuations of the $p(\bx,\bG)$ around their expectations. When we make this precise below, we will write this additive slack in terms of an `error tolerance' $\delta > 0$. 

\begin{theorem}\label{thm::LS:equitable-model}
    Let $\QQ$ and $\PP$ be as in Definition \ref{def:eSBM}. If $(d\eta)^2 > 4(d-1)$ then there exists $\delta > 0$ so that the degree $(2,2)$ Local Statistics algorithm with error tolerance $\delta$ can distinguish $\PP$ and $\QQ$. If $(d\eta)^2 \le 4(d-1)$, then there do not exist such a $D$ and $\delta$.
\end{theorem}

The proof of Theorem \ref{thm::LS:equitable-model} closely follows \cite{banks2019local}. As in that work, we will first study a simpler SDP hierarchy for the $\PP$ vs.\ $\QQ$ hypothesis testing problem, and then show that its feasiblity is equivalent to that of degree $(2,D)$ local statistics SDP.  We begin with some standard facts about non-backtracking walks, which will be a central tool in our analysis.

\subsection{Non-Backtracking Walks}
\label{sec:nb-walks}

Let $A_G$ be the adjacency matrix for a $d$-regular graph $G$ (which may have self-loops and multi-edges). A length-$s$ \textit{non-backtracking walk} on $G$ is an alternating sequence of vertices and edges $v_1, e_1,v_2,e_2,...,v_s$ without terms of the form $v,e,w,e,v$. The matrices $\nb{s}{G}$ whose $u,v$ entries count the number of such walks between vertices $u$ and $v$ are given by
\begin{align*}
    \nb{0}{G} &= 1 \\
    \nb{1}{G} &= A_G \\ 
    \nb{2}{G} &= A^2_G - d \\
    \nb{s+1}{G} &= A\nb{s}{G} - (d-1)\nb{s-1}{G} \qquad s \ge 2.
\end{align*}
In particular, $\nb{s}{G} = q_s(A_G)$, for a sequence of monic univariate polynomials $q_s \in \RR[z]$, with $\deg q_s = s$, which are known to be orthogonal with respect to the Kesten-McKay measure
$$
    d\mkm(z) \colonequals \frac{d}{2\pi} \frac{\sqrt{4(d-1) - z^2}}{d^2 - z^2} \indicator{|z| < \lam} dz
$$ 
on the interval $(-\lam,\lam)$. This fact has appeared innumerable times in the literature, dating back at least to \cite{mckay1981expected}. Thus the polynomials $q_s$ are a basis for the Hilbert space of square integrable functions on this interval, equipped with the inner product
$$
    \langle f, g \rangle_{\km} \colonequals \int f(z)g(z) d\mkm(z),
$$
and associated norm
$$
    \|f\|_{\km}^2 \colonequals \langle f, f \rangle_{\km} = \int f(z)^2 d\mkm(z),
$$
and in particular for any polynomial $f \in \RR[z]$, we have the orthogonal decomposition
$$
    f = \sum_{s \ge 0} \frac{\langle f, q_s \rangle_{\km}}{\|q_s\|^2_{\km}}q_s.
$$

We record for later use that
$$
    \|q_s\|_{\km}^2 = q_s(d) = \begin{cases} 1 & s=1 \\ d(d-1)^{s-1} & s > 1 \end{cases};
$$
this is equal to the number of vertices at depth $s$ in a rooted $d$-regular tree, or equivalently $n^{-1}$ times the total number of length-$s$ non-backtracking walks in a $d$-regular graph on $n$ vertices. For a derivation of this and other related facts, the reader may refer to \cite{sole1996spectra} or \cite{sodin2007random}, but should beware of differing normalization conventions.

We will also require some standard and generic properties sequences of univariate polynomials orthogonal with respect to a measure on an interval of $\RR$ \cite[Theorems 3.3.1, 6.6.1, and 3.4.1-2]{szeg1939orthogonal}: each $q_s$ has $s$ roots in the interval $(-\lam,\lam)$, the union of these roots over all $s \ge 0$ are dense in this interval, and we have a nonnegative quadrature rule

\begin{lemma}[Quadrature] \label{lem:quadrature}
    For each $s$, call $r_1 < r_2 < \cdots < r_s$ there roots of $q_s$. There exist weights $w_1,...,w_s \ge 0$ with the property that
    $$
        \langle f, q_s \rangle_{\km} = \sum_{i \in [s]} f(r_i) q_s(r_i) w_i
    $$
    for every polynomial $f$ of degree at most $2s - 1$.
\end{lemma}

\subsection{The Symmetric Path Statistics SDP}

In this section we study a simplified version of the Local Statistics SDP, which will ultimately be key to our analysis of the full Local Statistics SDP. Let $(\bx,\bG) \sim \PP$, thinking of $\bx$ as a collection of $k$ vectors $\bx_1,...,\bx_k \in \{0,1\}^n$. One can check that the partition matrix for the planted labelling, in the sense of Definition \ref{def:labeling_partition}, is
\begin{align*}
    \bP 
    &\colonequals \frac{k}{k-1}\left(\sum_{i \in [k]} \bx_i \bx_i^\top - \frac{1}{k}\onesmat_k \right) \\
    \intertext{Since deterministically $\bx_1 + \cdots + \bx_k = 1$, we as well have}
    &= \frac{k}{k-1}\left(\sum_{i \in [k]} \bx_i\bx_i^\top - \frac{1}{k}\sum_{i,j}\bx_i\bx_j^\top \right)
\end{align*}
As we observed in the Introduction, $\bP$ is PSD with ones on the diagonal, and $(k/n)\bP$ is the orthogonal projector onto the $(k-1)$-dimensional subspace spanned by $\bx_1,...,\bx_k$ and orthogonal to the vector of all-ones.

We will be particularly interested in the inner products $\langle \bP, \nb{s}{\bG}\rangle$, which count non-backtracking walks on $\bG$, with weight $1$ if the endpoints share a group label, and $-(k-1)^{-1}$ if they do not. The following is a consequence of Lemma \ref{lem:forests} in the sequel.

\begin{lemma}\label{lem:X-inner-products}
    For every $s \ge 1$ and increasing, nonnegative function $\Delta(n)$,
    $$
        \PP\left[ \left| \langle \bP, \nb{s}{\bG}\rangle - q_s(d\eta) n \right| > \Delta(n) \right] = O\left(\frac{n}{\Delta(n)^2}\right).
    $$
\end{lemma}

\noindent The \textit{Symmetric Path Statistics SDP}, given as input a graph $G_0$, attempts to find a ``pseudo-partition matrix,'' i.e. PSD matrix with ones on the diagonal, and whose inner products with the matrices $\nb{s}{G_0}$ are equal to $q_s(d\eta)n$, at least up to the fluctuations in Lemma \ref{lem:X-inner-products}.

\begin{definition}[Symmetric Path Statistics]
    \label{def:local-path-stats}
    The level-$D$ \emph{Symmetric Path Statistics Algorithm} with error tolerance $\delta > 0$, on input a $d$-regular graph $G_0$ on $n$ vertices, is the following SDP: find $\widetilde P \succeq 0$ so that
    \begin{enumerate}
        \item $\widetilde P_{u,u} = 1$ for every $u \in [n]$
        \item $\langle \widetilde P, \onesmat_n \rangle   \in [-\delta n^2, \delta n^2]$
        \item $\langle \widetilde P, \nb{s}{G_0}\rangle \in q_s(d\eta)n + [-\delta n, \delta n]$ for every $s \in [D]$.
    \end{enumerate}
\end{definition}

\begin{theorem}
    \label{thm:local-path-stats}
    Let $\QQ$ and $\PP$ be as in Definition 8.1. If $(d\eta)^2 > 4(d-1)$, then for every $D \ge 2$ there exists an error tolerance $\delta > 0$ at which the level $D$ Symmetric Path Statistics SDP can w.h.p.\ distinguish $\PP$ and $\QQ$. If $(d\eta)^2 \le 4(d-1)$, then no such $D$ and $\delta$ exist.
\end{theorem}

\begin{proof}
    By Lemma 8.5, this SDP is with high probability feasible on input $\bG \sim \PP$. Our proof will therefore show that when $\eta^2$ is sufficiently large, the SDP for some constant $D$ is \textit{infeasible} on input $\bG \sim \QQ$, whereas for $\eta^2$ sufficiently small, it is feasible for every constant $D$. 
    
    First, fix $D \ge 2$ and assume $(d\eta)^2 > 4(d-1)$. We will show that there exists $\delta > 0$ so that with high probability the level-$D$ Symmetric Path Statistics is infeasible on input $\bG \sim \QQ$. Our strategy will be to find a polynomial $f$ ith the property that $f(A_{\bG}) \succeq 0$ with high probability, but we can deduce $\langle P, f(A_{\bG}) \rangle < 0$ from the affine constraints in Definition \ref{def:local-path-stats}.
    
    Let $f$ be a degree $D$ polynomial which is strictly positive on the closed interval $[-2\sqrt{d-1},2\sqrt{d-1}]$ and satisfies $f(d\eta) < 0$; our assumption on $\eta$ ensures that this is possible, for instance by setting $f(z) = 2(d-1) + \tfrac{1}{2}(d\eta)^2 - z^2$. From our preliminaries on non-backtracking walks and the polynomials $q_s$ we know
    $$
        f = \sum_{s = 0}^D \frac{\langle f,q_s\rangle_{\km}}{\|q_s\|^2_{\km}} \, q_s.
    $$
    When $\bG \sim \QQ$, $A_{\bG}$ has an eigenvalue at $d$ whose eigenvector is the all-ones vector and by Friedman's Theorem \cite{Friedman-2003-SecondEigenvalue} its remaining eigenvalues with high probability have absolute value at most $2\sqrt{d-1} + o(1)$. Our assumptions on $f$ therefore imply $f(A_{\bG}) - f(d)\onesmat/n = f(A_{\bG} - d\onesmat/n) \succeq 0$ with high probability.
    
    On the other hand, if $\widetilde\bP \succeq 0$ is a feasible solution for the degree-$D$ Symmetric Path Statistics SDP on input $\bG \sim \QQ_n$, then
    \begin{align*}    
        0 
        &\le \langle \widetilde\bP, f(A_{\bG}) - f(d)\onesmat/n \rangle \\
        &= \sum_{s = 0}^D \frac{\langle f, q_s \rangle_{\km}}{\|q_s\|^2_{\km}} \langle \widetilde\bP, q_s(A_{\bG})\rangle + \delta |f(d)| n \\
        &\le \sum_{s = 0}^D \frac{\langle f,q_s \rangle_{\km}}{\|q_s\|^2_{\km}}q_s(d\eta) n + \delta\left( \|f\|^2_{\km} + |f(d)|\right)n \\
        &= \left(f(d\eta) + \delta\left( \|f\|^2_{\km} + |f(d)|\right)\right)n < 0,
    \end{align*}
    if we set $\delta$ sufficiently small.
    
    We can now turn to the case $(d\eta)^2 \le 4(d-1)$, seeking to prove that on input $\bG \sim \QQ$, with high probability \textit{every} level of the Symmetric Path Statistics hierarchy is feasible, for every error tolerance $\delta$. We will use the following lemma, which may be proved by adapting the proof of \cite[Proposition 4.8]{banks2019local}. The proof proceeds by setting $\widetilde\bP$ equal to a mild modification of the matrix $g(A_{\bG}) - g(d)\onesmat_n/n$.\footnote{The referenced result in \cite{banks2019local} was proved for an SDP which shared constraints (1) and (2) from Definition \ref{def:local-path-stats}, but in which constraint (3) read $\langle P,\nb{s}{G_0}\rangle = \lambda^s \|q_s\|^2_{\km} n$. The proof may be adapted simply by adopting the hypotheses below, changing every instance of the aforementioned constraints, and taking some care with the $\delta$ slack.}
    
    \begin{lemma}
        \label{lem:feasible-from-polynomial}
        Assume there exists a constant-degree polynomial $g \in \RR[z]$ that is strictly positive on $[-2\sqrt{d-1},2\sqrt{d-1}]$ and satisfies
        $$
            \langle g, q_s\rangle_{\km} \in q_s(d\eta) + [-\delta,\delta].
        $$
        for every $s = 1,...,D$. Then the level-$D$ Symmetric Path Statistics SDP with error tolerance $\delta$ is w.h.p.\ feasible on input $\bG\sim \QQ$.
    \end{lemma}
    
    Lemma 8.9 in hand, we need only to construct such a polynomial. Assume that $(d\eta)^2 \le 4(d-1)$, so that $d\eta \in [-2\sqrt{d-1},2\sqrt{d-1}]$, the interval in which lie the roots of every polynomial $q_s$. Let $P \gg D$, and write $r_1 < r_2 < \cdots < r_P$ for the roots of $q_P$. Let $I \subset [P]$ contain the indices of the $(D + 1)/2$ roots of $q_P$ closest to $d\eta$, and set
    $$
        g_\eta = \frac{1}{\zeta}\prod_{i \notin I}(z - r_i)^2,
    $$
    where $\zeta$ is a normalizing factor to ensure that $\langle g_\eta, 1 \rangle_{\km} = 1$. 
    
    This polynomial is certainly nonnegative, and its degree is $2P - D - 1$. From Lemma 8.4, then, there exist $w_1,...,w_P \ge 0$ so that for any $s = 0,...,D$
    $$
        \langle g_\eta, q_s \rangle_{\km} = \sum_{i \in I} w_i g_\eta(r_i) q_s(r_i).
    $$
    In particular, setting $s = 0$ and recalling the definition of $\zeta$, we have
    $$
        1 = \langle g_\eta, 1 \rangle_{\km} = \sum_{i \in I} w_i g_\eta(r_i).
    $$
    This means that for any $s = 0,...,D$, the inner product $\langle g_\eta, q_s \rangle$ is a weighted average of $q_s$ evaluated at the $(D + 1)/2$ roots of $q_P$ closest to the point $d\eta$. Since the roots of the $q$ polynomials are dense in $(-2\sqrt{d-1},2\sqrt{d-1})$, and $d\eta$ is in the closure of this interval, for each $D$ and $\delta > 0$, there exists a constant $P$ for which $|\langle g_\eta, q_s \rangle_{\km} - q_s(d\eta)| < \delta$ for every $s = 0,..,D$. 
\end{proof}

\subsection{Partially Labelled Subgraphs}

We are now ready to study the full Local Statistics algorithm. To start, we will need to develop a basis for the symmetric polynomials appearing as affine moment matching constraints in Definition~\ref{def:local-stats-informal}. Because $\EE p(\bx,\bG) = 0$ for any $p \in \mathcal{I}_k$, a constraint shared by the pseudoexpectation, it would suffice to study the subspace of $\RR[x,G]/\mathcal{I}_k$ fixed under the $S_n$ action inherited from $\RR[x,G]$. However, it is computationally favorable to work in a slightly larger vector space instead.

\begin{definition}
    \label{def:sym-polys}
    Let us write $\mathbb{S}[x,G]_{\le D_x,D_G} \subset \RR[x,G]$ for the vector space of polynomials that (1) satisfy $\deg_x \le D_x$ and $\deg_G \le D_G$, (2) are symmetric with respect to the $S_n$ action, (3) are multilinear in $G$ and $x$, and (4) for which at most one of $x_{u,1},...,x_{u,k}$ appears in each monomial, for every $u \in V$.
\end{definition}

Every polynomial appearing as a moment constraint in the level-$(D_x,D_G)$ Local Statistics algorithm belongs to $\mathbb{S}[x,G]_{\le D_x,D_G}$. We now give a combinatorially structured basis for this vector space, similar to the `shapes' of \cite{BHKKMP-2019-PlantedClique}.

\begin{definition}
    \label{def:plg}
    A \textit{partially labelled graph} $(H, S, \tau)$ consists of a graph $H$, a subset $S \subset V(H)$, and a map $\tau : S \to [k]$; we say that a graph is fully labelled if $S = V(H)$, and in this case write $(H,\tau)$ for short. A \emph{homomorhism} from $(H,S,\tau)$ into a fully labelled graph $(G,\sigma)$ is a map $\phi : V(H) \to V(G)$ that takes edges to edges and agrees on labels; an \emph{occurrence} of $(H,S,\tau)$ in $(G,\sigma)$ is an injective homomorphism. For each partially labelled graph $(H,S,\tau)$ there is an associated polynomial in $\mathbb{S}[x,G]$,
    $$
        p_{(H,S,\tau)}(x,G) = \sum_{\phi : V(H) \inj [n]}\prod_{(u,v) \in E(H)}G_{\phi(u),\phi(v)} \prod_{u \in S}x_{\phi(u),\tau(u)}.
    $$
    Each point in the zero locus of $\mathcal{I}_k$ may be identified with a fully labelled graph $(G,\sigma)$. Evaluated at such a point, this polynomial counts the number of occurrences of $(H,S,\tau)$ in $(G,\sigma)$. Finally, $\deg_x p_{(H,S,\tau)} = |S|$ and $\deg_G p_{(H,S,\tau)} = |E(H)|$.
\end{definition}

\begin{lemma}
    \label{lem:plg-basis}
    The polynomials $p_{(H,S,\tau)}$ with $|E(H)| \le D_G$ and $|S| \le D_x$ are a vector space basis for $\mathbb{S}[x,G]_{\le D_x,D_G}$
\end{lemma}

\begin{proof}
    Let $s : \RR[x,G] \to \mathbb{S}[x,G]$ be the map that sends a polynomial $p$ to the sum over its $S_n$ orbit. The vector space $\mathbb{S}[x,G]_{\le D_x,D_G}$ is spanned by the images under $s$ of the multilinear monomials with $x$-degree $D_x$ and $G$-degree $D_G$ in $\RR[x]$ in which at most one of $x_{u,1},...,x_{u,k}$ appears for each $u \in [n]$. From each such monomial $m(x,G)$ one can extract a partially labelled graph $(H_m,S_m,\tau_m)$, where $|S_m| \le D_x$, $|E(H_m)| \le D_G$, and $H_m$ is a subgraph of the complete graph: $E(H_m)$ is the union of all pairs $(u,v)$ appearing as an index of a $G$ variable, $S_m$ is the union of all $u$ occurring as an index of an $x$ variable, $V(H_m)$ is the union of $S$ and the endpoints of every edge, and $\tau(u) = i$ if the variable $x_{u,i}$ appears. Because $S_n$ acts transitively on $[n]$, the orbit of $m(x)$ corresponds to every possible injection $V(H_m) \inj [n]$, and thus
    \[
        s(m(x,G)) = p_{(H_m,S_m,\tau_m)}(x,G).
    \]
    This shows that the $p_{(H,S,\tau)}$ span. To see that they are independent, observe that each monomial appears as a term in exactly one $p_{(H,S,\tau)}$. 
\end{proof}

In view of this lemma, the moment constraints in Definition~\ref{def:local-stats-informal} are equivalent to the requirement that $\pseudo p_{(H,S,\tau)}(x,G_0) \approx \EE p_{(H,S,\tau)}(\bx,\bG)$ for every $(H,S,\tau)$ with at most $D_G$ edges and $D_x$ distinguished vertices, up to isomorphism. In order to instantiate and analyze the Local Statistics algorithm, we now need to compute these expectations, and bound the fluctuations around them.

\subsection{Local Statistics in the Planted Model}

Instead of working directly in $\PP$, we will as usual work in the \textit{configuration model} $\widehat \PP$, a distribution on multigraphs with two key properties: (i) with probability bounded away from zero as $n\to\infty$, $\widehat\bG \sim \widehat\PP$ is simple, and (ii) the conditioal distribution of $\widehat\PP$ on this event is equal to $\PP$. In this model, conditional on a balanced partition $\sigma$, we adorn each vertex in every group $i$ with $dM_{i,j}$ `half-edges' labelled $i \to j$, and then for each $i,j \in [k]$ randomly match the $i\to j$ half-edges with the $j\to i$ half-edges. When $k = 1$, this is the usual configuration model on $d$-regular graphs, and when $k=2$ and $M_{i,i} = 0$, it gives bipartite regular graphs. Thus many results we prove for $\PP$ will apply to $\QQ$ as well. 

\begin{claim}
    \label{fact:config-aas}
    Write $\mathcal{G}_n$ and $\widehat{\mathcal{G}}_n$ for the sets of all $d$-regular, $n$-vertex graphs and multigraphs, respectively. If $\widehat{\mathcal{E}}_n \subset \widehat{\mathcal{G}}_n$ is a sequence of events holding w.h.p.\ in $\widehat\PP_n$, then $\widehat{\mathcal{E}}_n \cap \mathcal{G}_n$ holds w.h.p.\ in $\PP_n$.
\end{claim}

\begin{proof}
    Since $\widehat\PP_n[\mathcal{G}_n]$ is bounded away from zero and $\widehat\PP[\widehat{\mathcal{E}}_n] = 1 - o_n(1)$, we have
    \[
        \PP_n[\widehat{\mathcal{E}}_n \cap \mathcal{G}_n] = \frac{\widehat\PP_n[\widehat{\mathcal{E}}_n \cap \mathcal{G}_n]}{\widehat\PP_n[\mathcal{G}_n]} \ge \frac{\hat\PP_n[\mathcal{G}_n] - o_n(1)}{\hat\PP_n[\mathcal{G}_n]} = 1 - o_n(1). \qedhere
    \]
\end{proof}

As a warm-up, let us recall some standard calculations of subgraph probabilities in these two simpler situations. 

\begin{lemma}
    \label{lem:config-hom}
    Let $\widehat\bG$ be a multi-graph produced by the $d$-regular configuration model on $n$ vertices. If $H$ is a simple graph, the probability that a fixed injection $\phi : V(H) \inj V(\widehat\bG)$ is a homomorphism is
    $$
        \prod_{v \in V(H)} \frac{d!}{(d - \deg(v)) !} \cdot \frac{(dn - 2|E(H)|-1)!!}{(dn - 1)!!} = \prod_{v \in V(H)} \frac{d!}{(d - \deg(v))!} \cdot (dn)^{-|E(H)|} + O(n^{-|E(H)| - 1}).
    $$
    Similarly, let $(\widehat\bG,\sigma)$ be generated by the $d$-biregular configuration model on $2n$ vertices, $\sigma : V(\widetilde\bG) \to [2]$ its left/right labelling. If $(H,\tau)$ is a simple bipartite graph with left/right labelling $\tau$, then the probability that a fixed injection $\phi :V(H) \inj V(\widehat\bG)$ agreeing on labels is a homomorphism is
    $$
        \prod_{v\in V(H)} \frac{d!}{(d - \deg(v))!} \cdot \frac{(dn - |E(H)|)!}{(dn)!} =  \prod_{v\in V(H)} \frac{d!}{(d - \deg(v))!} \cdot (dn)^{-|E(H)|} + O(n^{-|E(H)| - 1}).
    $$
\end{lemma}

\begin{proof}[Sketch]
    In the case of the $d$-regular configuration model, once an injective $\phi$ has been chosen, there are $\tfrac{d!}{(d - \deg(v))!}$ ways to choose $\deg(v)$ stubs from the $d$ available at each vertex $v$ to be matched with the appropriate stubs at each intended neighbor of $v$, and then $(d_n - 2|E(H)| - 1)!!$ ways to match the remaining stubs. On the other hand, the total number of multi-graphs possible to output is $(dn - 1)!!$. The calculation is analogous in the bipartite regular case.
\end{proof}

We will use Lemma~\ref{lem:config-hom} to compute local statistics in the planted model, but first we will need a bit more notation. Let $(H,S,\tau)$ be a partially labelled graph, and $\widehat \tau$ an extension of $\tau$, i.e. $\widehat \tau : V(H) \to [k]$, and $\widehat\tau|_S = \tau$; for each $i,j \in[k]$ and $u \in V_i(H)$, write $\deg_j(u)$ for the number of neighbors that $u$ has in group $j$ according to $\widehat \tau$. Let us define
\begin{equation}
    (dM)^{(H,S)}_\tau \colonequals \sum_{\widehat \tau : \widehat \tau|_S = \tau} \frac{\prod_v \prod_{j \in [k]} dM_{\widehat\tau(v), j} (dM_{\widehat \tau(v),j} - 1) \cdots (dM_{\widehat \tau(v),j} - \deg_{i \to j}(v) + 1)}{\prod_{(u,v) \in E(H)} dM_{\widehat\tau(u),\widehat\tau(v)}}
\end{equation}
if $dM_{i,j} \ge \deg_j(v)$ for every $i,j \in [k]$ and $v \in \widehat\tau^{-1}(i)$, and zero otherwise. Note that this operation is multiplicative on disjoint unions: 
$$
    (dM)_{\tau_1 \sqcup \tau_2}^{(H_1\sqcup H_2,S_1\sqcup S_2)} = (dM)_{\tau_1}^{(H_1,S_1)} (dM)_{\tau_2}^{(H_2,S_2)}.
$$
Finally, let us write $\chi(H) \colonequals |V(H)| - |E(H)|$ and $\cc(H)$ for the number of connected components.

Since we are aiming to prove high probability statements regarding $p_{(H,S,\tau)}(\bx,\bG)$ for $(\bx,\bG) \sim \PP$ by studying the configuration model, we need to extend the quantities $p_{(H,S,\tau)}(x,G)$ to the case when $G$ is a multigraph with self-loops. For convenience, we will define an \textit{occurrence} of $(H,S,\tau)$ in a fully labelled, loopy multigraph as an occurence of $(H,S,\tau)$ in the simple graph obtained by deleting all self-loops and merging all multiedges between each pair of vertices. Our key lemma computes the expected number of occurrences in the configuration model.

\begin{lemma}
    \label{lem:config-occ}
    Let $(H,S,\tau)$ be a partially labelled graph on $O(1)$ edges, and  $(\bx,\widehat\bG)$ be drawn from the configuration model $\widehat\PP$. Then
    $$
        \EE p_{(H,S,\tau)}(\bx,\widehat\bG) = (n/k)^{\chi(H)}(dM)^{(H,S)}_\tau + O(n^{\chi(H) - 1}).
    $$
\end{lemma}

\begin{proof}
    Let $V(\widehat\bG) = [n]$, fix a labelling $\sigma : [n] \to [k]$, and let $\widehat\bG$ be drawn from the configuration model. Fix an extension $\widehat \tau$ of $\tau$. If $\phi: V(H) \inj V(\widehat\bG)$ is an injection that agrees on labels, applying Lemma 8.9 to the multigraphs on each set of vertices $\sigma^{-1}(i)$ and between each pair of sets $\sigma^{-1}(i)$ and $\sigma^{-1}(j)$,
    \begin{align*}
        \PP[\phi \text{ is an occurrence}] &= (n/k)^{|E(H)|} \cdot \prod_{i\le j} (dM_{i,j})^{-|E_{i,j}(H,\widehat\tau)|} \cdot \prod_{v \in V(H,\widehat\tau)} \prod_{i,j} \frac{(d M_{i,j})!}{(dM_{i,j} - \deg_{j}(v))!} \\
        &\qquad + O(n^{-|E(H)| - 1}),
    \end{align*}
    and there are $(n/k)^{V(H)} + O(n^{|V(H)| - 1})$ injective choices for $\phi$ that agree on labels.
    
    Finally, writing $\Phi(H,S,\tau)$ for the total number of occurrences of $(H,S, \tau)$ in $\widehat\bG$, and $\Phi(H,\widehat \tau)$ for the number of occurrences of the fully labelled graph $(H,\widehat\tau)$,
    \begin{align*}
        \EE \Phi(H,S,\tau)
        &= \EE \sum_{\widehat \tau : \widehat \tau|_S = \tau} \Phi(H,\widehat \tau) \\
        &= (n/k)^{\chi(H)} (dM)^{(H,S)}_{\tau} + O(n^{|V(H)| - |E(H)| - 1}). \qedhere
    \end{align*}
\end{proof}

This lemma has some immediate consequences. First, it tells us that occurrences of partially labelled forests are sharply concentrated.

\begin{lemma}
    \label{lem:forests}
    Let $(H,S,\tau) = \bigsqcup_{t \in [\cc(H)]} (H_t,S_t,\tau_t)$ be a partially labelled graph with $O(1)$ vertices, $\cc(H)$ connected components $H_t$, and no cycles. Then for any function $f(n) > 0$, 
    \begin{align*}
        \PP\left[\left| p_{(H,S,\tau)}(\bx,\bG) - (n/k)^{\cc(H)}\prod_{t \in [\cc(H)]} (dM)^{(H_t,S_t)}_{\tau_t}\right| > f(n) \right] &= O\left(\frac{n^{2\cc(H) - 1}}{f(n)^2}\right).
    \end{align*}
\end{lemma}

\begin{proof}
    The expectation $\EE [p_{(H,S,\tau)}(\bx,\bG)^2]$ is a sum over all pairs of injective, label-consistent maps $\phi_1,\phi_2$, of the probability that both maps are occurrences. The image of two disjoint copies of $H$ under these two injective maps is a graph $H'$ that is either $H\sqcup H$, or is obtained by identifying some pairs of vertices whose $\tau$ labels agree---each pair with one vertex from each copy of $H$. We can promote $H'$ to a partially labelled graph by taking the induced partial labelling $\tau'$ from $\tau$. Thus, let us think of the pair $\phi_1,\phi_2$ as a single injective map $\varphi :V(H') \inj V(\widetilde\bG)$ that agrees with $\tau'$. Thus
    \begin{align*}
        \EE [p_{(H,S,\tau)}(\bx,\bG)^2]
        &= \sum_{H'} \EE p_{(H',S',\tau')}(\bx,\bG) \\
        &= \sum_{H'}\left((n/k)^{\chi(H')} (dM)_{\tau'}^{(H',S'} + O(n^{\chi(H') - 1})\right).
    \end{align*}
    When $H' = H\sqcup H$, from our observation above
    $$
        (dM)_{\tau'}^{(H',S')} = ((dM)_{\tau}^{(H,S)})^2 = \left(\prod_{t \in [\cc(H)]} ((dM)_{\tau_t}^{(H_t,S_t)})\right)^2.
    $$ 
    As $H$ has no cycles, every other $H'$ satisfies $\chi(H') < 2\chi(H)$. Thus since $\cc(H) = \chi(H)$, the assertion is true in the configuration model by Chebyshev, and transfers immediately to the planted model. For good measure, an application of the triangle inequality shows as well that
    \[
        \PP\left[ \left| p_{(H,S,\tau)}(\bx,\bG) - \prod_{t \in [\cc(H)]}p_{(H_t,S_t,\tau_t)}(\bx,\bG)\right| > f(n) \right] = O\left(\frac{n^{2\cc(H) - 1}}{f(n)^2}\right). \qedhere
    \]
\end{proof}

Thus with high probability, the counts of partially labelled forests enjoy concentration of $\pm o(n^{\cc(H)})$. On the other hand, an immediate application of Markov in $\widetilde \PP$ tells us that there are very few occurrences of partially labelled graphs with cycles.

\begin{lemma}
    \label{lem:cycles}
    Let $(H,S,\tau)$ be a partially labelled graph with $O(1)$ edges and at least one cycle. Then for any function $f(n) > 0$,
    $$
        \PP \left[ p_{(H,S,\tau)}(\bx,\bG) > f(n)\right] = O\left(\frac{n^{\cc(H) - 1}}{f(n)}\right).
    $$
\end{lemma}

\noindent In particular, we will need later on the fact that there are very few vertices within constant distance of a constant length cycle. The proof is once again Markov, combined with a union bound.

\begin{lemma}
    \label{cor:bad-vertices}
    Let $\bG \sim \PP$ or $\QQ$. Fix constants $L$ and $C$, and call a vertex \emph{bad} if it is at most $L$ steps a way from a cycle of length at most $C$. Then w.h.p.\ there are fewer than $f(n)$ bad vertices, for any increasing function $f(n)$.
\end{lemma}

We can now restate the local statistics algorithm more precisely. 
 
\begin{definition}[Local Statistics Algorithm with Formal Moment Constraints]
    \label{def:full-local-stats}
    The \emph{degree-$(D_x,D_G)$ Local Statistics algorithm} with error tolerance $\delta$ is the following SDP. Given an input graph $G_0$, find $\pseudo : \RR[x]_{\le D_x} \to \R$ s.t.
    \begin{enumerate}
        \item (Positivity) $\pseudo p(x)^2 \ge 0$ whenever $\deg p^2 \le D_x$ 
        \item (Hard Constraints) $\pseudo p(x,G_0) = 0$ for every $p \in \mathcal{I}_k$
        \item (Moment Constraints) For every $(H,S,\tau)$ with at most $D_x$ distinguished vertices, $D_G$ edges, and $\ell$ connected components,
        $$
            \pseudo p_{(H,S,\tau)}(x,G_0) = (dM)_\tau^{(H,S)}(n/k)^{\chi(H)} \pm \delta n^{\cc(H)}.
        $$
    \end{enumerate}
\end{definition}

\noindent The $n^{\chi(H)}$ vs. $n^{\cc(H)}$ scaling may seem ad hoc, but as promised above we have arranged things so that the the SDP is w.h.p.\ feasible when its input is drawn from the planted model.

\begin{lemma}
    \label{lem:feasible-on-planted}
    Fix $D_x,D_G$ constant, and $\delta > 0$. with high probability, the degree-$(D_x,D_G)$ Local Statistics Algorithm with error tolerance $\delta$ is feasible on input $\bG \sim \PP$.
\end{lemma}

\begin{proof}
    Let $(\bx,\bG) \sim \PP$, and for each $p(x,G) \in \RR[x,G]_{\le D_x,D_G}$ define
    $$
        \pseudo p(x,\bG) = p(\bx,\bG).
    $$
    This satisfies positivity, as $\pseudo p(x,\bG)^2 = p(\bx,\bG)^2 \ge 0$, and obeys the hard constraints because $(\bx,\bG)$ lies in the zero locus of $\mathcal{I}_k$. Finally, let $(H,S,\tau)$ be a partially labelled graph. If $H$ has a cycle, then by Corollary 8.15, w.h.p.
    $$
        \left|p_{(H,S,\tau)}(\bx,\bG) - (n/k)^{\chi(H)} (dM)_\tau^{(H,S)}\right| \le p_{(H,S,\tau)}(\bx,\bG) + O(n^{\chi(H)}) \le \delta n^{\cc(H)}.
    $$
    On the other hand, if $H$ has no cycles, then $\chi(H) = \cc(H)$ and w.h.p.
    $$
        \left|p_{(H,S,\tau)}(\bx,\bG) - (n/k)^{\chi(H)}(dM)_\tau^{(H,S)}\right| \le \delta n^{\cc(H)}
    $$
    by Proposition 8.14. There are only constantly many partially labelled subgraphs with at most $D_x$ vertices and $D_G$ edges, so a union bound finishes the proof.
\end{proof}

Finally, we end this subsection with the proof of Lemma \ref{lem:X-inner-products}, which concerned the affine constraints in the Symmetric Path Statistics SDP

\begin{proof}[Proof of Lemma \ref{lem:X-inner-products}]
    Recall the partition matrix
    $$
        \bP = \frac{k}{k-1}\left(\sum_{i \in [k]} \bx_i \bx_i^\top - \frac{1}{k}\sum_{i,j \in [k]} \bx_i \bx_j^\top\right)
    $$
    from Section 8.2, where each $\bx_i \in \{0,1\}^n$ is the indicator vector for membership in the $i$th group. We are interested in $\langle \bP, \nb{s}{\bG}\rangle$.
    
    Let $(P_s,\{0,s\},\{i,j\})$ denote a path of length $s$ with distinguished endpoints labelled $i$ and $j$, and write its vertices as $V(P_s) = \{0,1,...,s\}$. From Lemma \ref{lem:forests}, w.h.p.\ for $(\bx,\bG)\sim\PP$,
    $$
        p_{(P_s,\{0,s\},\{i,j\}}(\bx,\bG) = \frac{1}{k} (dM)_{i,j}^{(P_s,\{0,s\}}n \pm o(n).
    $$
    Expanding the right hand side,
    \begin{align*}
        (dM)_{i,j}^{(P_s,\{0,s\}))}
        &= \sum_{\widehat \tau : \widehat \tau |_S = \tau} \frac{\prod_{v \in V(P_s)} \prod_{j' \in [k]} dM_{\widehat\tau(v),j'}(dM_{\widehat\tau(v),j'} - 1)\cdots(dM_{\widehat\tau(v),j'} - \deg_{j}(v))}{\prod_{(u,v) \in E(P_s)} dM_{\widehat\tau(u),\widehat\tau(v)}} \\
        &= \sum_{\widehat\tau : \widehat\tau|_S = \tau} dM_{i,\widehat\tau(1)} \prod_{t = 1,...,s-1}\left(dM_{\widehat\tau(t),\widehat\tau(t+1)} - \{\widehat\tau(t+1)=\widehat\tau(t-1)\}\right) \\
        &= q_s(dM)_{i,j}.
    \end{align*}
    The expression in the second to last line counts the number of non-backtracking walks of length $s$ between vertices $i$ and $j$ on the multi-graph whose adjacency matrix is $dM$; from Section \ref{sec:nb-walks} these may be enumerated using the polynomial $q_s$ applied to $dM$. 
    
    Now, let us define $\sa{s}{\bG}$ as the $n\times n$ matrix whose entries count \textit{self-avoiding} (as opposed to non-backtracking) walks on $\bG$. By definition
    $$
        p_{(P_s,\{0,s\},\{i,i\})}(x,\bG) = \sum_{u,v} \left(\sa{s}{\bG}\right)_{u,v} \cdot x_{u,i}x_{v,j}.
    $$
    Thus, with high probability
    \begin{align*}
        \langle \bP, \sa{s}{\bG} \rangle 
        &= \frac{k}{k-1}\left(\sum_i p_{(P_s,\{0,s\},\{i,i\})}(\bx,\bG) - \frac{1}{k}\sum_{i,j}p_{(P_s,\{0,s\},\{i,j\})}(\bx,\bG)\right) \\
        &= \frac{1}{k-1}\langle q_s(dM), \id - \onesmat/k \rangle n \pm o(n) \\
        &= \frac{1}{k-1}\Big(\Tr q_s(dM) - q_s(d)\Big) n + o(n)\\
        &= q_s(d\lambda)n + o(n).
    \end{align*}
    The last equality follows from the fact that the spectrum of $dM$ consists of an eigenvalue $d$ with multiplicity one, and an eigenvalue $d\lambda$ with multiplicty $k-1$.
    
    We need finally to ensure that the same inner product constraint holds for the matrices $\nb{s}{\bG}$. It is an easy consequence of Lemma 8.16 that w.h.p.\ for $\bG \sim \QQ$, the matrices $\nb{s}{\bG}$ and $\sa{s}{\bG}$ disagree in at most $o(n)$ rows. Thus, since the $L^1$ norm of every row is bounded by a constant (by degree-regularity), w.h.p.\ $\| \nb{s}{\bG} - \sa{s}{\bG}\|_F^2 = o(n)$. Since $X$ is PSD with ones on the diagonal, every off-diagonal element has magnitude at most one---thus
    $$
        \langle X,\nb{s}{\bG}\rangle = \langle X,\sa{s}{\bG}\rangle + o(n),
    $$
    and we are done.
\end{proof}

\subsection{Proof of Theorem \ref{thm::LS:equitable-model}: Upper Bound}

We will show that if $(d\lambda)^2 > 4(d-1)$, then the degree $(2,D)$ Local Statistics algorithm can distinguish $\PP$ and $\QQ$ for every $D \ge 2$. Specifically, we will show that for any such $D$, with high probability over input $\bG \sim \QQ$ there exists a $\delta$ at which the SDP is \textit{infeasible}. Our goal, here and in the proof of the lower bound, will be to reduce to our characterization of the Symmetric Path Statistics SDP in Theorem 8.7.

Assume that $\pseudo$ is a feasible pseudoexpectation for the degree $(2,D)$ Local Statistics SDP with tolerance $\delta>0$, on input $\bG\sim\QQ$, and consider the matrix $X$ with entries
$$
    \widetilde \bP_{u,v} = \frac{k}{k-1}\left(\sum_{i \in [k]} \pseudo (x_{u,i} - 1/k)(x_{v,i} -  1/k)\right) = \frac{k}{k-1}\left(\sum_{i \in [k]} \pseudo x_{u,i}x_{v,i} -  \frac{1}{k}\sum_{i,j\in[k]} \pseudo x_{u,i}x_{v,j}\right).
$$
We will show that $\widetilde \bP$ is a feasible solution to the level-$D$ Symmetric Path Statistics SDP with the input $\bG$, at some tolerance $\delta' = c\delta$---thus by Theorem 8.7, when $(d\lambda)^2 > 4(d-1)$ and $\delta$ is sufficiently small, we will have a contradiction. We observe first that $X$ is PSD with ones on the diagonal---these facts follow immediately from the hard constraints in Definition 8.17, which tell us that $\pseudo x_{u,i}^2 = \pseudo x_{u,i}$ for every $u$ and $i$, and $\pseudo\sum_{i \in [k]} x_{u,i} = 1$ for every $u$. 

We turn now to the moment constraints, with the goal of showing
\begin{align*}
    \langle \widetilde \bP, \nb{s}{\bG}\rangle &= q_s(d\lambda)n \pm \delta' n & & \forall s \in [D] \\
    \langle \widetilde \bP, \onesmat \rangle &= 0 \pm \delta' n^2.
\end{align*}
Rehashing our calculations from the proof of Lemma \ref{lem:X-inner-products}, we find that w.h.p.
\begin{align*}
    \langle \bP, \nb{s}{\bG} \rangle 
    &= \langle \widetilde \bP, \sa{s}{\bG} \rangle + o(n) \\
    &= \pseudo \frac{k}{k-1}\left(\sum_i p_{(P_s,\{0,s\},\{i,i\})}(x,\bG) - \frac{1}{k}\sum_{i,j}p_{(P_s,\{0,s\},\{i,j\})}(x,\bG)\right) + o(n) \\
    &= \frac{1}{k-1}\langle q_s(dM), \id - \onesmat/k \rangle n \pm \frac{k}{k-1}\cdot 2k\delta n + o(n) \\
    &= q_s(d\lambda)n \pm \frac{k^3}{k-1}\delta n + o(n).
\end{align*}

To verify that $X$ has the correct inner product against the all-ones matrix, consider two partially labelled subgraphs: a single vertex labelled $i \in [k]$, and two disjoint vertices both labelled $i \in [k]$. The sum of their corresponding polynomials is $\sum_{u,v} x_{u,i}x_{v,i}$, and identically in the planted model
$$
    \sum_{u,v}\bx_{u,i}\bx_{v,i} = (n/k)^2.
$$
Our pseudoexpectation is required to match this up to an additive $\delta(n + n^2)$, the $n$ and $n^2$ terms respectively coming from the additive slack in the one vs. two vertex graphs. Thus
$$
    \langle \widetilde \bP, \onesmat \rangle = \frac{k}{k-1}\left(n^2/k \pm \delta(n^2 + n) - n^2/k\right) = 0 \pm \frac{2k}{k-1}\delta n^2.
$$

\subsection{Proof of Theorem \ref{thm::LS:equitable-model}: Lower Bound}

Assume that $(d\lambda)^2 \le 4(d-1)$; we need to explicitly construct a degree-$(2,D)$ pseudoexpectation that is with high probability feasible for $\bG \sim \QQ$. Our tactic will be to show that such an operator can be constructed from a feasible solution to the Symmetric Path Statistics SDP guaranteed us by Theorem 8.7. 

Before building the degree-$(2,D)$ pseudoexpectation asserted to exist in the theorem statement, we will first prove a series of structural lemmas showing that it suffices to check only a subset of the moment constraints of a Local Statistics pseudoexpectation. 

First, we show that the moment constraints regarding the pseudo-expected counts of partially labelled graphs containing cycles are satisfied more or less for free.

\begin{lemma}
    \label{lem:cycles-for-free}
    Let $\bG \sim \QQ$, and $\pseudo = \pseudo(\bG)$ be a degree $(D_x,D_G)$ pseudoexpectation, perhaps dependent on $\bG$, that satisfies positivity and the hard constraints. For every error tolerance $\delta$, w.h.p.\ $\pseudo$ satisfies the moment constraints for all partially labelled subgraphs containing a cycle.
\end{lemma}

\begin{proof}
    It is a routine sum-of-squares calculation that for any monomial $\mu(x)$,
    $$
        (\pseudo \mu(x))^2 \le \pseudo \mu(x)^2 = \pseudo\mu(x),
    $$
    meaning that $|\pseudo \mu(x)| \le 1$. Thus for any $(H,S,\tau)$,
    $$
        |\pseudo p_{(H,S,\tau)}(x,\bG)| = \left| \sum_{\phi: V(H) \inj [n]} \prod_{(\alpha,\beta) \in E(H)} \bG_{\phi(\alpha),\phi(\beta)} \pseudo \prod_{\alpha \in S} x_{\phi(\alpha),\tau(\alpha)}\right| \le \left|\sum_{\phi : V(H)} \prod_{(\alpha,\beta) \in E(H)} \bG_{\phi(\alpha),\phi(\beta)}\right|,
    $$
    and the right hand side is simply the number of occurrences of the unlabelled graph $H$ in $\bG$. From Proposition 8.15, if $H$ has $\cc(H)$ connected components and at least one cycle, (i) w.h.p.\ this quantity is smaller than $\delta' n^{\cc(H)}$ for every $\delta' > 0$, and (ii) $\chi(H) < \cc(H)$. Thus trivially, if we set $\delta' < \delta$, w.h.p.
    \[
        |\pseudo p_{(H,S,\tau)}(x,\bG) - (dM)_\tau^{(H,S)}| \le \delta' n^{\cc(H)} + (dM)_\tau^{(H,S)}n^{\chi(H)} \le \delta n^{\cc(H)}. \qedhere
    \]
\end{proof}

\noindent It therefore suffices to check only the moment constraints for partially labelled forests. In fact, only a subset of these are important.

\begin{definition}
    \label{def:pruning}
    Let $(H,S,\tau)$ be a partialy labelled tree. The \emph{pruning} of $(H,S,\tau)$ is the unique partially labelled subtree in which every leaf belongs to $S$. The pruning of an unlabelled graph is the empty graph, and the pruning of a forest is defined tree-by-tree. We say that a partially labelled forest is \emph{pruned} if it is equal to its pruning.
\end{definition}

\begin{lemma}
    \label{lem:pruned-ratio}
    Let $(H,S,\tau)$ be a partially labelled forest with maximal degree $d$, $(\widetilde H,S,\tau)$ its pruning, and write $\deg$ and $\widetilde\deg$ for the vertex degrees in $H$ and $\widetilde H$ respectively. If $X$ is a symmetric nonnegative integer matrix with row and column sums equal to $d$, then
    $$
        \frac{X^{(H,S)}_\tau}{X^{(\widetilde H,S)}_\tau} = \prod_{v \in V(H)} \prod_{q = \widetilde\deg(v)}^{\deg(v)-1} (d - q).
    $$
\end{lemma}

\begin{proof}
    Let's reason combinatorially. Any such matrix $X$ can be thought of as the adjacency matrix for a $d$-regular multigraph with self-loops; let's fix $X$ and call this graph $\Gamma$. Since $\Gamma$'s vertex set is $[k]$, we will think of it as a fully labelled graph. By multiplicativity on disjoint unions, we may freely assume that $(H,S,\tau)$ is a tree. Let's choose a root $r \in V(\widetilde H) \subset V(H)$; having done so, $E(H)$ is in bijection with $V(H) \setminus r$ (and similarly for $E(\widetilde H)$ and $V(\widetilde H)$). Let's write $p(v)$ for the unique parent of every vertex. We can thus write
    \begin{align*}
        X_\tau^{(H,S)} = \sum_{\widetilde \tau :\widetilde \tau|_S = \tau} & \prod_{j \in [k]} X_{\widetilde\tau(r),j}(X_{\widetilde \tau(r),j} - 1) \cdots (X_{\widetilde \tau(r),j} - \deg_j(r) + 1) \\
        &\qquad\times \prod_{v \in V(H) \setminus r} \frac{\prod_{j \in [k]} X_{\widetilde\tau(v),j} \cdots (X_{\widetilde\tau(v),j} - \deg_j(v) + 1)}{X_{\widetilde\tau(v),\widetilde\tau(p(v))}}.
    \end{align*}
    Thinking of each $\widetilde\tau$ as a map $V(H) \to V(\Gamma)$, the summand above gives the number of ways to map $\eta: E(H) \to E(\Gamma)$ with the following constraints: (1) each edge $(u,v)$ must be mapped to one of the $X_{\tau(u),\tau(v)}$ edges between $\tau(u)$ and $\tau(v)$, and (2) no two edges in $E(H)$ with the same endpoint may be mapped to the same edge in $\Gamma$. We'll call the pair $(\widetilde \tau, \eta)$ a \textit{locally injective occurrence of $(H,S,\tau)$} in the fully labelled graph $\Gamma$. Thus the expression $X_\tau^{(H,S)}$ gives the number of such occurrences.
    
    The same argument applies to the pruning $(\widetilde H,S,\tau)$. Now, the graph $(H,S,\tau)$ consists pruning $(\widetilde H,S,\tau)$, plus some trees hanging off the edges. For each locally injective occurrence of $(\widetilde H,S,\tau)$, there are
    $$
        \prod_{v \in V(H)} \prod_{q = \widetilde\deg(v)}^{\deg(v) - 1}(d - q)
    $$
    ways to extend it to a locally injective occurrence of $(H,S,\tau)$, since $\Gamma$ is $d$-regular.
\end{proof}

\begin{lemma}
    \label{lem:pruned-polys-bound}
    Let $\bG \sim \QQ$, $(H,S,\tau)$ be a partially labelled forest, and $(\widetilde H,S,\tau)$ its pruning. Then w.h.p.
    $$
        \left\|p_{(H,S,\tau)}(x,\bG) - n^{\cc(H) - \cc(\widetilde H)}\frac{(dM)^{(H,S)}_{\tau}}{(dM)^{(\widetilde H,S)}_\tau} p_{(\widetilde H,S,\tau)}(x,\bG)\right\|_1 = o(n^{\cc(H)}),
    $$
    where $\|\cdot\|_1$ is the coefficient-wise $L^1$ norm.
\end{lemma}

\begin{proof}
    Let $(H,S,\tau)$ be a partially labelled forest, $(\widetilde H,S,\tau)$ its pruning. If $\widetilde\phi : V(\widetilde H) \inj V(\bG)$ is an occurrence of $(\widetilde H,S,\tau)$, then we call $\phi : V(H) \inj V(\bG)$ an \emph{extension} of $\widetilde \phi$ if its an occurrence of $(H,S,\tau)$ and agrees with $\widetilde \phi$ on $V(\widetilde H)$. Let's write $\widetilde\Phi$ for the set of occurrences of $(\widetilde H,S,\tau)$ in $(\bG,\sigma)$, and for each $\widetilde \phi \in \widetilde \Phi$, write $\Xi(\widetilde\phi)$ for its set of extensions. Thus, incorporating Lemma 8.20,
    \begin{align*}
        &\left\|p_{(H,S,\tau)}(x,\bG) - n^{\cc(H) - \cc(\widetilde H)}\frac{(dM)^{(H,S)}_{\tau}}{(dM)^{(\widetilde H,S)}_\tau} p_{(\widetilde H,S,\tau)}(x,\bG)\right\|_1 \\
        &\qquad \qquad \le \sum_{\widetilde\phi \in \widetilde \Phi} \left| |\Xi(\widetilde\phi)| - n^{\cc(H) - \cc(\widetilde H)}        \prod_{v \in V(H)} \prod_{q = \widetilde\deg(v)}^{\deg(v) - 1}(d - q)\right|.
    \end{align*}
    
    By Proposition 8.15, with high probability there are $o(n^{\cc(\widetilde H)}$ occurrences of $\widetilde H$ whose  $|E(H)|$ neighborhoods in $\widetilde \bG$ either intersect or contain a cycle, so we can safely restrict the right hand side above to the remaining ones. Let's fix such an occurrence and enumerate the possible extensions. First, for each connected component $\widetilde J$ of $\widetilde H$, and its corresponding component $J$ of $H$, because $\bG$ is $d$-regular and locally treelike in the neighborhood of $\phi(\widetilde J)$, there are exactly
    $$
        \prod_{v \in V(J)} \prod_{q = \widetilde\deg(v)}^{\deg(v) - 1} (d-q)
    $$
    ways to extend $\widetilde \phi$ to the remainder of $J$. Having already chosen how to extend the occurrence on these connected components, call $K$ the union of all connected components in $H$ that have no distinguished vertex. We need to find an injective homomorphism from $K$ into $\bG$ that does not collide with $\widetilde \phi(\widetilde H)$ or the portion of the extension that we have already constructed. Since $|V(H)| = O(1)$, there are
    $$
        n^{\cc(H) - \cc(\widetilde H)}\prod_{v \in V(K)} \prod_{q = 0}^{\deg(v) - 1}(d-q) + O(n^{\cc(H) - \cc(\widetilde H) - 1})
    $$
    ways to do this. Since $|\widetilde \Phi| = O(n^{\cc(\widetilde H)})$, 
    $$
        \sum_{\widetilde\phi \in \widetilde \Phi} \left| |\Xi(\widetilde\phi)| - n^{\cc(H) - \cc(\widetilde H)}        \prod_{v \in V(H)} \prod_{q = \widetilde\deg(v)}^{\deg(v) - 1}(d - q)\right| = O(n^{\cc(\widetilde H)}) \cdot O(n^{\cc(H) - \cc(\widetilde H) - 1}) = O(n^{\cc(H) - 1})
    $$
    as desired.
\end{proof}

\begin{lemma}
    \label{lem:pruned-suffice}
    Let $\bG \sim \QQ$ and $\pseudo = \pseudo(\bG)$ be a degree-$(D_x, D_G)$ pseudoexpectation, perhaps dependent on $\bG$. If $\pseudo$ w.h.p.\ satisfies the moment constraints for pruned partially labelled forests, w.h.p.\ it does so for every partially labelled forest.
\end{lemma}

\begin{proof}
    This is a direct consequence of Lemma 8.21. Retaining $(H,S,\tau)$ and $(\widetilde H,S,\tau)$ from the proof of that lemma, since the pseudoexpectation of any monomial has absolute value at most one,
    $$
        \pseudo p_{(H,S,\tau)}(x,\bG) = \frac{(dM)_\tau^{(H,S)}}{(dM)_\tau^{(\widetilde H,S)}} p_{(\widetilde H,S,\tau)}(x,\bG) \pm o(n^\ell) = (dM)_\tau^{(H,S)} \pm \delta n^\ell
    $$
    for every $\delta > 0$. We can take a union bound over all finitely many $(H,S,\tau)$.
\end{proof}

We are finally ready to describe our own degree $(2,D)$ pseudoexpectation. Our key building block will be the feasible solution $\bP \succeq 0$ to the degree-$D$ Symmetric Path Statistics on input $\bG \sim \QQ$ SDP whose asymptotic almost sure existence is guaranteed us by Theorem 8.7. Recall that this matrix satisfies
\begin{enumerate}
    \item $\bP_{u,u} = 1$ for every $u \in [n]$
    \item $\langle \bP, \onesmat \rangle$ = 0
    \item $\langle \bP, \nb{s}{\bG} \rangle = q_s(d\lambda) \pm \delta n$ for every $s = 1,...,D$.
\end{enumerate}
A degree-$2$ pseudoexpectation $\pseudo : \RR[x]_{\le 2} \to \RR$ may be expressed as a $(1 + nk)\times(1 + nk)$ block matrix
$$
    \begin{pmatrix} 1 & l^\top \\ l & Q \end{pmatrix} = \begin{pmatrix} 
    1 & l^\top_1 & \cdots & l^\top_n \\ 
    l_1 & Q_{1,1} & \cdots & Q_{1,n} \\
    \vdots & \vdots & \ddots & \vdots \\
    l_n & Q_{n,1} & \cdots & Q_{n,n}
    \end{pmatrix}
$$
where $(l_u)_i = \pseudo x_{u,i}$ and $(Q_{u,v})_{i,j} = \pseudo x_{u,i}x_{v,j}$. Our construction will set $(l_u)_i = 1/k$ for every $u$ and $i$, and
$$
    Q = \frac{1}{k}\left(\onesmat_{nk}/k + \bP \otimes (\id - \onesmat_k/k) \right)
$$

Let us first check the hard consstraints. For positivity it suffices to observe that
$$
    Q - ll^\top = \frac{1}{k} \bP \otimes (\id - \onesmat_k/k) \succeq 0,
$$
as $\bP,(\id - \onesmat_k/k) \succeq 0$. We also have
$$
    \pseudo x_{u,i}^2 = (Q_{u,u})_{i,i} = \frac{1}{k}\left(1/k + (1 - 1/k)\bP_{i,i}\right) = 1/k = \pseudo x_{u,i}
$$
since $\bP_{u,u} = 1$ for every $u\in [n]$. It remains only to check that $\pseudo (x_{u,1} + \cdots + x_{u,k})p(x) = \pseudo p(x)$ for every $p(x)$ of degree one. For this it is sufficient to verify that
$$
    \pseudo (x_{u,1} + \cdots + x_{u,k})x_{v,j} = \sum_i (Q_{u,v})_{i,j} = \frac{1}{k}\sum_i (1/k + \bP_{u,v}(1 - \onesmat_k/k)_{i,j}) = 1/k = \pseudo x_{v,j}.
$$

Finally, we need to verify the moment constraints. By Lemmas 8.21-22, if we do so only on the minimal partially labelled forests with at most two distinguished vertices, then w.h.p.\ the remainder of the moment constraints are satisfied. A minimal partially labelled forest with one distinguished vertex is just a single vertex with a label $i \in [k]$; the associated polynomial in this case is just $x_{1,k} + \cdots + x_{n,k}$, and its required pseudoexpectation is $n/k \pm \delta n$, since there are identically $n/k$ vertices with each label in the planted model. Our pseudoexpectation assigns a value of
$$
    \pseudo \sum_{u}x_{u,i} = \sum_u (l_u)_i = n/k
$$
as desired. 

A minimal partially labelled forest on two vertices is either a path of length $s \in [d]$ with endpoints labelled $i,j \in [k]$, or two isolated vertices labelled $i,j \in [k]$. In the former case, the pseudoexpectation is required to read $\tfrac{1}{k}q(dM)_{i,j}n \pm \delta n$ Recycling some calculations from Section 8.4, our pseudoexpectation on this polynomial reads
\begin{align*}
    \langle Q_{i,j}, \sa{s}{\bG}\rangle 
    &= \langle Q_{i,j},\nb{s}{\bG}\rangle \pm o(n) \\
    &= \frac{1}{k}\left(\langle \onesmat_n/k, \nb{s}{\bG}\rangle + \langle \bP, \nb{s}{\bG}\rangle (\id - \onesmat_k/k)_{i,j}\right) \\
    &= \frac{1}{k}\left(q_s(d)(\onesmat/k)_{i,j} + q_s(d\lambda)(\id - \onesmat_k/k)_{i,j}\right)n \pm (\delta/k)n  \\
    &= \frac{1}{k}q_s(dM)_{i,j} \pm (\delta/k)n.
\end{align*}
The last line follows since $dM = d\onesmat_k/k + d\lambda (\id - \onesmat_k/k)$ is the spectral decomposition of $dM$. 

Finally, we verify the case of two disjoint vertices labelled $i,j\in[k]$. The polynomial here is $\sum_{u\neq v} x_{u,i}x_{v,j}$, and the pseudoexpectation is requried to give a value of $(n/k)^2 \pm \delta n^2$. As needed, our pseudoexpectation gives
$$
    \langle Q_{i,j},\onesmat_n \rangle = \frac{1}{k}\left(\langle \onesmat_n/k, \onesmat_n\rangle + \langle \bP,\onesmat_n\rangle(\id - \onesmat_k/k)_{i,j}\right) = (n/k)^2,
$$
as $\langle \bP,\onesmat_n\rangle = 0$.

\section{Low-Degree Analysis of Spiked Models}
\label{sec:low-deg}

In this section, we develop machinery for low-degree analysis of general spiked Wigner and Wishart models, culminating in the proofs of Theorems~\ref{thm:general-wigner} and \ref{thm:general-wishart}.

\subsection{Preliminaries}

\begin{definition}
    A random vector $\bx$ is $\varepsilon$-local $c$-subgaussian if for any fixed vector $v$ with $\|v\| \le \varepsilon$,
    \[ 
        \EE \exp(\langle v,\bx \rangle) \le \exp\left(\frac{c}{2} \|v\|^2\right). 
    \]
\end{definition}

\noindent It is straightforward to verify the following fact.

\begin{fact}\label{fac:subg}
    Suppose $\bx$ is $\varepsilon$-local $c$-subgaussian. For a (non-random) scalar $\alpha \ne 0$, $\alpha \bx$ is $\varepsilon/|\alpha|$-local $c \alpha^2$-subgaussian. Also, the sum of $n$ independent copies of $\bx$ is $\varepsilon$-local $cn$-subgaussian.
\end{fact}

\begin{proposition}
    \label{prop:local-chernoff}
    If a random vector $\bx$ is $\varepsilon$-local $c$-subgaussian then it admits the following local Chernoff bound: for any $\|v\| = 1$ and $0 \le t \le \varepsilon c$,
    \[ 
        \Pr\{\langle v,\bx \rangle \ge t\} \le \exp\left(-\frac{t^2}{2c}\right).
    \]
\end{proposition}

\begin{proof}
    Apply the standard Chernoff bound argument: for any $\alpha > 0$,
    \begin{align*}
        \Pr\{\langle v,\bx \rangle \ge t\} 
        &= \Pr\{\exp(\alpha \langle v,\bx \rangle) \ge \exp(\alpha t)\} \\
        &\le \EE[\exp(\alpha \langle v,\bx \rangle)] / \exp(\alpha t) \\
        &\le \exp(c \alpha^2/2-\alpha t) & & \text{provided } \alpha \le \varepsilon.
    \end{align*}
    Set $\alpha = t/c$ to complete the proof.
\end{proof}

\begin{proposition}
    \label{prop:norm-bound}
    If a random vector $\bx \in \RR^k$ is $\varepsilon$-local $c$-subgaussian then for any $\delta > 0$ and any $0 \le t \le \varepsilon c / (1-\delta)$,
    \[ \Pr\{ \|x\| \ge t \} \le C(\delta,k) \exp\left(-\frac{1}{2c} (1-\delta)^2 t^2 \right) \]
    where $C(\delta,k)$ is a constant depending only on $\delta$ and $k$.
\end{proposition}
\begin{proof}
    Let $\mathcal{N} \subseteq \RR^k$ be a $\delta$-net of the unit sphere in $\RR^k$, in the sense that for any $v \in \RR^k$,
    \[ 
        \max_{u \in \mathcal{N}}\, \langle u,v \rangle \ge (1-\delta) \|v\| 
    \]
    where $\|u\| = 1$ for all $u \in \mathcal{N}$. Let $C(\delta,k) = |\mathcal{N}|$. Using a union bound and the local Chernoff bound (Proposition~\ref{prop:local-chernoff}), for all $0 \le t \le \varepsilon c / (1-\delta)$,
    \[ 
        \Pr\{ \|\bx\| \ge t \} \le \Pr\left\{ \max_{u \in \mathcal{N}}\, \langle u,\bx \rangle \ge (1-\delta) t \right\} \le C(\delta,k) \exp\left(-\frac{1}{2c} (1-\delta)^2 t^2 \right). \qedhere
    \]
\end{proof}

\begin{proposition}
    \label{prop:exp-norm-bound}
    Let $\delta > 0$. If a random vector $\bx \in \RR^k$ is $\varepsilon$-local $c$-subgaussian with $c < (1-\delta)^2/2$ then
    \[ 
        \EE \left[\indicator{\|\bx\| \le \varepsilon c/(1-\delta)} \exp(\|\bx\|^2)\right] \le 1 + \frac{C(\delta,k)}{(1-\delta)^2/(2c) - 1} 
    \]
    where $C(\delta,k)$ is a constant depending only on $\delta$ and $k$.
\end{proposition}

\begin{proof}
    Let $\Delta = \varepsilon c/(1-\delta)$, and integrate the tail bound from Proposition~\ref{prop:norm-bound}:
    \begin{align*}
        \EE \left[\indicator{\|\bx\| \le \Delta}\exp(\|\bx\|^2)\right] &= \int_0^\infty \Pr\{\indicator{\|\bx\| \le \Delta}\exp(\|\bx\|^2) \ge t\} \,dt \\
        &\le 1 + \int_1^{\exp(\Delta^2)} \Pr\{\exp(\|\bx\|^2) \ge t\} \,dt \\
        &= 1 + \int_1^{\exp(\Delta^2)} \Pr\{\|\bx\| \ge \sqrt{\log t}\} \,dt \\
        &\le 1 + C(\delta,k) \int_1^\infty \exp\left(-\frac{1}{2c} (1-\delta)^2 \log t \right) dt \\
        &= 1 + \frac{C(\delta,k)}{(1-\delta)^2/(2c) - 1}.
    \end{align*}
\end{proof}

\subsection{The Wigner Model}

\begin{proof}[Proof of Theorem~\ref{thm:general-wigner}]
    We start with a formula for $\|L^{\le D}\|^2$ from \cite{lowdeg-notes} (adapted slightly for the case of symmetric Gaussian noise):
    \begin{equation}\label{eq:L-wig}
        \|L^{\le D}\|^2 = \Ex_{\bX,\bX'} \exp^{\le D} \left(\frac{\lambda^2 n}{2}\langle \bX,\bX' \rangle\right)
    \end{equation}
    where $\bX'$ is an independent copy of $\bX$ and $\exp^{\le D}(x) = \sum_{d=0}^D \frac{x^d}{d!}$ denotes the Taylor series truncation of $\exp$. Write $\|L^{\le D}\|^2 = L_1 + L_2$ where $L_1$ is the \emph{small deviations} term
    \[
        L_1 \colonequals \Ex_{\bX,\bX'} \indicator{\langle \bX,\bX' \rangle \le \Delta}\exp^{\le D} \left(\frac{\lambda^2 n}{2}\langle \bX,\bX' \rangle\right),
    \]
    and $L_2$ is the \emph{large deviations} term
    \[ 
        L_2 \colonequals \Ex_{\bX,\bX'} \indicator{\langle \bX,\bX' \rangle > \Delta}\exp^{\le D} \left(\frac{\lambda^2 n}{2}\langle \bX,\bX' \rangle\right), \]
    where $\Delta > 0$ is a small constant to be chosen later. Lemmas~\ref{lem:wig-small-dev} and \ref{lem:wig-large-dev}, proved in the following two subsections, show that for some choice of $\Delta$, $L_1$ and $L_2$ are both $O(1)$, completing the proof.
\end{proof}

\subsubsection{Small Deviations}

\begin{lemma}\label{lem:wig-small-dev}
    In the setting of Theorem~\ref{thm:general-wigner}, if $|\lambda| < 1$ then there exists $\Delta > 0$ such that $L_1 = O(1)$ for any $D$.
\end{lemma}

\begin{proof}
    Note that
    \begin{equation}
        \label{eq:XXR}
        \langle \bX,\bX' \rangle = \frac{1}{n^2} \langle \bU \bU^\top, \bU'(\bU')^\top \rangle = \frac{1}{n^2} \|\bR\|_\F^2
    \end{equation}
    where $\bR = \bU^\top \bU'$. In particular, $\langle \bX,\bX' \rangle \ge 0$. Since $\exp^{\le D}(x) \le \exp(x)$ for all $x \ge 0$,
    \begin{equation}
        \label{eq:L1-wig}
        L_1 \le \EE \,\indicator{\|\bR\|_\F^2 \le \Delta n^2} \exp\left(\frac{\lambda^2}{2n}\|\bR\|_\F^2 \right).
    \end{equation}
    We have $\bR = \sum_{i=1}^n \bR_i$ where the $\bR_i$ are independent $k \times k$ matrices, each distributed as $\bpi (\bpi')^\top$. Since $\bpi$ has bounded support, the moment-generating function $M(T) = \EE \exp(\langle T, \bR_i \rangle)$ exists in a neighborhood of $T = 0$ (in fact, it exists everywhere) and thus, by the defining property of the MGF, has gradient $\nabla M(0) = \EE[\bR_i] = 0$ and Hessian $(\mathrm{Hess}\, M)(0) = \cov(\bR_i) = \cov(\bpi)^{\otimes 2} \preceq \id_{k^2}$. Thus for any $\eta > 0$ there exists $\varepsilon > 0$ such that
    \[ 
        M(T) \le \exp\left(\frac{1}{2} (1+\eta) \|T\|_\F^2\right) \qquad \text{for all } \|T\|_\F \le \varepsilon. 
    \]
    In other words, $\bR_i$ is $\varepsilon$-local $(1+\eta)$-subgaussian. From Fact~\ref{fac:subg}, this implies $\bR$ is $\varepsilon$-local $(1+\eta)n$-subgaussian, and $\lambda \bR/\sqrt{2n}$ is $\varepsilon \sqrt{2n}/\lambda$-local $(1+\eta)\lambda^2/2$-subgaussian. Since $\lambda^2 < 1$, we can choose $\delta > 0$ and $\eta > 0$ such that $(1+\eta)\lambda^2 < (1-\delta)^2$. Letting $\Delta = \big(\varepsilon (1+\eta) / (1-\delta)\big)^2$ and using Proposition~\ref{prop:exp-norm-bound},
    \[ 
        L_1 \le \EE \,\indicator{\|\lambda \bR/\sqrt{2n}\|_\F \le \lambda\sqrt{\Delta n/2}} \exp\left(\|\lambda \bR/\sqrt{2n}\|_\F^2 \right) = O(1). \qedhere
    \]
\end{proof}

\subsubsection{Large Deviations}

\begin{lemma}\label{lem:wig-large-dev}
    In the setting of Theorem~\ref{thm:general-wigner}, for any constants $\lambda \in \RR$ and $\Delta > 0$, and for any $D = o(n/\log n)$, we have $L_2 = o(1)$.
\end{lemma}

\begin{proof}
    Recall from above that $\bR$ is $\varepsilon$-local $(1+\eta)n$-subgaussian. By Proposition~\ref{prop:norm-bound} (taking $t$ to be $n$ times a small constant),
    \begin{equation}
        \label{eq:large-dev-tail}
        \Pr\{ \langle \bX,\bX' \rangle > \Delta \} \le \Pr\{ \|\bR\|_\F > \sqrt{\Delta} n \} = \exp(-\Omega(n)).
    \end{equation}
    The boundedness of $\bpi$ guarantees $|\langle \bX,\bX' \rangle| \le n^C$ for some constant $C > 0$, and so
    \[ 
        L_2 \le \exp(-\Omega(n)) \sum_{d=0}^D \left(\frac{\lambda^2 n}{2} n^C\right)^d \le \exp(-\Omega(n)) (D+1) n^{O(D)}, 
    \]
    which is $o(1)$ provided $D = o(n / \log n)$.
\end{proof}

\subsection{The Wishart Model}\label{sec:wishart}

\begin{proof}[Proof of Theorem~\ref{thm:general-wishart}]
    In Appendix~\ref{app:wishart-LDLR} we derive a formula for $\|L^{\le D}\|^2$ in the general spiked Wishart model. The version we will need here is summarized in Proposition~\ref{prop:wish-facts}. The formula takes the form
    \begin{equation}
        \label{eq:wishart-LDLR}
        \|L^{\le D}\|^2 = \Ex_{\bX,\bX'} \sum_{d=0}^D r_d(\beta \bX,\beta \bX')
    \end{equation}
    for some polynomials $r_d$. As in the Wigner case, we write $\|L^{\le D}\|^2 = L_1 + L_2$ where
    \[ 
        L_1 \colonequals \Ex_{\bX,\bX'} \indicator{\langle \bX,\bX' \rangle \le \Delta} \sum_{d=0}^D r_d(\beta \bX,\beta \bX') \]
    and
    \[ 
        L_2 \colonequals \Ex_{\bX,\bX'} \indicator{\langle \bX,\bX' \rangle > \Delta} \sum_{d=0}^D r_d(\beta \bX,\beta \bX') \]
    for a small constant $\Delta > 0$ to be chosen later. Lemmas~\ref{lem:wish-small-dev} and \ref{lem:wish-large-dev}, proved in the following two subsections, show that for some choice of $\Delta$, $L_1$ and $L_2$ are both $O(1)$, completing the proof.
\end{proof}

\subsubsection{Small Deviations}

Before bounding the small deviations term in Lemma~\ref{lem:wish-small-dev}, we state two deterministic facts that will be useful in the proof.

\begin{proposition}
    \label{prop:det-cycle}
    For $n \times m$ matrices $A$ and $B$,
    \[ 
        \det(\id_n - AA^\top BB^\top) = \det(\id_m - A^\top B B^\top A). 
    \]
    \end{proposition}
\begin{proof}
    If $A,B$ are square and $A$ is nonsingular, this can be shown by taking determinants on both sides of the equation $(I - AA^\top BB^\top)A = A(I - A^\top BB^\top A)$. For the general case, pad $A,B$ with zeros to make them square, and consider a sequence of nonsingular matrices converging to $A$.
\end{proof}

\begin{lemma}
    \label{lem:Delta}
    For any $\eta > 0$ there exists $\varepsilon > 0$ such that for all $0 \le t \le \varepsilon$, we have $(1-t)^{-1} \le \exp((1+\eta)t)$.
    \end{lemma}
\begin{proof}
    Letting $f(t) = (1-t)^{-1}$ and $g(t) = \exp((1+\eta)t)$, we have $f(0) = g(0) = 1$ and $f'(0) = 1 < 1+\eta = g'(0)$.
\end{proof}

\noindent We are now ready to bound the small deviations term.

\begin{lemma}
    \label{lem:wish-small-dev}
    In the setting of Theorem~\ref{thm:general-wishart}, if $\beta^2 < \gamma$ then there exists $\Delta > 0$ such that $L_1 = O(1)$ for any $D$.
\end{lemma}

\begin{proof}
    Since $\bX \succeq 0$, we have from Proposition~\ref{prop:wish-facts}(c,d) that $r_d(\beta \bX,\beta \bX') \ge 0$. Analogous to~\eqref{eq:XXR}, we have either $\langle \bX,\bX' \rangle = 0$ or $\langle \bX,\bX' \rangle = \frac{1}{n^2} \|\bR\|_\F^2$. 
    Recall from Definition~\ref{def:general-wishart-model} that, when drawing from $\PP$, we first draw $\widetilde \bX \sim \sX$, and then threshold to form $\bX$ as:
    \begin{equation}
    \label{eq:wishart-cases-2}
        \bX = \begin{cases}
            \widetilde \bX & \text{if }\beta \widetilde \bX \succ -\id_n, \\
            0 & \text{else}.
        \end{cases}
    \end{equation}
    
    \noindent We can upper bound $L_1$ by dropping the low-degree truncation and using Proposition~\ref{prop:wish-facts}(b):
    \begin{align*}
        L_1 &\le \EE_{\bX,\bX'} \,\indicator{\langle \bX,\bX' \rangle \le \Delta} \sum_{d=0}^\infty r_d(\beta \bX,\beta \bX') \\
        &= \EE_{\bX,\bX'} \,\indicator{\langle \bX,\bX' \rangle \le \Delta} \det(\id_n - \beta^2 \bX\bX')^{-N/2} \\
        &\le 1 + \EE \,\indicator{\|\bR\|_\F^2 \le \Delta n^2} \det(\id_n - \beta^2 \widetilde \bX \widetilde \bX')^{-N/2}
    \end{align*}
    where the $+1$ in the last line covers the second case of~\eqref{eq:wishart-cases-2}. 
    Let $\{\blambda_i\}$ be the eigenvalues of $\bR\bR^\top$ (which are nonnegative). Let $\eta > 0$ (to be chosen later), take $\varepsilon$ according to Lemma~\ref{lem:Delta}, and choose $\Delta \le \varepsilon/\beta^2$. Note that $\|\bR\|_\F^2 = \sum_i \blambda_i$. Provided $\|\bR\|_\F^2 \le \Delta n^2$, we have $\blambda_i \le \Delta n^2$ for all $i$, i.e., $\beta^2 n^{-2} \blambda_i \le \varepsilon$, and so
    \begin{align*}
        \det(\id_n - \beta^2 \widetilde \bX \widetilde \bX')^{-N/2} &= \det(\id_n - \beta^2 n^{-2} \bU\bU^\top \bU' \bU'^\top)^{-N/2} \\
        &= \det(\id_n - \beta^2 n^{-2} \bU^\top \bU' \bU'^\top \bU)^{-N/2}  \tag{using Proposition~\ref{prop:det-cycle}}\\
        &= \det(\id_n - \beta^2 n^{-2} \bR\bR^\top)^{-N/2} \\
        &= \prod_i (1 - \beta^2 n^{-2} \blambda_i)^{-N/2} \\
        &\le \prod_i \exp((1+\eta)\beta^2 n^{-2} \blambda_i)^{N/2} \tag{using Lemma~\ref{lem:Delta}} \\
        &= \exp\left((1+\eta)\frac{\beta^2 N}{2n^2} \sum_i \blambda_i \right) \\
        &= \exp\left((1+\eta)\frac{\beta^2 N}{2n^2} \|\bR\|_\F^2 \right).
    \end{align*}
    We now have
    \[ 
        L_1 \le 1 + \EE \,\indicator{\|\bR\|_\F^2 \le \Delta n^2} \exp\left((1+\eta)\frac{\beta^2 N}{2n^2} \|\bR\|_\F^2 \right). 
    \]
    Comparing this to~\eqref{eq:L1-wig}, we see that we have reduced to the case of Wigner small deviations. In place of $\lambda^2$, we have $(1+\eta) \frac{\beta^2 N}{n} \to (1+\eta) \frac{\beta^2}{\gamma}$. Thus, provided $\beta^2 < \gamma$, we can choose $\eta > 0$ and $\Delta > 0$ small enough so that $L_1 = O(1).$
\end{proof}

\subsubsection{Large Deviations}

\begin{lemma}
    \label{lem:wish-large-dev}
    In the setting of Theorem~\ref{thm:general-wishart}, for any constants $\beta > -1$, $\gamma > 0$, and $\Delta > 0$, and for any $D = o(n/\log n)$, we have $L_2 = o(1)$.
\end{lemma}

\begin{proof}
    Consider the case $\beta > 0$ so that $\beta \bX \succeq 0$; the case $\beta < 0$ is handled similarly. We will use the formula for $r_d$ given in Proposition~\ref{prop:wish-facts}(c): 
    \[ 
        r_d(\beta \bX,\beta \bX') = \sum_{\substack{d_1, \dots, d_N \in 2\NN \\ \sum_{i = 1}^N d_i = d}} \prod_{i = 1}^N \frac{1}{d_i!} \Ex_{\substack{\bx \sim \sN( 0, \beta \bX) \\ \bx' \sim \sN( 0, \beta \bX')}} \la \bx, \bx' \ra^{d_i}. 
    \]
    
    \noindent Let $\{\blambda_i,\bv_i\}$ be an eigendecomposition of $\bX$, and let $\bX_i = \blambda_i \bv_i \bv_i^\top$ so that $\bX = \bX_1 + \cdots + \bX_k$. Let $\bx_i \sim \mathcal{N}(0,\beta \bX_i)$ independently so that $\bx \colonequals \sum_{i=1}^k \bx_i \sim \mathcal{N}(0,\beta \bX)$. Similarly define $\bX_i', \bx_i', \bx'$. For fixed $\bX,\bX',\{\bX_i\},\{\bX'_i\}$,
    \[ 
        \EE [\langle \bx,\bx' \rangle^d] 
        \le \EE \|\bx\|^d \|\bx'\|^d = \EE \|\bx\|^d \cdot \EE \|\bx'\|^d. \]
    The boundedness of $\bpi$ guarantees $\sum_i \blambda_i^2 = \|\bX\|_\F^2 \le n^C$ for some constant $C > 0$. 
    Recall that $\bx_i \sim \sN(0, \beta \blambda_i \bv_i\bv_i^{\top})$, so for $\bg$ a standard Gaussian vector, $\bx_i$ has the same law as $\sqrt{\beta \blambda_i} \cdot \langle \bv_i, \bg \rangle \bv_i$.
    Since $\bv_i$ is a unit vector, $\|\bx_i\|$ then has the same law as $\sqrt{\beta \blambda_i} \cdot |\langle \bv_i, \bg \rangle|$, which in turn has the same law as $\sqrt{\beta \blambda_i} \cdot |\bg_1|$.
    So, for $d$ even, $\EE\|\bx_i\|^d = (\beta\blambda_i)^{d/2} (d-1)!! \le \beta^{d/2} n^{Cd/4} d^d$.
    This means
    \[ 
        \EE\|\bx\|^d 
        \le \EE \left(\sum_{i=1}^k \|\bx_i\| \right)^d 
        \le \EE\left(k \max_i \|\bx_i\|\right)^d 
        \le k^d\, \EE \sum_{i=1}^k \|\bx_i\|^d 
        \le k^{d+1} \beta^{d/2} n^{Cd/4} d^d. 
    \]
    Since $\EE\|\bx\|^d = 1$ when $d = 0$, we can rewrite this as $\EE\|\bx\|^d \le k^{2d} \beta^{d/2} n^{Cd/4} d^d \le (dn)^{O(d)}$. Thus for $d_1,\ldots,d_N \in 2\mathbb{N}$ with $\sum_i d_i = d \le D$,
    \[ 
        \prod_{i=1}^N \EE[\langle \bx,\bx' \rangle^{d_i}] \le \prod_{i=1}^N (d_i n)^{O(d_i)} 
        \le (Dn)^{O(D)}. 
    \]
    
    \noindent Now we have
    \[ 
        r_d(\bX,\bX') \le N^D (Dn)^{O(D)} 
    \]
    and so
    \[ 
        L_2 \le \Pr\{\langle \bX,\bX' \rangle > \Delta\} \sum_{d=0}^D N^D (Dn)^{O(D)}. 
    \]
    Similarly to~\eqref{eq:large-dev-tail} we have $\Pr\{\langle \bX,\bX' \rangle > \Delta\} \le \exp(-\Omega(n))$ and so $L_2 = o(1)$ provided $D = o(n/\log n)$.
\end{proof}

\bibliographystyle{alpha}
\bibliography{main}

\newcommand{\etalchar}[1]{$^{#1}$}
\begin{thebibliography}{DKMZ11b}

\bibitem[AB{\v{C}}13]{ABC}
Antonio Auffinger, G{\'e}rard {Ben Arous}, and Ji{\v{r}}{\'\i} {\v{C}}ern{\`y}.
\newblock Random matrices and complexity of spin glasses.
\newblock {\em Communications on Pure and Applied Mathematics}, 66(2):165--201,
  2013.

\bibitem[AKS98]{Alon_HiddenClique_98}
N.~Alon, M.~Krivelevich, and B.~Sudakov.
\newblock Finding a large hidden clique in a random graph.
\newblock {\em Random Structures \& Algorithms}, 13:457--466, 1998.

\bibitem[AS16]{AS-acyclic}
Emmanuel Abbe and Colin Sandon.
\newblock Achieving the {KS} threshold in the general stochastic block model
  with linearized acyclic belief propagation.
\newblock In {\em Advances in Neural Information Processing Systems}, pages
  1334--1342, 2016.

\bibitem[Bar17]{barucca}
Paolo Barucca.
\newblock Spectral partitioning in equitable graphs.
\newblock {\em Phys. Rev. E}, 95:062310, 2017.

\bibitem[BBH18]{BBH}
Matthew Brennan, Guy Bresler, and Wasim Huleihel.
\newblock Reducibility and computational lower bounds for problems with planted
  sparse structure.
\newblock In {\em Conference On Learning Theory}, pages 48--166, 2018.

\bibitem[BBP05]{bbp}
Jinho Baik, G{\'e}rard {Ben Arous}, and Sandrine P{\'e}ch{\'e}.
\newblock Phase transition of the largest eigenvalue for nonnull complex sample
  covariance matrices.
\newblock {\em The Annals of Probability}, 33(5):1643--1697, 2005.

\bibitem[BC19]{bordenave2019eigenvalues}
Charles Bordenave and Beno{\^\i}t Collins.
\newblock Eigenvalues of random lifts and polynomials of random permutation
  matrices.
\newblock {\em Annals of Mathematics}, 190(3):811--875, 2019.

\bibitem[BDG{\etalchar{+}}16]{BDGHT-2016-RecoveryRigidityRegularSBM}
Gerandy Brito, Ioana Dumitriu, Shirshendu Ganguly, Christopher Hoffman, and
  Linh~V Tran.
\newblock Recovery and rigidity in a regular stochastic block model.
\newblock In {\em Proceedings of the twenty-seventh annual ACM-SIAM symposium
  on Discrete algorithms}, pages 1589--1601. Society for Industrial and Applied
  Mathematics, 2016.

\bibitem[BHK{\etalchar{+}}19]{BHKKMP-2019-PlantedClique}
Boaz Barak, Samuel Hopkins, Jonathan Kelner, Pravesh~K Kothari, Ankur Moitra,
  and Aaron Potechin.
\newblock A nearly tight sum-of-squares lower bound for the planted clique
  problem.
\newblock {\em SIAM Journal on Computing}, 48(2):687--735, 2019.

\bibitem[BKM17]{BKM-2017-LovaszThetaRandomGraphs}
Jess Banks, Robert Kleinberg, and Cristopher Moore.
\newblock The {Lov\'{a}sz} theta function for random regular graphs and
  community detection in the hard regime.
\newblock {\em arXiv preprint arXiv:1705.01194}, 2017.

\bibitem[BKW20]{BKW-2019-ConstrainedPCA}
Afonso~S Bandeira, Dmitriy Kunisky, and Alexander~S Wein.
\newblock Computational hardness of certifying bounds on constrained pca
  problems.
\newblock In {\em 11th Innovations in Theoretical Computer Science Conference
  (ITCS 2020)}. Schloss Dagstuhl-Leibniz-Zentrum f{\"u}r Informatik, 2020.

\bibitem[BMR19]{banks2019local}
Jess Banks, Sidhanth Mohanty, and Prasad Raghavendra.
\newblock Local statistics, semidefinite programming, and community detection.
\newblock {\em arXiv preprint arXiv:1911.01960}, 2019.

\bibitem[BR13]{BR-reduction}
Quentin Berthet and Philippe Rigollet.
\newblock Complexity theoretic lower bounds for sparse principal component
  detection.
\newblock In {\em Conference on Learning Theory}, pages 1046--1066, 2013.

\bibitem[BS06]{BS-spiked}
Jinho Baik and Jack~W Silverstein.
\newblock Eigenvalues of large sample covariance matrices of spiked population
  models.
\newblock {\em Journal of multivariate analysis}, 97(6):1382--1408, 2006.

\bibitem[BS16]{sos-notes}
Boaz Barak and David Steurer.
\newblock Proofs, beliefs, and algorithms through the lens of sum-of-squares.
\newblock {\em Course notes: http://www. sumofsquares. org/public/index. html},
  2016.

\bibitem[CDF09]{CDF-wigner}
Mireille Capitaine, Catherine {Donati-Martin}, and Delphine F{\'e}ral.
\newblock The largest eigenvalues of finite rank deformation of large wigner
  matrices: convergence and nonuniversality of the fluctuations.
\newblock {\em The Annals of Probability}, 37(1):1--47, 2009.

\bibitem[CO03]{CO-2003-LovaszThetaRandomGraphs}
Amin Coja-Oghlan.
\newblock The {Lov{\'a}sz} number of random graphs.
\newblock In {\em Approximation, Randomization, and Combinatorial
  Optimization.. Algorithms and Techniques}, pages 228--239. Springer, 2003.

\bibitem[COEH16]{COEH-2016-ChromaticNumberRandomRegular}
Amin Coja-Oghlan, Charilaos Efthymiou, and Samuel Hetterich.
\newblock On the chromatic number of random regular graphs.
\newblock {\em Journal of Combinatorial Theory, Series B}, 116:367--439, 2016.

\bibitem[DAM17]{DAM}
Yash Deshpande, Emmanuel Abbe, and Andrea Montanari.
\newblock Asymptotic mutual information for the balanced binary stochastic
  block model.
\newblock {\em Information and Inference: A Journal of the IMA}, 6(2):125--170,
  2017.

\bibitem[DDSW03]{Diaz_03}
Josep Díaz, Norman Do, Maria Serna, and Nicholas Wormald.
\newblock Bounds on the max and min bisection of random cubic and random
  4-regular graphs.
\newblock {\em Theoretical Computer Science}, 307:531--547, 10 2003.

\bibitem[DKMZ11a]{sbm-ks-2}
Aurelien Decelle, Florent Krzakala, Cristopher Moore, and Lenka Zdeborov{\'a}.
\newblock Asymptotic analysis of the stochastic block model for modular
  networks and its algorithmic applications.
\newblock {\em Physical Review E}, 84(6):066106, 2011.

\bibitem[DKMZ11b]{sbm-ks-1}
Aurelien Decelle, Florent Krzakala, Cristopher Moore, and Lenka Zdeborov{\'a}.
\newblock Inference and phase transitions in the detection of modules in sparse
  networks.
\newblock {\em Physical Review Letters}, 107(6):065701, 2011.

\bibitem[DMM09]{amp}
David~L Donoho, Arian Maleki, and Andrea Montanari.
\newblock Message-passing algorithms for compressed sensing.
\newblock {\em Proceedings of the National Academy of Sciences},
  106(45):18914--18919, 2009.

\bibitem[DMS17]{DMS-2017-CutsSparseRandomGraphs}
Amir Dembo, Andrea Montanari, and Subhabrata Sen.
\newblock Extremal cuts of sparse random graphs.
\newblock {\em The Annals of Probability}, 45(2):1190--1217, 2017.

\bibitem[FGR{\etalchar{+}}17]{sq-clique}
Vitaly Feldman, Elena Grigorescu, Lev Reyzin, Santosh~S Vempala, and Ying Xiao.
\newblock Statistical algorithms and a lower bound for detecting planted
  cliques.
\newblock {\em Journal of the ACM (JACM)}, 64(2):1--37, 2017.

\bibitem[FP07]{fp}
Delphine F{\'e}ral and Sandrine P{\'e}ch{\'e}.
\newblock The largest eigenvalue of rank one deformation of large wigner
  matrices.
\newblock {\em Communications in mathematical physics}, 272(1):185--228, 2007.

\bibitem[Fri03]{Friedman-2003-SecondEigenvalue}
Joel Friedman.
\newblock A proof of {Alon's} second eigenvalue conjecture.
\newblock In {\em Proceedings of the thirty-fifth annual ACM symposium on
  Theory of computing}, pages 720--724, 2003.

\bibitem[GM75]{Grimmett_McDiarmid-75}
G.~R. Grimmett and C.~J.~H. McDiarmid.
\newblock On colouring random graphs.
\newblock {\em Math. Proc. Camb. Phil. Soc.}, 77:313--324, 1975.

\bibitem[Gue03]{Guerra-2003-BrokenRSB}
Francesco Guerra.
\newblock Broken replica symmetry bounds in the mean field spin glass model.
\newblock {\em Communications in mathematical physics}, 233(1):1--12, 2003.

\bibitem[GZ17]{gz-ogp}
David Gamarnik and Ilias Zadik.
\newblock High-dimensional regression with binary coefficients. {Estimating}
  squared error and a phase transition.
\newblock {\em arXiv preprint arXiv:1701.04455}, 2017.

\bibitem[HKP{\etalchar{+}}17]{HKPRSS-2017-SOSSpectral}
Samuel~B Hopkins, Pravesh~K Kothari, Aaron Potechin, Prasad Raghavendra, Tselil
  Schramm, and David Steurer.
\newblock The power of sum-of-squares for detecting hidden structures.
\newblock In {\em 2017 IEEE 58th Annual Symposium on Foundations of Computer
  Science (FOCS)}, pages 720--731. IEEE, 2017.

\bibitem[Hof70]{Hoffman-1970-Eigenvalues}
Alan~J Hoffman.
\newblock On eigenvalues and colorings of graphs.
\newblock {\em Graph Theory and its Applications}, 1970.

\bibitem[Hop18]{sam-thesis}
Samuel Hopkins.
\newblock {\em Statistical Inference and the Sum of Squares Method}.
\newblock PhD thesis, Cornell University, 2018.

\bibitem[HS17]{HS-bayesian}
Samuel~B Hopkins and David Steurer.
\newblock Bayesian estimation from few samples: community detection and related
  problems.
\newblock {\em arXiv preprint arXiv:1710.00264}, 2017.

\bibitem[JKS18]{JKS-2018-kCutPotts}
Aukosh Jagannath, Justin Ko, and Subhabrata Sen.
\newblock Max $\kappa$-cut and the inhomogeneous potts spin glass.
\newblock {\em Annals of Applied Probability}, 28(3):1536--1572, 2018.

\bibitem[Kar76]{Karp76}
Richard~M. Karp.
\newblock The probabilistic analysis of some combinatorial search algorithms.
\newblock 1976.

\bibitem[Kar86]{KarpTuring}
Richard~M. Karp.
\newblock Combinatorics, complexity, and randomness.
\newblock 1986.

\bibitem[Kuc95]{Kucera_95}
Ludek Kucera.
\newblock Expected complexity of graph partitioning problems.
\newblock {\em Discrete Applied Mathematics}, 57:193--212, 1995.

\bibitem[KWB19]{lowdeg-notes}
Dmitriy Kunisky, Alexander~S Wein, and Afonso~S Bandeira.
\newblock Notes on computational hardness of hypothesis testing: Predictions
  using the low-degree likelihood ratio.
\newblock {\em arXiv preprint arXiv:1907.11636}, 2019.

\bibitem[KZ09]{quiet-1}
Florent Krzakala and Lenka Zdeborov{\'a}.
\newblock Hiding quiet solutions in random constraint satisfaction problems.
\newblock {\em Physical review letters}, 102(23):238701, 2009.

\bibitem[Lau09]{Laurent-2009-SOS}
Monique Laurent.
\newblock Sums of squares, moment matrices and optimization over polynomials.
\newblock In {\em Emerging applications of algebraic geometry}, pages 157--270.
  Springer, 2009.

\bibitem[LCY12]{LCY-2012-AsymptoticsStatistics}
Lucien Le~Cam and Grace~Lo Yang.
\newblock {\em Asymptotics in statistics: some basic concepts}.
\newblock Springer Science \& Business Media, 2012.

\bibitem[LKZ15a]{LKZ-mmse}
Thibault Lesieur, Florent Krzakala, and Lenka Zdeborov{\'a}.
\newblock {MMSE} of probabilistic low-rank matrix estimation: Universality with
  respect to the output channel.
\newblock In {\em 2015 53rd Annual Allerton Conference on Communication,
  Control, and Computing}, pages 680--687. IEEE, 2015.

\bibitem[LKZ15b]{LKZ-sparse}
Thibault Lesieur, Florent Krzakala, and Lenka Zdeborov{\'a}.
\newblock Phase transitions in sparse {PCA}.
\newblock In {\em 2015 IEEE International Symposium on Information Theory
  (ISIT)}, pages 1635--1639. IEEE, 2015.

\bibitem[Mas14]{massoulie}
Laurent Massouli{\'e}.
\newblock Community detection thresholds and the weak ramanujan property.
\newblock In {\em Proceedings of the forty-sixth annual ACM symposium on Theory
  of computing}, pages 694--703, 2014.

\bibitem[McK81]{mckay1981expected}
Brendan~D McKay.
\newblock The expected eigenvalue distribution of a large regular graph.
\newblock {\em Linear Algebra and its Applications}, 40:203--216, 1981.

\bibitem[MM09]{MM-2009-InformationPhysicsComputation}
Marc Mezard and Andrea Montanari.
\newblock {\em Information, physics, and computation}.
\newblock Oxford University Press, 2009.

\bibitem[MNS18]{mns}
Elchanan Mossel, Joe Neeman, and Allan Sly.
\newblock A proof of the block model threshold conjecture.
\newblock {\em Combinatorica}, 38(3):665--708, 2018.

\bibitem[MRX19]{MRX-2019-SOS4}
Sidhanth Mohanty, Prasad Raghavendra, and Jeff Xu.
\newblock Lifting sum-of-squares lower bounds: Degree-2 to degree-4.
\newblock {\em arXiv preprint arXiv:1911.01411}, 2019.

\bibitem[MS16]{MS-2016-SDPSparseRandomGraphs}
Andrea Montanari and Subhabrata Sen.
\newblock Semidefinite programs on sparse random graphs and their application
  to community detection.
\newblock In {\em Proceedings of the forty-eighth annual ACM symposium on
  Theory of Computing}, pages 814--827. ACM, 2016.

\bibitem[NM14]{newman-martin}
M.~E.~J. Newman and Travis Martin.
\newblock Equitable random graphs.
\newblock {\em Phys. Rev. E}, 90:052824, 2014.

\bibitem[NN12]{Nadakuditi_12}
Raj~Rao Nadakuditi and M.~E.~J. Newman.
\newblock Graph spectra and the detectability of community structure in
  networks.
\newblock {\em Phys. Rev. Lett.}, 108:188701, May 2012.

\bibitem[NP33]{NP-1933-MostEfficientTests}
Jerzy Neyman and Egon~Sharpe Pearson.
\newblock {IX}. on the problem of the most efficient tests of statistical
  hypotheses.
\newblock {\em Philosophical Transactions of the Royal Society of London.
  Series A, Containing Papers of a Mathematical or Physical Character},
  231(694-706):289--337, 1933.

\bibitem[Pan11]{Panchenko-ultra}
Dmitry Panchenko.
\newblock The {Parisi} ultrametricity conjecture.
\newblock {\em arXiv preprint arXiv:1112.1003}, 2011.

\bibitem[Pan13]{Panchenko-SK}
Dmitry Panchenko.
\newblock {\em The Sherrington-Kirkpatrick model}.
\newblock Springer Science \& Business Media, 2013.

\bibitem[Par79]{Parisi-SK}
Giorgio Parisi.
\newblock Infinite number of order parameters for spin-glasses.
\newblock {\em Physical Review Letters}, 43(23):1754, 1979.

\bibitem[Rom05]{Roman-2005-Umbral}
Steven Roman.
\newblock {\em The umbral calculus}.
\newblock Springer, 2005.

\bibitem[Sen18]{Sen-2018-SparseRandomHypergraphs}
Subhabrata Sen.
\newblock Optimization on sparse random hypergraphs and spin glasses.
\newblock {\em Random Structures \& Algorithms}, 53(3):504--536, 2018.

\bibitem[SK75]{SK-1975-SolvableModel}
David Sherrington and Scott Kirkpatrick.
\newblock Solvable model of a spin-glass.
\newblock {\em Physical review letters}, 35(26):1792, 1975.

\bibitem[Sod07]{sodin2007random}
Sasha Sodin.
\newblock Random matrices, nonbacktracking walks, and orthogonal polynomials.
\newblock {\em Journal of Mathematical Physics}, 48(12):123503, 2007.

\bibitem[Sol96]{sole1996spectra}
Patrick Sol{\'e}.
\newblock Spectra of regular graphs and hypergraphs and orthogonal polynomials.
\newblock {\em European Journal of Combinatorics}, 17(5):461--477, 1996.

\bibitem[Sze39]{szeg1939orthogonal}
Gabor Szeg.
\newblock {\em Orthogonal polynomials}, volume~23.
\newblock American Mathematical Soc., 1939.

\bibitem[Tal06]{Talagrand-Parisi}
Michel Talagrand.
\newblock The {Parisi} formula.
\newblock {\em Annals of mathematics}, pages 221--263, 2006.

\bibitem[ZB10]{Zdeborova_10}
Lenka Zdeborov{\'{a}} and Stefan Boettcher.
\newblock A conjecture on the maximum cut and bisection width in random regular
  graphs.
\newblock {\em Journal of Statistical Mechanics: Theory and Experiment},
  2010(02):P02020, 2010.

\bibitem[ZK11]{quiet-2}
Lenka Zdeborov{\'a} and Florent Krzakala.
\newblock Quiet planting in the locked constraint satisfaction problems.
\newblock {\em SIAM Journal on Discrete Mathematics}, 25(2):750--770, 2011.

\end{thebibliography}

\appendix

\section{Low-Degree Analysis of General Wishart Models}\label{app:wishart-LDLR}

In this section we derive the formula~\eqref{eq:wishart-LDLR} for $\|L^{\le D}\|^2$ in the general spiked Wishart model.

\subsection{Hermite Polynomial Facts}

We first define and give the key facts that we will use about the Hermite polynomials, the orthogonal polynomials with respect to the standard Gaussian measure.

\begin{definition}[Hermite polynomials]
    \label{def:hermite-polynomials}
    The \emph{univariate Hermite polynomials} are the sequence of polynomials $h_k(y) \in \RR[y]$ for $k \in \NN$, defined by the recursion
    \begin{align}
        h_0(y) &= 1, \\
        h_{k + 1}(y) &= yh_k(y) -  h_{k}^{\prime}(y).
    \end{align}
    The \emph{$n$-variate Hermite polynomials} are the polynomials $H_{\alpha}(y) \in \RR[y_1, \dots, y_n]$ indexed by $\alpha \in \NN^n$ and $H_{\alpha}(y) = \prod_{i = 1}^n h_{\alpha_i}(y_i)$.
    Finally, the \emph{normalized $n$-variate Hermite polynomials} are $\widehat{H}_{\alpha}(y) = (\alpha!)^{-1/2} H_{\alpha}(y)$, where we abbreviate $\alpha! = \prod_{i = 1}^n \alpha_i!$.
\end{definition}

The main and defining property of the Hermite polynomials is their orthogonality, which we record below.
\begin{proposition}[Orthogonality under Gaussian measure]
    For any $\alpha, \beta \in \NN^n$,
    \begin{equation}
        \Ex_{\by \sim \sN(0, I)}[\widehat{H}_{\alpha}(\by) \widehat{H}_{\beta}(\by)] = \delta_{\alpha \beta}.
    \end{equation}
\end{proposition}

Beyond this, the key additional tool for our analysis is a generalization of the following fact from the ``umbral calculus'' of Hermite polynomials (a proof will be subsumed in our more general result below).
\begin{proposition}[Mismatched Variance Formula]
\label{prop:mismatched-variance-univar}
    Let $x > -1$.
    Then,
    \begin{equation}
        \Ex_{\by \sim \sN(0, 1 + x)}h_k(\by) = 
        \begin{cases} 
            (k-1)!! \cdot x^{k/2} & k \text{ even} \\
            0 & k \text{ odd}
        \end{cases}
    \end{equation}
\end{proposition}
\noindent
Note that the formula on the right-hand side is that for the moments of a Gaussian random variable with variance $x$, but we extend it to apply even for negative $x$, which is the ``umbral''  case of the result, admitting an interpretation in terms of a fictitious Gaussian of negative variance---even if $x \in (-1, 0)$, the right-hand side may be viewed formally as the value of  ``$\EE_{\bg \sim \sN(0, x)} \bg^k$.'' 
A thorough exposition of such analogies arising in combinatorics and the theory of orthogonal polynomials is given in \cite{Roman-2005-Umbral}.

In fact, the same holds even for multivariate Gaussians.
The correct result in this case is given by imitating the formula for the moments of a multivariate Gaussian, via Wick's (or Isserlis') formula.
While Proposition~\ref{prop:mismatched-variance-univar} is well-known in the literature on Hermite polynomials and the umbral calculus, we are not aware of previous appearances of the formula below.
\begin{proposition}[Multivariate Mismatched Variance Formula]
    \label{prop:multi-mismatched-var}
     Let $X \in \RR^{n \times n}_\sym$ with $X \succ -\id_n$.
     For $\alpha \in \NN^n$ viewed as a multiset of elements of $[n]$, let $\mathcal{P}(\alpha)$ be the set of pairings of elements of $\alpha$, and for each $P \in \mathcal{P}(\alpha)$ write $X^P$ denote the product of the entries of $X$ located at paired indices from $P$.
     Then,
     \begin{equation}
         \label{eq:multi-mismatched-var}
         \Ex_{\by \sim \sN( 0, I + X)} H_{\alpha}(\by) = \sum_{P \in \mathcal{P}(\alpha)} X^P.
     \end{equation}
 \end{proposition}
 \noindent
 Note that if $X \succeq 0$, then the right-hand side equals $\EE_{\bx \sim \sN( 0, X)} x^{\alpha}$ by Wick's formula, but we again have an umbral extension to non-PSD matrices $X$.
 \begin{proof}
     Define
     \begin{equation}
         \ell_{\alpha} \colonequals \Ex_{\by \sim \sN(0, I + X)} H_{\alpha}(\by).
     \end{equation}
     Let $e_i \in \NN^n$ have $i$th coordinate equal to 1 and all other coordinates equal to zero, and write $0 \in \NN^n$ for the vector with all coordinates equal to zero.
     Clearly $\ell_{0} = 1$ and $\ell_{e_i} = 0$ for any $i \in [n]$.
     We then proceed by induction and use Gaussian integration by parts:
     \begin{align}
       \ell_{\alpha + e_i}
       &= \Ex_{\by \sim \sN(0, I + X)}h_{\alpha_i + 1}(\by_i)\prod_{j \in [n] \setminus \{i\}} h_{\alpha_j}(\by_j) \nonumber \\
       &= \Ex_{\by \sim \sN(0, I + X)}\left(\by_i h_{\alpha_i}(\by_i) - h_{\alpha_i}^\prime(\by_i)\right)\prod_{j \in [n] \setminus \{i\}}  h_{\alpha_j}(\by_j) \tag{Definition~\ref{def:hermite-polynomials}} \\
       &= \Ex_{\by \sim \sN(0, I + X)}\left[\sum_{k = 1}^n(I + X)_{ik}\prod_{j \in [n]} h_{\alpha_j}^{(\delta_{jk})}(\by_j) - \prod_{j \in [n]} h_{\alpha_j}^{(\delta_{ij})}(\by_j)\right] \tag{integration by parts} \\
       &= \sum_{\substack{k \in [n] \\ \alpha_k > 0}}X_{ik}\Ex_{\by \sim \sN(0, I + X)}\prod_{j \in [n]} h_{\alpha_j}^{(\delta_{jk})}(\by_j) \nonumber \\
       &= \sum_{\substack{k \in [n] \\ \alpha_k > 0}}\alpha_k X_{ik}\Ex_{\by \sim \sN(0, I + X)}\prod_{j \in [n]} h_{\alpha_j - \delta_{jk}}(\by_j) \nonumber \\
       &= \sum_{\substack{k \in [n] \\ \alpha_k > 0}}\alpha_k X_{ik} \ell_{\alpha - e_k} \tag{inductive hypothesis}
     \end{align}
     so $\ell_{\alpha}$ satisfy the same recursion and initial condition as the sum-of-products formula on the right-hand side of \eqref{eq:multi-mismatched-var}.
 \end{proof}

\subsection{Components of the LDLR}

Let $\QQ$, $\PP$ be as in the general spiked Wishart model (Definition~\ref{def:general-wishart-model}), and let $L$ be the associated likelihood ratio. Throughout this section, we assume without loss of generality that (i) $\beta = 1$ (since $\beta$ can be absorbed into $X$), and (ii) the prior $\sX$ is supported on $X$ for which $X \succ -\id_n$. For $\alpha \in (\NN^n)^N$, we denote
\begin{align}
  m_{\alpha}
  &\colonequals \Ex_{\bY \sim \QQ} H_{\alpha}(\bY)L(\bY), \\
  \what{m}_{\alpha}
  &\colonequals \Ex_{\bY \sim \QQ} \what{H}_{\alpha}(\bY)L(\bY) = \frac{1}{\sqrt{\alpha!}}m_{\alpha}.
\end{align}

We may compute these numbers as follows.
For any $\alpha \in (\NN^n)^N$, we have, passing to an expectation under $\PP$ rather than $\QQ$ and then using Proposition~\ref{prop:multi-mismatched-var},
\begin{align}
    m_{\alpha}
      &= \Ex_{\bY \sim \PP} H_{\alpha}(\bY) \nonumber \\
      &= \Ex_{\bX \sim \sX} \prod_{i = 1}^N \Ex_{\by \sim \sN(0, I + \bX)} H_{\alpha_i}(y) \nonumber \\
      &= \Ex_{\bX \sim \sX} \prod_{i = 1}^N\left(\sum_{P \in \sP(\alpha_i)} \bX^P\right).
\end{align}

\subsection{Taylor Expansion of the LDLR}

\begin{lemma}
    Suppose $\sX$ is as in Definition~\ref{def:general-wishart-model}. 
    Denote by $\mathcal{T}^{\leq D}$ the operation of truncating the Taylor series of a function that is real-analytic in a neighborhood of zero to degree-$D$ polynomials.
    For the sake of clarity, this will only ever apply to the variable named $t$.
    Then,
    \begin{equation}
        \label{eq:general-wishart-ldlr-taylor}
        \|L^{\leq D}\|^2 = \Ex_{\bX, \bX^{\prime} \sim \sX}\mathcal{T}^{\leq D}\left[\det(\id_n - t^2\bX \bX^{\prime})^{-N / 2}\right](1),
    \end{equation}
    by which we mean evaluation of the function of $t$ on the RHS at $t=1$.
\end{lemma}
\begin{proof}
    We have
    \begin{align}
        \|L^{\leq D}\|^2 
        &= \sum_{|\alpha| \leq D}\what{m}_{\alpha}^2 \nonumber \\
        &= \Ex_{\bX, \bX^{\prime} \sim \sX} \sum_{|\alpha| \leq D} \prod_{i = 1}^N\frac{1}{ \alpha_i!}\left(\sum_{P \in \mathcal{P}(\alpha_i)} \bX^P\right)\left(\sum_{P \in \mathcal{P}(\alpha_i)} (\bX^{\prime})^P\right) \nonumber \\
        &= \Ex_{\bX, \bX^{\prime} \sim \sX}\sum_{d = 0}^D \sum_{|\alpha| = d}\prod_{i = 1}^N\frac{1}{ \alpha_i!}\left(\sum_{P \in \mathcal{P}(\alpha_i)} \bX^P\right)\left(\sum_{P \in \mathcal{P}(\alpha_i)} (\bX^{\prime})^P\right).
    \end{align}
    On the other hand, we may expand the right-hand side of \eqref{eq:general-wishart-ldlr-taylor} by repeatedly differentiating with respect to $t$ to extract Taylor coefficients.
    In doing so we will repeatedly apply the chain rule, and since $\frac{d}{dt}\det(I - tA) = \Tr(A)$ each derivative will be a rational function in $(t, X, X^{\prime})$. Therefore, the Taylor coefficients are some rational functions $r_d(X, X^{\prime})$ (depending on $N$) for $d \in \NN$, and for these coefficients
    \begin{equation}\label{eq:wish-taylor}
        \Ex_{\bX, \bX^{\prime} \sim \sX}\mathcal{T}^{\leq D}\left[\det(\id_n - t^2\bX \bX^{\prime})^{-N / 2}\right](1) = \Ex_{\bX, \bX^{\prime} \sim \sX} \sum_{d = 0}^D r_d(\bX, \bX^{\prime}).
    \end{equation}
    We will show that in fact termwise equality holds, \emph{inside the expectations}, namely
    \begin{equation} 
        \label{eq:general-wishart-ldlr-coeff}
        r_d(X, X^{\prime}) = \sum_{|\alpha| = d}\prod_{i = 1}^N\frac{1}{ \alpha_i!}\left(\sum_{P \in \mathcal{P}(\alpha_i)} X^P\right)\left(\sum_{P \in \mathcal{P}(\alpha_i)} (X^{\prime})^P\right) \equalscolon r_d^\prime(X, X^{\prime})
    \end{equation}
    for each $d \in \NN$ and for all (deterministic) $X, X^{\prime} \in \RR^{n \times n}_{\sym}$.
    Since either side of \eqref{eq:general-wishart-ldlr-coeff} is a rational function of $(X, X^{\prime})$, it suffices to show that this is true on a set of matrices of positive measure.
    We will show that it holds for all $X, X^{\prime} \succeq 0$.

    In this case, we have the convenient Gaussian interpretation of the expression for $r_d^\prime(X, X^{\prime})$ from Wick's formula:
    \begin{align}
        r_d^\prime(X, X^{\prime}) 
        &= \sum_{|\alpha| = d}\prod_{i = 1}^N\frac{1}{ \alpha_i!} \left(\Ex_{\bx \sim \sN( 0, X)} \bx^{\alpha_i}\right)\left(\Ex_{\bx \sim \sN( 0, X^{\prime})} \bx^{\alpha_i}\right) \nonumber \\
        &= \Ex_{\substack{\bx_1, \dots, \bx_N \sim \sN( 0, X) \\ \bx_1^{\prime}, \dots, \bx_N^{\prime} \sim \sN( 0, X^{\prime})}}\sum_{\substack{|\alpha_i| \text{ even} \\ |\alpha| = d}}\prod_{i = 1}^N\frac{1}{ \alpha_i!}
        (\bx_i)^{\alpha_i}(\bx_i^{\prime})^{\alpha_i}
        \intertext{and grouping by the values of $|\alpha_i|$,}
        &= \Ex_{\substack{\bx_1, \dots, \bx_N \sim \sN( 0, X) \\ \bx_1^{\prime}, \dots, \bx_N^{\prime} \sim \sN( 0, X^{\prime})}} \sum_{\substack{d_1, \dots, d_N \in 2\NN \\ \sum_{i = 1}^N d_i = d}} \frac{1}{\prod_{i = 1}^N d_i!} \sum_{\substack{ \alpha_1, \dots, \alpha_N \in \NN^n \\ |\alpha_i| = d_i}} \prod_{i = 1}^N \binom{d_i}{\alpha_i} \prod_{j = 1}^n ((\bx_i)_j(\bx_i^{\prime})_j)^{\alpha_i(j)} \nonumber \\
        &= \Ex_{\substack{\bx_1, \dots, \bx_N \sim \sN( 0, X) \\ \bx_1^{\prime}, \dots, \bx_N^{\prime} \sim \sN( 0, X^{\prime})}} \sum_{\substack{d_1, \dots, d_N \in 2\NN \\ \sum_{i = 1}^N d_i = d}} \prod_{i = 1}^N \frac{\la \bx_i, \bx_i^{\prime} \ra^{d_i}}{d_i!} \nonumber \\
        &= \sum_{\substack{d_1, \dots, d_N \in 2\NN \\ \sum_{i = 1}^N d_i = d}} \prod_{i = 1}^N \frac{1}{d_i!} \Ex_{\substack{\bx \sim \sN( 0, X) \\ \bx^{\prime} \sim \sN( 0, X^{\prime})}} \la \bx, \bx^{\prime} \ra^{d_i}.
    \end{align}
    From here, we introduce the moment-generating function of the inner overlap variables.
    Define
    \begin{equation}
        \phi_{X, X^{\prime}}(t) 
        \colonequals \Ex_{\substack{\bx \sim \sN( 0, X) \\ \bx^{\prime} \sim \sN( 0, X^{\prime})}} \exp\left(t\la \bx, \bx^{\prime} \ra\right) = \sum_{d = 0}^\infty \frac{t^d}{d!}\Ex_{\substack{\bx \sim \sN( 0, X) \\ \bx^{\prime} \sim \sN( 0, X^{\prime})}} \la \bx, \bx^{\prime} \ra^d.
    \end{equation}
    Then, $r^\prime_d(X, X^{\prime})$ is simply the coefficient of $t^d$ in the Taylor series of $\phi_{X, X^{\prime}}(t)^N$.
    
    On the other hand, we may actually compute $\phi_{X,X^{\prime}}(t)$:
    \begin{align}
        \phi_{X,X^{\prime}}(t)
        &= \Ex_{\bg, \bh \sim \sN( 0,  \id_n)} \exp\left(t\bg^\top \sqrt{X} \sqrt{X^{\prime}} \bh \right) \nonumber \\
        &= \Ex_{\bg \sim \sN( 0,  \id_{2n})}\exp\left(\bg^\top \left[\begin{array}{cc}  0 & \frac{t}{2} \sqrt{X}\sqrt{X^{\prime}} \\ \frac{t}{2}\sqrt{X^{\prime}}\sqrt{X} &  0 \end{array}\right] \bg\right) \nonumber
        \intertext{which may be calculated as a moment generating function of the ``matrix $\chi^2$'' variable $gg^\top$, giving}
        &= \det\left(\left[\begin{array}{cc}  \id_n & -t\sqrt{X^{\prime}}\sqrt{X^{\prime}} \\ -t\sqrt{X^{\prime}}\sqrt{X} &  \id_n \end{array}\right]\right)^{-1/2} \nonumber \\
        &= \det\left( \id_n - t^2\sqrt{X}X^{\prime}\sqrt{X}\right)^{-1/2} \nonumber 
        \intertext{and applying Proposition~\ref{prop:det-cycle},}
        &= \det\left( \id_n - t^2XX^{\prime}\right)^{-1/2}. \label{eq:inner-mgf-formula}
    \end{align}
    Thus $r^\prime_d(X, X^{\prime})$ is also the coefficient of $t^d$ in the Taylor series of $\det(\id_n - t^2XX^{\prime})^{-N/2}$, whereby $r^\prime_d(X, X^{\prime}) = r_d(X, X^{\prime})$, completing the proof.
\end{proof}

\noindent Implicit in the above proof are the following facts which we have made use of in Section~\ref{sec:wishart}.

\begin{proposition}\label{prop:wish-facts}
    Consider the general spiked Wishart model (Definition~\ref{def:general-wishart-model}), and assume (without loss of generality) that $\beta = 1$ and $\sX$ is supported on $X$ for which $X \succ -\id_n$. The following formulas hold:
    \begin{enumerate}
        \item[(a)] $\displaystyle \|L^{\le D}\|^2 = \Ex_{\bX,\bX^{\prime} \sim \sX} \sum_{d=0}^D r_d(\bX,\bX^{\prime})$ for some polynomials $r_0,\dots, r_d, \dots$, 
        \item[(b)] $\displaystyle \sum_{d=0}^\infty r_d(X,X^{\prime}) = \det(\id_n - X X^{\prime})^{-N/2}$,
        \item[(c)] if $X \succeq 0$ and $X^{\prime} \succeq 0$ then $\displaystyle r_d(X,X^{\prime}) = \sum_{\substack{d_1, \dots, d_N \in 2\NN \\ \sum_{i = 1}^N d_i = d}} \prod_{i = 1}^N \frac{1}{d_i!} \Ex_{\substack{\bx \sim \sN( 0, X) \\ \bx^{\prime} \sim \sN( 0, X^{\prime})}} \la \bx, \bx^{\prime} \ra^{d_i}$, and
        \item[(d)] if $X \preceq 0$ and $X^{\prime} \preceq 0$ then $\displaystyle r_d(X,X^{\prime}) = \sum_{\substack{d_1, \dots, d_N \in 2\NN \\ \sum_{i = 1}^N d_i = d}} \prod_{i = 1}^N \frac{1}{d_i!} \Ex_{\substack{\bx \sim \sN( 0, -X) \\ \bx^{\prime} \sim \sN( 0, -X^{\prime})}} \la \bx, \bx^{\prime} \ra^{d_i}$.
    \end{enumerate}
\end{proposition}

\section{Exponential Low-Degree Hardness for SBM}
\label{app:low-deg-sbm}

The goal of this section is to prove Theorem~\ref{thm:sbm}. We first prove a general statement that reduces binary-valued models to the analogous Gaussian model.

\subsection{Comparing Binary-Valued Models to Gaussian}

\begin{proposition}
    \label{prop:binary-compare}
    Consider the following general binary-valued problem.
    \begin{itemize}
        \item Under the null distribution $\QQ$, we observe $\bY \in \RR^N$ where $\bY_i$ are independent, satisfy $\EE[\bY_i] = 0$ and $\EE[\bY_i^2] = 1$, and each take two possible values: $\bY_i \in \{a_i,b_i\}$ with $a_i < b_i$.
        \item Under the planted distribution $\PP$, a signal $\bX \in \RR^N$ is drawn from some prior, and then $\bY_i \in \{a_i,b_i\}$ are drawn independently (conditioned on $\bX$) such that $\EE[\bY_i | \bX_i] = \bX_i$. (This requires $\bX_i \in [a_i,b_i]$.)
    \end{itemize}    
    In the above setting,
    \begin{equation}
        \label{eq:L-gauss}
        \|L^{\le D}\|^2 \le \sum_{d=0}^D \frac{1}{d!} \Ex_{\bX,\bX'} \langle \bX,\bX' \rangle^d
    \end{equation}
    where $\bX'$ denotes an independent copy of $\bX$.
\end{proposition}

\noindent The significance of~\eqref{eq:L-gauss} is that the right-hand side is the exact formula for $\|L^{\le D}\|^2$ in the following additive Gaussian model (see \cite{lowdeg-notes}): under $\QQ$, $\bY \sim \mathcal{N}(0,\id_N)$; and under $\PP$, $\bY = \bX + \bZ$ with $\bZ \sim \mathcal{N}(0,\id_N)$ and $\bX$ drawn from some prior. Thus, Proposition~\ref{prop:binary-compare} can be interpreted as saying that a binary-valued problem is at least as hard as the corresponding Gaussian problem.

\begin{proof}[Proof of Proposition~\ref{prop:binary-compare}]
    The Fourier characters $\chi_S(Y) = \prod_{i \in S} Y_i$ for $S \subseteq [N]$ with $|S| \le D$ are orthonormal with respect to $\QQ$, in the sense that $\EE_{\bY \sim \QQ} [\chi_S(\bY) \chi_T(\bY)] = \indicator{S = T}$. They also span the subspace of degree $\le D$ polynomials, since for any $r \in \NN$, any $Y_i^r$ can be written as a degree-1 polynomial in $Y_i$. Thus, $\{\chi_S\}_{|S| \le D}$ is an orthonormal basis for the degree $\le D$ polynomials. It is a standard fact that this allows us to write $\|L^{\le D}\|^2 = \sum_{|S| \le D}(\EE_{\bY \sim \PP}[\chi_S(\bY)])^2$ (see e.g., \cite{HS-bayesian,sam-thesis}). We will use $S$ to denote a subset of $[N]$ and use $\alpha$ to denote an ordered multi-set of $[N]$, with $\chi_\alpha(Y) \colonequals \prod_{i \in \alpha} Y_i$. Now compute
    \begin{align*}
        \|L^{\le D}\|^2 &= \sum_{|S| \le D} (\Ex_{\bY \sim \PP} \chi_S(\bY))^2 \\
        &= \sum_{|S| \le D} (\Ex_{\bX} \chi_S(\bX))^2 \\
        &\overset{(\ast)}\le \sum_{d=0}^D \sum_{|\alpha| = d} \frac{1}{d!} (\Ex_{\bX} \chi_\alpha(\bX))^2 \qquad \tag{see below} \\
        &= \sum_{d = 0}^D \sum_{|\alpha| = d} \frac{1}{d!} \Ex_{\bX,\bX'} \chi_\alpha(\bX) \chi_\alpha(\bX') \\
        &= \sum_{d = 0}^D \frac{1}{d!} \Ex_{\bX,\bX'} \langle \bX,\bX' \rangle^d.
    \end{align*}
    To see that inequality $(\ast)$ holds, note that it would be an \emph{equality} if restricted to ordered multi-sets $\alpha$ that contain distinct elements.
\end{proof}

\subsection{Stochastic Block Model}

We now specialize to the case of the stochastic block model (Definition~\ref{def:sbm}), and use Proposition~\ref{prop:binary-compare} to reduce to a certain spiked Wigner model.

\begin{proof}[Proof of Theorem~\ref{thm:sbm}]
    It will be convenient to consider a modification of the SBM that allows self-loops. Specifically, the edge $(i,i)$ occurs with probability $\frac{d}{n}$ under $\QQ$ and with probability $(1+\frac{\eta}{\sqrt{2}}(k-1))\frac{d}{n}$ under $\PP$. It is clear from the variational formula~\eqref{eq:ldlr-var} for $\|L^{\le D}\|$ that revealing this extra information can only increase $\|L^{\le D}\|$.
    
    In order to place the SBM in the setting of Proposition~\ref{prop:binary-compare}, take $N = n(n+1)/2$ with a variable $\bY_{i,j}$ for every $i \le j$. Let $p = d/n$. In order to ensure $\EE_\QQ[\bY_{i,j}] = 0$ and $\EE_\QQ[\bY_{i,j}^2] = 1$, take $\bY_{i,j} = b \colonequals \sqrt{(1-p)/p}$ if edge $(i,j)$ is present, and $\bY_{i,j} = a \colonequals -\sqrt{p/(1-p)}$ otherwise.
    
    We now define $\bX$ appropriately. Conditioned on the community structure, edge $(i,j)$ occurs with probability $(1+\bDelta_{i,j})p$ where $\bDelta_{i,i} = \eta(k-1)/\sqrt{2}$ and $\bDelta_{i,j}$ (for $i < j$) is either $\eta(k-1)$ or $-\eta$ depending on whether $i$ and $j$ belong to the same community or not, respectively. Using the fact $pb + (1-p)a = 0$ (twice), this means we should take
    \[ 
        \bX_{i,j} = \EE[\bY_{i,j} | \bX_{i,j}] = (1+\bDelta_{i,j})pb + [1-(1+\bDelta_{i,j})p]a = \bDelta_{i,j} p (b-a) = -\bDelta_{i,j} a = \bDelta_{i,j} \sqrt{\frac{p}{1-p}}. 
    \]
    Let $\bU$ be the $n \times k$ matrix whose $i$th row is $\sqrt{k} e_{\bk_i} - \onesvec/\sqrt{k}$ where $\bk_i \in [k]$ is the community assignment of vertex $i$. One can check that $(\bU\bU^\top)_{i,j} = k \indicator{\bk_i = \bk_j} - 1$ and so $\bX_{i,j} = \eta \sqrt{\frac{p}{1-p}} (\bU\bU^\top)_{ij}$ for $i < j$, and $\bX_{i,i} = \frac{\eta}{\sqrt{2}} \sqrt{\frac{p}{1-p}} (\bU\bU^\top)_{i,i}$. Therefore $\langle \bX,\bX' \rangle = \frac{\eta^2}{2} \frac{p}{1-p} \langle \bU\bU^\top, \bU'(\bU')^\top \rangle$. By Proposition~\ref{prop:binary-compare},
    \[ 
        \|L^{\le D}\|^2 \le \sum_{d=0}^D \frac{1}{d!} \EE \left(\frac{\eta^2}{2} \frac{p}{1-p} \langle \bU\bU^\top, \bU'(\bU')^\top \rangle\right)^d. 
    \]
    Comparing this to~\eqref{eq:L-wig} reveals that this is precisely the expression for $\|L^{\le D}\|^2$ in the general spiked Wigner model with spike prior $\sX_k$ (see Definition~\ref{def:pi-k}), except in place of $\lambda^2$ we have $\frac{\eta^2 d}{1-p} = (1+o(1))\eta^2 d$. Therefore, appealing to the Wigner result (Theorem~\ref{thm:general-wigner}), we have $\|L^{\le D}\| = O(1)$ provided $d \eta^2 < 1$.
\end{proof}

\end{document}